\documentclass[sigplan,screen,authorversion]{acmart}
\pdfoutput=1

\setcopyright{rightsretained}
\acmPrice{}
\acmDOI{10.1145/3519939.3523721}
\acmYear{2022}
\copyrightyear{2022}
\acmSubmissionID{pldi22main-p504-p}
\acmISBN{978-1-4503-9265-5/22/06}
\acmConference[PLDI '22]{Proceedings of the 43rd ACM SIGPLAN International Conference on Programming Language Design and Implementation}{June 13--17, 2022}{San Diego, CA, USA}
\acmBooktitle{Proceedings of the 43rd ACM SIGPLAN International Conference on Programming Language Design and Implementation (PLDI '22), June 13--17, 2022, San Diego, CA, USA}

\usepackage{booktabs}   
\usepackage{subcaption} 

\usepackage{tcolorbox} 
\usepackage{bussproofs} 
\usepackage{stmaryrd} 
\usepackage{thm-restate} 
\usepackage[capitalise]{cleveref} 
\usepackage{amsthm} 
\usepackage{algorithm}
\usepackage[noend]{algpseudocode}

\usepackage{stackengine}

\usepackage{color, colortbl}

\usepackage{tikz}
\usepackage{amsmath}
\usepackage{stmaryrd}
\usepackage{algorithmicx}
\usepackage{booktabs}
\usepackage{colortbl}

\usepackage{array}
\usepackage{arydshln}
\usepackage{doi}

\usepackage{microtype}
\usepackage{balance}

\newif\iffullversion
\fullversiontrue

\newcommand{\ifFull}[1]{\iffullversion {#1}\else {the full version \cite{full}}\fi}

\usepackage{paralist}

\theoremstyle{definition}
\newtheorem{definition}{Definition}[section]
\newtheorem{example}[definition]{Example}

\theoremstyle{plain}
\newtheorem{theorem}{Theorem}[section]
\newtheorem{lemma}[theorem]{Lemma}
\newtheorem{corollary}[theorem]{Corollary}

\newcommand{\tr}{\boldsymbol{s}}
\newcommand{\ivtr}{\boldsymbol{t}}
\newcommand{\boundbox}{{\boldsymbol{t}}}

\newcommand{\toSym}{\leadsto} 
\newcommand{\refines}{\mathrel{\triangleleft}}

\newcommand{\traces}{\mathbb{T}}
\newcommand{\traceset}{\mathcal{T}}
\newcommand{\tracesin}[1]{\llparenthesis #1 \rrparenthesis}

\newcommand{\valueSem}[1]{\mathsf{val}_#1} 
\newcommand{\weightSem}[1]{\mathsf{wt}_#1} 
\newcommand{\measureSem}[1]{\llbracket #1 \rrbracket}

\newcommand{\posterior}{\mathsf{posterior}}

\newcommand{\Sat}{\mathsf{Sat}}

\newcommand{\exType}[2]{\left\{\begin{matrix}
	#1\\
	#2
\end{matrix} \right\}}

\newcommand{\exTypeS}[2]{\scalebox{0.8}{$\left\{\begin{matrix}
		#1\\
		#2
	\end{matrix} \right\}$}}

\newcommand{\iv}{\mathbb{I}}
\newcommand{\lowerBound}[1]{\mathsf{lowerBd}_{#1}}
\newcommand{\upperBound}[1]{\mathsf{upperBd}_{#1}}

\newcommand{\lit}[1]{\underline{#1}} 
\newcommand{\letIn}[1]{\mathsf{let} \,{#1}\, \allowbreak \mathsf{in}\,}
\newcommand{\fixLam}[2]{\mu^{#1}_{#2} \ldotp} 
\newcommand{\ifCond}[3]{\mathsf{if} (#1) \, \allowbreak {#2} \, \allowbreak \mathsf{else}\, {#3}}
\newcommand{\sample}{\mathsf{sample}}
\newcommand{\score}{\mathsf{score}}
\newcommand{\ifSimple}[3]{\mathsf{if}(#1, #2, #3)}
\newcommand{\ifElse}[3]{\ifSimple{#1}{#2}{#3}}

\newcommand{\typeReal}{\textbf{\textsf{R}}}

\newcommand\defn[1]{{\bf \em #1}}

\newcommand{\RR}{\ensuremath{\mathbb R}}
\newcommand{\NN}{\ensuremath{\mathbb N}}

\newcommand{\D}{\mathop{}\!\mathrm{d}} 

\newcommand{\Normal}{\ensuremath{\mathrm{Normal}}}

\DeclareMathOperator{\pdf}{pdf}

\DeclareMathOperator{\volume}{vol}

\newcommand{\tightsubset}{\Subset}

\def\metalambda{\mathop{\scalerel*{\stackengine{1.9pt}{$\lambda$}{%
				\kern3.4pt\smash{\clipbox{2pt -.5pt 0pt -.5pt}{$\lambda$}}}{O}{l}{F}{F}{L}}{X}\mkern1mu}}

\newcommand{\real}{\mathbb{R}}

\newcommand{\natnum}{\mathbb{N}}

\newcommand{\posreal}{\mathbb{R}_{\geq 0}}

\newcommand{\calM}{\mathcal{M}}
\newcommand{\calN}{\mathcal{N}}
\newcommand{\calP}{\mathcal{P}}

\newcommand{\calV}{\mathcal{V}}
\newcommand{\calW}{\mathcal{W}}
\newcommand{\calZ}{\mathcal{Z}}
\newcommand{\calE}{\mathcal{E}}
\newcommand{\calR}{\mathcal{R}}

\newcommand{\calA}{\mathcal{A}}
\newcommand{\calB}{\mathcal{B}}
\newcommand{\calC}{\mathcal{C}}

\definecolor{intcolor}{RGB}{30,76,135}

\newcommand{\myint}[1]{ [#1]}

\newcommand{\stdConf}[1]{(#1)}
\newcommand{\intConf}[1]{(#1)}
\newcommand\symPath[1]{(#1)}
\newcommand\symConf[1]{(#1)}
\newcommand{\rmFix}{\mathit{approxFix}}

\begin{document}

\title[{Guaranteed Bounds for Posterior Inference in Universal Probabilistic Programming}]{Guaranteed Bounds for Posterior Inference\\ in Universal Probabilistic Programming}

\author{Raven Beutner}
\authornote{Member of the Saarbrücken Graduate School of Computer Science.}          
\orcid{0000-0001-6234-5651}             
\affiliation{
	\institution{CISPA Helmholtz Center for Information Security}            
	\country{Germany}                    
}

\author{C.-H.~Luke Ong}
\orcid{0000-0001-7509-680X}             
\affiliation{
	\institution{University of Oxford}            
	\country{United Kingdom}                    
}

\author{Fabian Zaiser}
\orcid{0000-0001-5158-2002}             
\affiliation{
	\institution{University of Oxford}            
	\country{United Kingdom}                    
}

\begin{abstract}
We propose a new method to approximate the posterior distribution of probabilistic programs by means of computing \emph{guaranteed} bounds.
The starting point of our work is an interval-based trace semantics for a recursive, higher-order probabilistic programming language with continuous distributions.
Taking the form of (super-/subadditive) measures, these lower/upper bounds are non-stochastic and provably correct:
using the semantics, we prove that the actual posterior of a given program is sandwiched between the lower and upper bounds (soundness);
moreover, the bounds converge to the posterior (completeness).
As a practical and sound approximation, we introduce a weight-aware interval type system, which automatically infers interval bounds on not just the return value but also the weight of program executions, simultaneously.
We have built a tool implementation, called GuBPI, which automatically computes these posterior lower/upper bounds.
Our evaluation on examples from the literature shows that the bounds are useful, and can even be used to recognise wrong outputs from stochastic posterior inference procedures.
\end{abstract}



\maketitle


\section{Introduction}

Probabilistic programming is a rapidly developing discipline at the interface of programming and Bayesian statistics \cite{GordonHNR14,Goodman2014,VandeMeent2018}.
The idea is to express probabilistic models (incorporating the prior distributions) and the observed data as programs,
and to use a general-purpose Bayesian inference engine, which acts directly on these programs, to find the posterior distribution given the observations.

Some of the most influential probabilistic programming languages (PPLs) used in practice are \emph{universal} (i.e.~the underlying language is Turing-complete); e.g.~Church \cite{GoodmanMRBT08}, Anglican \cite{TolpinMW15}, Gen \cite{Cusumano-Towner19}, Pyro \cite{pyro19}, and Turing \cite{ge2018Turing}.
Using stochastic branching, recursion, and higher-order features, universal PPLs can express arbitrarily complex models.
For instance, these language constructs can be used to incorporate probabilistic context free grammars \cite{Manning99}, statistical phylogenetics \cite{Ronquist2020.06.16.154443}, and even physics simulations \cite{BaydinSBHM0MNGL19} into probabilistic models.
However, expressivity of the PPL comes at the cost of complicating the posterior inference.
Consider, for example, the following problem from \cite{MakOPW21,MakZO21}.

\begin{example}[Pedestrian]
    \label{ex:pedestrian}

    A pedestrian has gotten lost on a long road and only knows that they are a random distance between 0 and 3 km from their home. They repeatedly walk a uniform random distance of at most 1 km in either direction, until they find their home. When they arrive, a step counter tells them that they have traveled a distance of 1.1 km in total.
    Assuming that the measured distance is normally distributed around the true distance with standard deviation 0.1 km, what is the posterior distribution of the starting point?
    We can model this with a probabilistic program:

	\vspace{2mm}
    \algdef{SE}[SUBALG]{Indent}{EndIndent}{}{\algorithmicend\ }%
    \algtext*{Indent}
    \algtext*{EndIndent}

	\parbox{\linewidth-10pt}{
		\begin{algorithmic}
			\State \textbf{let} $\mathit{start}$ = $3 \times {}$\textbf{sample} uniform$(0, 1)$\textbf{ in}
			\State \textbf{letrec} $\mathit{walk}$ $x$ = \textbf{if} $x \leq 0$ \textbf{then} 0 \textbf{else}
			\Indent
			\State \textbf{let} $\mathit{step}$ = \textbf{sample} uniform$(0, 1)$\textbf{ in}
			\State $\mathit{step} + \mathit{walk}\big((x + \mathit{step}) \oplus_{0.5} (x - \mathit{step})\big)$
			\EndIndent
			\State \textbf{let} $\mathit{distance} = \mathit{walk}$ $\mathit{start}$\textbf{ in}
			\State \textbf{observe} $\mathit{distance}$  \textbf{from} $\Normal(1.1, 0.1)$;
			\State $\mathit{start}$
		\end{algorithmic}
	}

	\vspace{2mm}

    \noindent
    Here \textbf{sample} \!\! uniform$(a, b)$ samples a uniformly distributed value in $[a, b]$, $\oplus_{0.5}$ is probabilistic branching, and \textbf{observe} $M$ \textbf{from} $D$ observes the value of $M$ from distribution $D$.
\end{example}

\cref{ex:pedestrian} is a challenging model for inference algorithms in several regards:
not only does the program use stochastic branching and recursion, but the number of random variables generated is unbounded -- it's \emph{nonparametric} \cite{HjortHMW10,Ghahramani2013,MakZO21}.
To approximate the posterior distribution of the program, we apply two standard inference algorithms:
likelihood-weighted importance sampling (IS), a simple algorithm that works well on low-dimensional models with few observations \cite{mcbook}; and Hamiltonian Monte Carlo (HMC) \cite{DuaneKPR87}, a successful MCMC algorithm that uses gradient information to efficiently explore the parameter space of high-dimensional models.
\Cref{fig:pedestrian-stochastic} shows the results of the two inference methods as implemented in Anglican \cite{TolpinMW15} (for IS) and Pyro \cite{pyro19} (for HMC): they clearly disagree!
But how is the user supposed to know which (if any) of the two results is correct?

Note that \emph{exact} inference methods (i.e.~methods that try to compute a closed-form solution of the posterior inference problem using computer algebra and other forms of symbolic computation) such as PSI \cite{GehrMV16,GehrSV20}, Hakaru \cite{NarayananCRSZ16}, Dice \cite{HoltzenBM20}, and SPPL \cite{SaadRM21} are only applicable to non-recursive models, and so they don't work for \cref{ex:pedestrian}.

\begin{figure}
    \centering
    \includegraphics[width=\columnwidth]{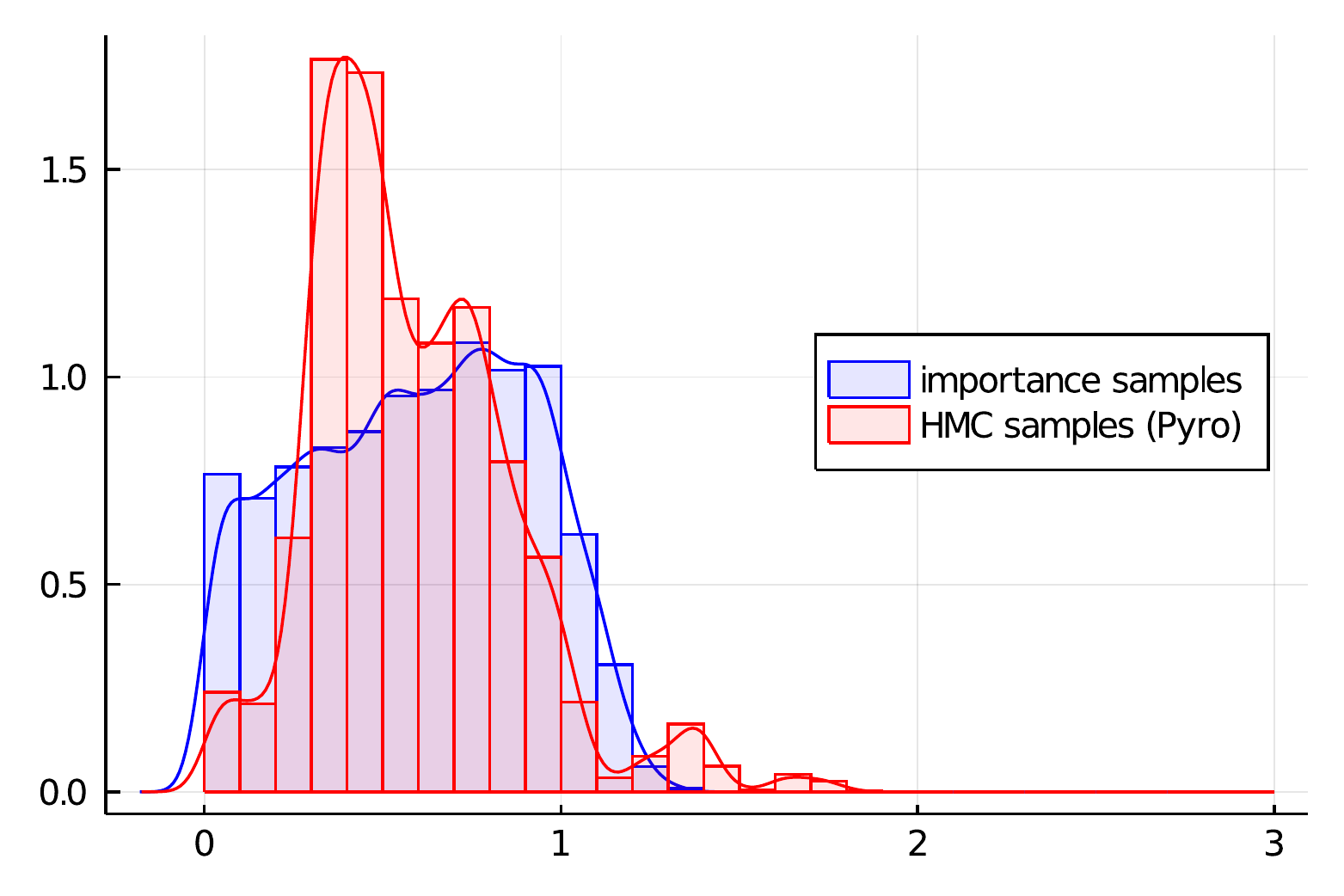}
    \vspace{-8mm}
    \caption{Histogram of samples from the posterior distribution of \cref{ex:pedestrian} and wrong samples produced by the probabilistic programming system Pyro.}
    \label{fig:pedestrian-stochastic}
\end{figure}

\subsection{Guaranteed Bounds}

The above example illustrates central problems with both approximate stochastic and exact inference methods.
For approximate methods, there are no guarantees for the results they output after a finite amount of time, leading to unclear inference results (as seen in \cref{fig:pedestrian-stochastic}).\footnote{Take MCMC sampling algorithms.
Even though the Markov chain will eventually converge to the target distribution,
we do not know how long to iterate the chain to ensure convergence \cite{Roy2020,mcbook}.
Likewise for variational inference \cite{Zhang2019}:
given a variational family, there is no guarantee that a given value for the KL-divergence (from the approximating to the posterior distribution) is attainable by the minimising distribution.}
For exact methods, the symbolic engine may fail to find a closed-form description of the posterior distribution and, more importantly, they are only applicable to very restricted classes of programs (most notably, non-recursive models).

Instead of computing approximate or exact results, this work is concerned with computing \emph{guaranteed} bounds on the posterior distribution of a probabilistic program.
Concretely, given a probabilistic program $P$ and a measurable set $U \subseteq \real$ (given as an interval), we infer upper and lower bounds on $\measureSem{P}(U)$ (formally defined in \cref{sec:2background}), i.e.~the posterior probability of $P$ on $U$.%
\footnote{By repeated application of our method on a discretisation of the domain we can compute histogram-like bounds.}
Such bounds provide a ground truth to compare approximate inference results with: if the approximate results violate the bounds, the inference algorithm has not converged yet or is even ill-suited to the program in question.
Crucially, our method is applicable to arbitrary (and in particular recursive) programs of a universal PPL.
For \cref{ex:pedestrian}, the bounds computed by our method (which we give in \cref{sec:7practical-evaluation}) are tight enough to separate the IS and HMC output.
In this case, our method infers that the results given by HMC are wrong (i.e.~violate the guaranteed bounds) whereas the IS results are plausible (i.e.~lie within the guaranteed bounds).
To the best of our knowledge, no existing methods can provide such definite answers for programs of a universal PPL.

\subsection{Contributions}

The starting point of our work is an interval-based operational semantics \cite{BeutnerO21}.
In our semantics, we evaluate a program on \emph{interval traces} (i.e.~sequences of intervals of reals with endpoints between 0 and 1) to approximate the outcomes of sampling, and use interval arithmetic \cite{Dawood2011} to approximate numerical operations (\cref{sec:3intervals}).
Our semantics is sound in the sense that any (compatible and exhaustive) set of interval traces yields lower and upper bounds on the posterior distribution of a program.
These lower/upper bounds are themselves super-/subadditive measures.
Moreover, under mild conditions (mostly restrictions on primitive operations), our semantics is also complete, i.e.~for any $\epsilon > 0$ there exists a countable set of interval traces that provides $\epsilon$-tight bounds on the posterior.
Our proofs hinge on a combination of stochastic symbolic execution and the convergence of Riemann sums, providing a natural correspondence between our interval trace semantics and the theory of (Riemann) integration (\cref{sec:4intervals-theory}).

Based on our interval trace semantics, we present a practical algorithm to automate the computation of guaranteed bounds.
It employs an interval type system (together with constraint-based type inference) that bounds both the value of an expression in a refinement-type fashion \emph{and} the score weight of any evaluation thereof.
The (interval) bounds inferred by our type system fit naturally in the domain of our semantics.
This enables a sound approximation of the behaviour of a program with finitely many interval traces (\cref{sec:5interval-analysis}).

We implemented our approach in a tool called GuBPI\footnote{GuBPI (pronounced ``guppy'') is available at \href{https://gubpi-tool.github.io/}{gubpi-tool.github.io}.} (\textbf{Gu}aranteed \textbf{B}ounds for \textbf{P}osterior \textbf{I}nference), described in \cref{sec:6linear}, and evaluate it  on a suite of benchmark programs from the literature.
We find that the bounds computed by GuBPI are competitive in many cases where the posterior could already be inferred exactly.
Moreover, GuBPI's bounds are useful (in the sense that they are precise enough to rule out erroneous approximate results as in \cref{fig:pedestrian-stochastic}, for instance) for recursive models that could not be handled rigorously by any method before (\cref{sec:7practical-evaluation}).

\subsection{Scope and Limitations}

The contributions of this paper are of both theoretical and practical interest.
On the theoretical side, our novel semantics underpins a sound and deterministic method to compute guaranteed bounds on program denotations.
As shown by our completeness theorem, this analysis is applicable---in the sense that it computes arbitrarily tight bounds---to a very broad class of programs.
On the practical side, our analyser GuBPI implements (an optimised version of) our semantics.
As is usual for exact/guaranteed%
\footnote{By ``exact/guaranteed methods'', we mean inference algorithms that compute deterministic (non-stochastic) results about the mathematical denotation of a program. In particular, they are correct with probability 1, contrary to stochastic methods.}
methods, our semantics considers an exponential number of program paths, and partitions each sampled value into a finite number of interval approximations.
Consequently, GuBPI generally struggles with high-dimensional models.
We believe GuBPI to be most useful for unit-testing of implementations of Bayesian inference algorithms such as \Cref{ex:pedestrian}, or to compute results on (recursive) programs when non-stochastic, guaranteed bounds are needed.

\section{Background}
\label{sec:2background}

\subsection{Basic Probability Theory and Notation}

We assume familiarity with basic probability theory, and refer to \cite{Pollard2002} for details.
Here we just fix the notation.
A \emph{measurable space} is a pair $(\Omega, \Sigma_\Omega)$ where $\Omega$ is a set (of outcomes) and $\Sigma_\Omega \subseteq 2^\Omega$ is a $\sigma$-algebra defining the measurable subsets of $\Omega$.
A \emph{measure} on $(\Omega, \Sigma_\Omega)$ is a function $\mu : \Sigma_\Omega \to \posreal \cup \{\infty\}$ that satisfies $\mu(\emptyset) = 0$ and is $\sigma$-additive.
For $\real^n$, we write $\Sigma_{\real^n}$ for the Borel $\sigma$-algebra and $\lambda_n$ for the Lebesgue measure on $(\real^n, \Sigma_{\real^n})$.
The Lebesgue integral of a measurable function $f$ with respect to a measure $\mu$ is written $\int f \D \mu$ or $\int f(x) \, \mu(\D x)$.
Given a predicate $\psi$ on $\Omega$, we define the Iverson brackets $[\psi] : \Omega \to \real$ by mapping all elements that satisfy $\psi$ to $1$ and all others to $0$.
For $A \in \Sigma_{\Omega}$ we define the bounded integral $\int_A f \D \mu := \int f(x) \cdot [x \in A]  \mu(\D x)$.

\subsection{Statistical PCF (SPCF)}

As our probabilistic programming language of study, we use \emph{statistical PCF} (SPCF) \cite{MakOPW21}, a typed variant of \cite{BorgstromLGS16}.
SPCF includes primitive operations which are measurable functions $f : \real^{|f|} \to \real$, where $|f| \geq 0$ denotes the arity of the function.
\emph{Values} and \emph{terms} of SPCF are defined as follows:
\begin{align*}
    V &:= x \mid \lit r \mid \lambda x. M \mid \fixLam \varphi x M\\
    M, N, P &:= V \mid M N \mid \ifElse M N P \mid \lit f(M_1, \dots, M_{|f|})\\
    &\quad\quad\mid \sample \mid \score(M)
\end{align*}%
where $x$ and $\varphi$ are variables, $f$ is a primitive operation, and $\lit r$ a constant with $r \in \RR$.
Note that we write $\fixLam \varphi x M$ instead of $\mathsf{Y}(\lambda \varphi x. M)$ for the fixpoint construct.
The branching construct is $\ifSimple M N P$, which evaluates to $N$ if $M \le 0$ and $P$ otherwise.
In SPCF,~$\sample$ draws a random value from the uniform distribution on $[0, 1]$, and $\score(M)$ weights the current execution with the value of $M$.
Samples from a different real-valued distribution $D$ can be obtained by applying the inverse of the cumulative distribution function for $D$ to a uniform sample \cite{RubinsteinK17}.
Most PPLs feature an \textbf{observe} statement instead of manipulating the likelihood weight directly with~$\score$, but they are equally expressive \cite{Staton17}.%
\footnote{\label{fnote:score observe}
In Bayesian terms, an \textbf{observe} statement multiplies the likelihood function by the probability (density) of the observation \cite{GordonHNR14} (as we have used in \cref{ex:pedestrian}).
Scoring makes this explicit by keeping a weight for each program execution \cite{BorgstromLGS16}.
Observing a value $v$ from a distribution $D$ then simply multiplies the current weight by $\pdf_D(v)$ where $\pdf_D$ is the probability density function of $D$ (for continuous distributions) or the probability mass function of $D$ (for discrete distributions).
}
As usual, we write $\letIn{x = M} N$ for $(\lambda x. N) M$, $M;N$ for $\letIn{\_ = M} N$ and $M \oplus_p N$ for $\ifSimple{\sample - p}{M}{N}$.
The type system of our language is as expected, with simple types being generated by $\alpha, \beta := \typeReal \mid \alpha \to \beta$.
Selected rules are given below:

\vspace{2mm}
\noindent
\begin{minipage}{0.5\linewidth}
    \begin{prooftree}
        \AxiomC{}
        \UnaryInfC{$\Gamma \vdash \sample : \typeReal $}
    \end{prooftree}
\end{minipage}%
\begin{minipage}{0.5\linewidth}
    \begin{prooftree}
        \AxiomC{$\Gamma \vdash M : \typeReal$}
        \UnaryInfC{$\Gamma \vdash \score(M) : \typeReal $}
    \end{prooftree}
\end{minipage}

\noindent
\begin{minipage}{0.5\linewidth}
    \begin{prooftree}
        \AxiomC{$\Gamma, \varphi: \alpha \to \beta, x:\alpha \vdash M : \beta$}
        \UnaryInfC{$\Gamma \vdash \fixLam \varphi x M : \alpha \to \beta $}
    \end{prooftree}
\end{minipage}%
\begin{minipage}{0.5\linewidth}
    \begin{prooftree}
        \AxiomC{$\{\Gamma \vdash M_i : \typeReal\}_{i=1}^{|f|}$}
        \UnaryInfC{$\Gamma \vdash \lit f(M_1, \dots, M_{|f|}) : \typeReal $}
    \end{prooftree}
\end{minipage}

\begin{figure}
    \small
        \begin{minipage}{0.55\columnwidth}
            \begin{prooftree}
                 \def\ScoreOverhang{1pt}
                \AxiomC{}
                \UnaryInfC{$\stdConf{ (\lambda x. M) V, \tr, w } \to \stdConf{ M[V/x], \tr, w}$}
            \end{prooftree}
        \end{minipage}%
        \begin{minipage}{0.45\columnwidth}
        	\vspace{1.3mm}
             \begin{prooftree}
                  \def\ScoreOverhang{1pt}
                \AxiomC{}
                \UnaryInfC{$\stdConf{\sample, r \, \tr, w} \to \stdConf{ \lit{r}, \tr, w} $}
            \end{prooftree}
        \end{minipage}

        \begin{prooftree}
        	\def\ScoreOverhang{1pt}
            \AxiomC{}
            \UnaryInfC{$\stdConf {(\fixLam \varphi x M) V, \tr, w} \to \stdConf {M[V/x, (\fixLam \varphi x M)/\varphi], \tr, w} $}
        \end{prooftree}

        \begin{prooftree}
        	\def\ScoreOverhang{1pt}
            \AxiomC{}
            \UnaryInfC{$\stdConf {f(\lit{r_1}, \dots, \lit{r_{|f|}}), \tr, w} \to \stdConf{ \lit{f(r_1, \dots, r_{|f|})}, \tr, w} $}
        \end{prooftree}

         \begin{minipage}{0.5\columnwidth}
            \begin{prooftree}
            	\def\ScoreOverhang{1pt}
               \AxiomC{$r \leq 0$}
               \UnaryInfC{$\stdConf {\ifSimple{\lit{r}}{N}{P},  \tr, w} \to \stdConf {N, \tr, w} $}
           \end{prooftree}
        \end{minipage}%
        \begin{minipage}{0.5\columnwidth}
            \begin{prooftree}
            	\def\ScoreOverhang{1pt}
                \AxiomC{$r > 0$}
                \UnaryInfC{$\stdConf {\ifSimple{\lit{r}}{N}{P},  \tr, w} \to \stdConf {P, \tr, w} $}
            \end{prooftree}
        \end{minipage}

        \begin{minipage}{0.5\columnwidth}
        	\vspace{1mm}
            \begin{prooftree}
            	\def\ScoreOverhang{1pt}
                \AxiomC{$r \geq 0$}
                \UnaryInfC{$\stdConf {\score(\lit{r}), \tr, w} \to \stdConf{ \lit{r}, \tr, w \cdot r} $}
            \end{prooftree}
        \end{minipage}%
        \begin{minipage}{0.5\columnwidth}
            \begin{prooftree}
            	\def\ScoreOverhang{1pt}
                \AxiomC{$\stdConf {R, \tr, w} \to \stdConf {M, \tr', w'}$}
                \UnaryInfC{$\stdConf {E[R], \tr, w} \to \stdConf {E[M], \tr', w'} $}
            \end{prooftree}
        \end{minipage}
    \caption{Standard (CbV) reduction rules for SPCF ($\to$).} \label{fig:reductionRules}
\end{figure}

\subsection{Trace Semantics}

Following \cite{BorgstromLGS16}, we endow SPCF with a trace-based operational semantics.
We evaluate a probabilistic program $P$ on a fixed \defn{trace} $\tr = \langle r_1, \dots, r_n \rangle \in \traces := \bigcup_{n \in \NN} [0,1]^n$, which \emph{predetermines} the probabilistic choices made during the evaluation.
Our semantics therefore operates on configurations of the form $\stdConf{M, \tr, w}$ where $M$ is an SPCF term, $\tr$ is a trace and $w \in \real_{\geq 0}$ a weight.
The call-by-value (CbV) reduction is given by the rules in \cref{fig:reductionRules}, where $E[\cdot]$ denotes a CbV evaluation context.
The definition is standard \cite{BorgstromLGS16, MakOPW21,BeutnerO21}.
Given a program $\vdash P : \typeReal$, we call a trace $\tr$ \defn{terminating} just if $\stdConf {P, \tr, 1} \to^* \stdConf {V, \langle\rangle, w}$ for some value $V$ and weight $w$, i.e.~if the samples drawn are as specified by $\tr$, the program $P$ terminates.
Note that we require the trace $\tr$ to be completely used up.
As $P$ is of type $\typeReal$ we can assume that $V = \lit{r}$ for some $r \in \real$.
Each terminating trace $\tr$ therefore uniquely determines the returned \emph{value} $\lit{r}$ where $r =: \valueSem P(\tr) \in \RR$, and the \emph{weight} $w =: \weightSem P(\tr) \in  \RR_{\ge 0}$, of the execution.
For a nonterminating trace $\tr$, $\valueSem P(\tr)$ is undefined and $\weightSem P(\tr) := 0$.

\begin{example}
\label{ex:pedestrian2}
    Consider \cref{ex:pedestrian}.
    On the trace $\tr = \langle 0.1,\allowbreak 0.2,\allowbreak 0.4,\allowbreak 0.7,\allowbreak 0.8\rangle \in [0,1]^5 \subseteq \traces$, the pedestrian walks $0.2$ away from their home (taking the left branch of $\oplus_{0.5}$ as $0.4 \le 0.5$) and $0.7$ towards their home (as $0.8 > 0.5$), hence:
    \[ \valueSem P(\tr) = 3 \times 0.1 = 0.3, \; \weightSem P(\tr) = \pdf_{\Normal(1.1, 0.1)}(0.9). \]
\end{example}

In order to do measure theory, we need to turn our set of traces into a measurable space.
The trace space $\traces$ is equipped with the $\sigma$-algebra $\Sigma_\traces := \{ \bigcup_{n \in \NN} U_n \mid U_n \in \Sigma_{[0,1]^n} \}$ where $\Sigma_{[0,1]^n}$ is the Borel $\sigma$-algebra on $[0,1]^n$ . We define a measure $\mu_\traces$ by $\mu_\traces(U) := \sum_{n\in \NN} \lambda_n(U \cap [0,1]^n)$, as in \cite{BorgstromLGS16}.

We can now define the semantics of an SPCF program $\vdash P : \typeReal$ by using the weight and returned value of (executions of $P$ determined by) individual traces.
Given $U \in \Sigma_\real$, we need to define the likelihood of $P$ evaluating to a value in $U$.
To this end, we set $\valueSem P^{-1}(U) := \{ \tr \in \traces \mid \stdConf{P, \tr, 1} \to^* \stdConf{\lit{r}, \langle\rangle, w}, r \in U  \}$, i.e.~the set of traces on which the program $P$ reduces to a value in $U$.
As shown in \citep[Lem.~9]{BorgstromLGS16}, $\valueSem P^{-1}(U)$ is measurable.
Thus, we can define (cf. \cite{BorgstromLGS16,MakOPW21})
\[
\textstyle\measureSem P(U) := \int_{\valueSem P^{-1}(U)} \weightSem P(\tr)  \,\mu_\traces(\D \tr).
\]
That is, the integral takes all traces $\tr$ on which $P$ evaluates to a value in $U$, weighting each $\tr$ with the weight $\weightSem P(\tr)$ of the corresponding execution.
{A program $P$ is called \defn{almost surely terminating (AST)} if it terminates with probability 1, i.e.~$\mu_\traces(\valueSem P^{-1}(\RR)) = 1$.
This is a necessary assumption for approximate inference algorithms (since they execute the program).
See \cite{BorgstromLGS16} for a more in-depth discussion of this (standard) sampling-style semantics.

\paragraph{Normalizing constant and integrability}

In Bayesian statistics, one is usually interested in the \emph{normalised} posterior, which is a conditional probability distribution.
We obtain the normalised denotation as
$ \posterior_P := \frac{\measureSem P}{Z_P}$
where $Z_P := \measureSem P(\RR)$ is the \defn{normalising constant}.
We call $P$ \defn{integrable} if $0 < Z_P < \infty$.
The bounds computed in this paper (on the unnormalised denotation $\measureSem{P}$) allow us to compute bounds on the normalizing constant $Z_P$, and thereby also on the normalised denotation.
All bounds reported in this paper (in particular in \cref{sec:7practical-evaluation}) refer to the \emph{normalised} denotation.


\section{Interval Trace Semantics}
\label{sec:3intervals}

In order to obtain guaranteed bounds on the distribution denotation $\measureSem P$ (and also on $\posterior_P$) of a program $P$, we present an interval-based semantics.
In our semantics, we approximate the outcomes of $\sample$ with intervals and handle arithmetic operations by means of interval arithmetic (which is similar to the approach by \citet{BeutnerO21} in the context of termination analysis).
Our semantics enables us to reason about the denotation of a program \emph{without} considering the uncountable space of traces explicitly.

\subsection{Interval Arithmetic}\label{sec:intervalArith}

For our purposes, an \defn{interval} has the form $[a, b]$ which denotes the set $\{ x \in \RR \mid a \le x \le b \}$, where $a \in \RR \cup \{-\infty\}$, $b \in \RR \cup \{\infty\}$, and $a \leq b$.
For consistency, we write $[0, \infty]$ instead of the more typical $[0,\infty)$.
For $X \subseteq \RR \cup \{-\infty,\infty\}$, we denote by $\iv_X$ the set of intervals with endpoints in $X$, and simply write $\iv$ for $\iv_{\RR \cup \{-\infty,\infty\}}$.
An $n$-dimensional \defn{box} is the Cartesian product of $n$ intervals.

We can lift functions on real numbers to intervals as  follows: for each $f: \RR^n \to \RR$ we define $f^\iv: \iv^n \to \iv$ by
\[ f^\iv([a_1,b_1], \dots, [a_n,b_n]) := [\inf F, \sup F] \]
where $F := f([a_1,b_1], \dots, [a_n,b_n])$.
For common functions like $+$, $-$, $\times$, $|\cdot|$, $\min$, $\max$, and monotonically increasing or decreasing functions $f: \RR \to \RR$, their interval-lifted counterparts can easily be computed, from the values of the original function on just the endpoints of the input interval.
For example, addition lifts to $[a_1,b_1] +^\iv [a_2,b_2] = [a_1 + a_2, b_1 + b_2]$;
similarly for multiplication $\times^\iv$.

\subsection{Interval Traces and Interval SPCF}

In our interval interpretation, probabilistic programs are run on \emph{interval traces}.
An \defn{interval trace}, $\langle I_1, \dots, I_n \rangle \in \traces_\iv := \bigcup_{n\in \NN} (\iv_{[0,1]})^n$, is a finite sequence of intervals $I_1, \dots, I_n$, each with endpoints between $0$ and $1$.
To distinguish ordinary traces $\tr \in \traces$ from interval traces $\ivtr \in \traces_\iv$, we call the former \emph{concrete} traces.

We define the \defn{refinement} relation $\refines$ between concrete and interval traces as follows:
for $\tr = \langle r_1, \dots, r_n \rangle \in \traces$ and $\ivtr = \langle I_1, \dots, I_m \rangle\in \traces_\iv$, we define $\tr \refines \ivtr$ just if $n = m$ and for all $i$, $r_i \in I_i$.
For each interval trace $\ivtr$, we denote by $\tracesin{\ivtr} := \{ \tr \in \traces \mid \tr \refines \ivtr \}$ the set of all refinements of $\ivtr$.

To define a reduction of a term on an interval trace, we extend SPCF with \emph{interval literals} $\lit{[a,b]}$, which replace the literals $\lit{r}$ but are still considered values of type $\typeReal$.
In fact, $\lit r$ can be read as an abbreviation for $\lit{[r,r]}$.
We call such terms \defn{interval terms}, and the resulting language \defn{Interval SPCF}.

\begin{figure}
    \small
    \begin{minipage}{0.55\columnwidth}
        \begin{prooftree}
            \def\ScoreOverhang{0pt}
            \AxiomC{}
            \UnaryInfC{$\intConf {(\lambda x. M) V, \ivtr, w} \to_\iv \intConf {M[V/x], \ivtr, w}$}
        \end{prooftree}
    \end{minipage}%
    \begin{minipage}{0.45\columnwidth}
    	\vspace{1.5mm}
        \begin{prooftree}
            \def\ScoreOverhang{0pt}
            \AxiomC{}
            \UnaryInfC{$\intConf{ \sample, I \, \ivtr, w}  \to_\iv \intConf{ \lit{I}, \ivtr, w} $}
        \end{prooftree}
    \end{minipage}

    \begin{prooftree}
        \AxiomC{}
        \UnaryInfC{$\intConf {(\fixLam \varphi x M) V, \ivtr, w} \to_\iv \intConf {M[V/x, (\fixLam \varphi x M)/\varphi], \ivtr, w} $}
    \end{prooftree}

    \begin{minipage}{1\columnwidth}
          \begin{prooftree}
              \def\ScoreOverhang{-2pt}
             \AxiomC{$b \leq 0$}
             \UnaryInfC{$\intConf{\ifSimple{\lit{[a, b]}}{N}{P}, \ivtr, w} \to_\iv \intConf {N, \ivtr, w} $}
         \end{prooftree}
    \end{minipage}

    \begin{minipage}{1\columnwidth}
    \begin{prooftree}
        \def\ScoreOverhang{-2pt}
        \AxiomC{$a > 0$}
        \UnaryInfC{$\intConf {\ifSimple{\lit{[a, b]}}{N}{P},\ivtr, w} \to_\iv \intConf {P, \ivtr, w} $}
    \end{prooftree}
    \end{minipage}

    \begin{prooftree}
        \AxiomC{}
        \UnaryInfC{$\intConf{ f(\lit{I_1}, \dots, \lit{I_{|f|}}), \ivtr, w} \to_\iv \intConf{ \lit{f^\iv(I_1, \dots, I_{|f|})}, \ivtr, w} $}
    \end{prooftree}

    \begin{prooftree}
        \AxiomC{$a \geq 0$}
        \UnaryInfC{$\intConf {\score(\lit{[a, b]}), \ivtr, w} \to_\iv \intConf {\lit{[a, b]}, \ivtr, w \times^\iv [a, b]} $}
    \end{prooftree}

    \begin{prooftree}
        \AxiomC{$\intConf {R, \ivtr, w} \to_\iv \intConf {M, \ivtr', w'}$}
        \UnaryInfC{$\intConf {E[R], \ivtr, w} \to_\iv \intConf {E[M], \ivtr', w'} $}
    \end{prooftree}
    \caption{Interval reduction rules for (interval) SPCF ($\to_\iv$).}
    \label{fig:interval-semantics}
\end{figure}

\paragraph{Reduction}

The interval-based reduction $\to_\iv$ now operates on configurations $\intConf{M, \ivtr, w}$ of interval terms $M$, interval traces $\ivtr \in \traces_\iv$, and interval weights $w \in \iv_{\RR_{\ge 0} \cup \{\infty\}}$.
The redexes and evaluation contexts of SPCF extend naturally to interval terms.
The reduction rules are given in \cref{fig:interval-semantics}.%
\footnote{
For conditionals, the interval bound is not always precise enough to decide which branch to take, so the reduction can get stuck if $a \le 0 < b$.
We could include additional rules to overapproximate the branching behaviour (see \ifFull{\cref{app:sec-additional-reduction-rules}}).
But the rules given here simplify the presentation and are sufficient to prove soundness and completeness.
}

Given a program $\vdash P : \typeReal$, the reduction relation $\to_\iv$ allows us to define the \emph{interval weight} function ($ \weightSem P^\iv : \traces_\iv \to \iv_{\RR_{\ge 0} \cup \{\infty\}}$) and \emph{interval value} function ($\valueSem P^\iv : \traces_\iv \to \iv$) by:
\begin{align*}
  \weightSem P^\iv(\ivtr) &:= \begin{cases}
    w &\text{if } \intConf{P, \ivtr, 1} \to^*_\iv \intConf{\lit I, \langle\rangle, w} \\
    [0, \infty] &\text{otherwise,}
  \end{cases} \\
  \valueSem P^\iv(\ivtr) &:= \begin{cases}
    I &\text{if } \intConf{P, \ivtr, 1} \to^*_\iv \intConf{\lit{I}, \langle\rangle, w} \\
    [-\infty, \infty] &\text{otherwise.}
  \end{cases}
\end{align*}

It is not difficult to prove the following relationship between standard and interval reduction.
\begin{restatable}{lemma}{lemIntervalApproximation}
  Let $ \vdash P : \typeReal$ be a program.
  For any interval trace $\ivtr$ and concrete trace $\tr \refines \ivtr$, we have $\weightSem P(\tr) \in \weightSem P^\iv(\ivtr)$ and $\valueSem P(\tr) \in \valueSem P^\iv(\ivtr)$ (provided $\valueSem P(\tr)$ is defined).
  \label{lem:interval-approximation}
\end{restatable}

\subsection{Bounds from Interval Traces}
\label{sec:boundsFromIntervalTraces}

\paragraph{Lower bounds}
How can we use this interval trace semantics to obtain lower bounds on $\measureSem P$?
We need a few definitions.
Two intervals $[a_1,b_1],\allowbreak [a_2,b_2] \in \iv$ are called \defn{almost disjoint} if $b_1 \le a_2$ or $b_2 \le a_1$.
Interval traces $\langle I_1,\dots,I_m \rangle$ and $\langle J_1,\dots,J_n \rangle \in \traces_\iv$ are called \defn{compatible} if there is an index $i \in \{1,\dots,\min(m,n)\}$ such that $I_i$ and $J_i$ are almost disjoint.
A set of interval traces is called  \emph{compatible} if its elements are pairwise compatible.
We define the \defn{volume} of an interval trace $\ivtr = \langle[a_1,b_1],\allowbreak\dots,\allowbreak[a_n,b_n]\rangle$ as $\volume (\ivtr) := \prod_{i=1}^n (b_i - a_i)$.

Let $\traceset \subseteq \traces_\iv$ be a countable and compatible set of interval traces.
Define the \emph{lower bound} on $\measureSem P$ by
\begin{align*}
	\lowerBound P^\traceset (U) \!:=\! \sum_{\ivtr \in \traceset} \volume(\ivtr) \! \cdot \! (\min \weightSem P^\iv(\ivtr)) \!\cdot\! \big[\valueSem P^\iv(\ivtr) \subseteq U\big]
\end{align*}
for $U \in \Sigma_\RR$.
That is, we sum over each interval trace in $\traceset$ whose value is \emph{guaranteed} to be in $U$, weighted by its volume and the lower bound of its weight interval.
Note that, in general, $\lowerBound P^\traceset$ is not a measure, but merely a \emph{superadditive measure}.%
\footnote{A \emph{superadditive measure} $\mu$ on $(\Omega, \Sigma_\Omega)$ is a measure, except that $\sigma$-additivity is replaced by $\sigma$-superadditivity: $\mu(\bigcup_{i\in\NN} U_i) \ge \sum_{i\in\NN} \mu(U_i)$ for a countable, pairwise disjoint family $(U_i)_{i \in \NN} \in \Sigma_\Omega$.}

\paragraph{Upper bounds}

For upper bounds, we require the notion of a set of interval traces being \emph{exhaustive}, which is easiest to express in terms of infinite traces.
Let $\traces_\infty := [0,1]^\omega$ be the set of infinite traces.
Every interval trace $\ivtr$ \emph{covers} the set of infinite traces with a prefix contained in $\ivtr$, i.e.~$\mathit{cover}(\ivtr) := \tracesin{\ivtr} \times \traces_\infty$ (where the Cartesian product $\times$ can be viewed as trace concatenation).
A countable set of (finite) interval traces $\traceset \subseteq \traces_\iv$ is called \defn{exhaustive} if $\bigcup_{\ivtr \in \traceset} \mathit{cover}(\ivtr)$ covers almost all of $\traces_\infty$, i.e.~$\mu_{\traces_\infty}(\traces_\infty \setminus \bigcup_{\ivtr \in \traceset} \mathit{cover}(\ivtr)) = 0$.%
\footnote{The $\sigma$-algebra on $\traces_\infty$ is defined as the smallest $\sigma$-algebra that contains all sets $U \times \traces_\infty$ where $U \in \Sigma_{[0, 1]^n}$ for some $n \in \NN$. The measure $\mu_{\traces_\infty}$ is the unique measure with $\mu_{\traces_\infty}(U \times \traces_\infty) = \lambda_n(U)$ when $U \in \Sigma_{[0, 1]^n}$.}
Phrased differently, almost all concrete traces must have a finite prefix that is contained in some interval trace in $\traceset$.
Therefore, the analysis in the interval semantics on $\traceset$ covers the behaviour on almost all concrete traces (in the original semantics).

\begin{example}
	(i) The singleton set $\{\langle [0,1], [0,0.6] \rangle\}$ is not exhaustive as, e.g.~all infinite traces $\langle r_1, r_2, \dots \rangle$ with $r_2 > 0.6$ are not covered.
	(ii) The set $\{\langle [0,0.6] \rangle , \langle [0.3, 1] \rangle\}$ is exhaustive, but not compatible.
	(iii) Define
        $\traceset_1 := \{ \langle [\tfrac12, 1]^{\dots n}, [0,\tfrac13] \rangle \mid n \in \NN \}$ and $\traceset_2 := \{ \langle [\tfrac12, 1]^{\dots n},\allowbreak [0,\tfrac12] \rangle \mid n \in \NN \}$ where $x^{\dots n}$ denotes $n$-fold repetition of $x$.
	$\traceset_1$ is compatible but not exhaustive. For example, it doesn't cover the set $[\tfrac12, 1] \times (\tfrac13, \tfrac12) \times \traces_\infty$, i.e.~all traces $\langle r_1, r_2, \dots \rangle$ where $r_1 \in [\tfrac12, 1]$ and $r_2 \in (\tfrac{1}{3}, \tfrac{1}{2})$.
	$\traceset_2$ is compatible and exhaustive (the set of non-covered traces $(\tfrac12, 1]^\omega$ has measure $0$).
\end{example}

Let $\traceset \subseteq \traces_\iv$ be a countable and exhaustive set of interval traces.
Define the \emph{upper bound} on $\measureSem P$ by
\begin{align*}
	\upperBound P^\traceset(U) \!:=\!\! \sum_{\ivtr \in \traceset} \volume(\ivtr) \! \cdot \! (\sup \weightSem P^\iv(\ivtr)) \! \cdot \! \big[\valueSem P^\iv(\ivtr) \cap U \ne \emptyset\big]
\end{align*}
for $U \in \Sigma_\RR$.
That is, we sum over each interval trace in $\traceset$ whose value \emph{may} be in $U$, weighted by its volume and the upper bound of its weight interval.
Note that {$\upperBound P^\traceset$} is not a measure but only a \emph{subadditive measure}.%
\footnote{A \emph{subadditive measure} $\mu$ on $(\Omega, \Sigma_\Omega)$ is a measure, except that $\sigma$-additivity is replaced by $\sigma$-subadditivity: $\mu(\bigcup_{i\in\NN} U_i) \le \sum_{i\in\NN} \mu(U_i)$ for a countable, pairwise disjoint family $(U_i)_{i \in \NN} \in \Sigma_\Omega$.}


\section{Soundness and Completeness}
\label{sec:4intervals-theory}

\subsection{Soundness}

We show that the two bounds described above are \emph{sound}, in the following sense.

\begin{theorem}[Sound lower bounds]\label{thm:lowerBoundsSound}
  Let $\traceset$ be a countable and compatible set of interval traces and $\vdash P: \typeReal$ a program.
  Then
  \( \lowerBound P^\traceset \le \measureSem P. \)
\end{theorem}
\begin{proof}
For any $U \in \Sigma_\RR$, we have:
\begin{align}
\lowerBound P^\traceset(U)
  &= \sum_{\ivtr \in \traceset} \volume(\ivtr) (\min \weightSem P^\iv(\ivtr)) \big[\valueSem P^\iv(\ivtr) \subseteq U\big] \nonumber\\
  &= \sum_{\ivtr \in \traceset} \int_{\tracesin{\ivtr}} (\min \weightSem P^\iv(\ivtr)) \big[\valueSem P^\iv(\ivtr) \subseteq U\big] \D \tr \nonumber\\
  &\le \sum_{\ivtr \in \traceset} \int_{\tracesin{\ivtr}} \weightSem P(\tr) \big[\valueSem P(\tr) \in U\big] \D \tr \label{eq:fourth soundness}\\
  &= \int_{\bigcup_{\ivtr \in \traceset} \tracesin{\ivtr}} \weightSem P(\tr) \big[\valueSem P(\tr) \in U\big] \D \tr \label{eq:fifth soundness}\\
  &\le \int_{\traces} \weightSem P(\tr) \big[\valueSem P(\tr) \in U\big] \D \tr = \measureSem P(U) \label{eq:sixth soundness}
\end{align}
where \cref{eq:fourth soundness} follows from \cref{lem:interval-approximation}, \cref{eq:fifth soundness} from compatibility, and \cref{eq:sixth soundness} from $\bigcup_{\ivtr \in \traceset} \tracesin{\ivtr} \subseteq \traces$.
\end{proof}

\begin{restatable}[Sound upper bounds]{theorem}{thmSoundUpper}\label{thm:upperBoundsSound}
  Let $\traceset$ be a countable and exhaustive set of interval traces and $\vdash P: \typeReal$ a program.
  Then
  \( \measureSem P \le \upperBound P^\traceset. \)
\end{restatable}
\begin{proof}[Proof sketch]
The formal proof is similar to the soundness proof for the lower bound in \Cref{thm:lowerBoundsSound}, but needs an infinite trace semantics \cite{CulpepperC17} for probabilistic programs and is given in \ifFull{\cref{sec:inf-trace-sem}}.
The idea is that each interval trace $\ivtr$ summarises all infinite traces starting with $\tracesin{\ivtr}$, i.e.~all traces in $\mathit{cover}(\ivtr)$.
Exhaustivity ensures that almost all infinite traces are ``covered''.
\end{proof}

\subsection{Completeness}
\label{sec:completeness}

The soundness results for upper and lower bounds allow us to derive bounds on the denotation of a program.
One would expect that a finer partition of interval traces will yield more precise bounds.
In this section, we show that for a program $P$ and an interval $I \in \iv$, the approximations $\lowerBound P^\traceset(I)$ and $\upperBound P^\traceset(I)$ can in fact come arbitrarily close to $\measureSem P(I)$ for suitable $\traceset$.
However, this is only possible under certain assumptions.

\paragraph{Assumption 1: use of sampled values}
Interval arithmetic is imprecise if the same value is used more than once: consider, for instance, $\letIn{s = \sample} \ifElse{s - s}{0}{1}$ which deterministically evaluates to $0$.
However, in interval arithmetic, if $x$ is approximated by an interval $[a,b]$ with $a < b$, the difference $x - x$ is approximated as $[a - b, b - a]$, which always contains both positive and negative values.
So no non-trivial interval trace can separate the two branches.

To avoid this, we could consider a call-by-name semantics (as done in \cite{BeutnerO21}) where sample values can only be used once by definition.
However, many of our examples cannot be expressed in the call-by-name setting, so we instead propose a less restrictive criterion to guarantee completeness for call-by-value:
we allow sample values to be used more than once, but at most once in the guard of each conditional, at most once in each score expression, and at most once in the return value.
While this prohibits terms like the one above, it allows, e.g.~$\letIn{s = \sample} \ifElse{s}{\lit f(s)}{\lit g(s)}$.
This sufficient condition is formalised in \ifFull{\cref{sec:qtt}}.
Most examples we encountered in the literature satisfy this assumption.

\paragraph{Assumption 2: primitive functions}

In addition, we require mild assumptions on the primitive functions, called \emph{boxwise continuity} and \emph{interval separability}.

We need to be able to approximate a program's weight function by step functions in order to obtain tight bounds on its integral.
A function $f: \RR^n \to \RR$ is \defn{boxwise continuous} if it can be written as the countable union of continuous functions on boxes, i.e.~if there is a countable union of pairwise almost disjoint boxes $B_i$ such that $\bigcup B_i = \RR^n$ and the restriction $f|_{B_i}$ is continuous for each $B_i$.

Furthermore, we need to approximate preimages.
Formally, we say that $A$ is a \defn{tight subset} of $B$ (written $A \tightsubset B$) if $A \subseteq B$ and $B \setminus A$ is a null set.
A function $f: \RR^n \to \RR$ is called \defn{interval separable} if for every interval $[a,b] \in \iv$, there is a countable set $\mathcal B$ of boxes in $\RR^n$ that tightly approximates the preimage, i.e.~$\bigcup \mathcal B \tightsubset f^{-1}([a,b])$.
A sufficient condition for checking this is the following.
If $f$ is boxwise continuous and preimages of points have measure zero, then $f$ is already interval separable (cf.~\ifFull{\cref{lem:continuous-preimage-null-interval-separable}}).

We assume the set $\mathcal F$ of primitive functions is closed under composition and each $f\in \mathcal F$ is boxwise continuous and interval separable.

\paragraph{The completeness theorem}

Using these two assumptions, we can state completeness of our interval semantics.

\begin{restatable}[Completeness of interval approximations]{theorem}{thmCompleteness}\label{thm:completeness}
\label{thm:Completeness of interval approximations}
Let $I \in \iv$ and $\vdash P : \typeReal$ be an almost surely terminating program satisfying the two assumptions discussed above.
Then, for all $\epsilon > 0$, there is a countable set of interval traces $\traceset \subseteq \traces_\iv$ that is compatible and exhaustive such that
\begin{align*}
	\textstyle\upperBound P^{\traceset}(I) - \epsilon \le \measureSem P(I) \le \lowerBound P^{\traceset}(I) + \epsilon.
\end{align*}
\end{restatable}
\begin{proof}[Proof sketch.]
We consider each branching path through the program separately.
The set of relevant traces for a given path is a preimage of intervals under compositions of interval separable functions, hence can essentially be partitioned into boxes.
By boxwise continuity, we can refine this partition such that the weight function is continuous on each box.
To approximate the integral, we pass to a refined partition again, essentially computing Riemann sums.
The latter converge to the Riemann integral, which agrees with the Lebesgue integral under our conditions, as desired.
\end{proof}

For the lower bound, we can actually derive $\epsilon$-close bounds using only finitely many interval traces:

\begin{restatable}{corollary}{corollaryCompleteness}
	Let $I \in \iv$ and $\vdash P : \typeReal$ be as in \cref{thm:Completeness of interval approximations}.
    There is a sequence of finite, compatible sets of interval traces $\traceset_1, \traceset_2, \ldots \subseteq \traces_\iv$ s.t.~\( \lim_{n \to \infty} \lowerBound P^{\traceset_n}(I) = \measureSem P(I). \)
\end{restatable}

For the upper bound, a restriction to finite sets $\traceset$ of interval traces is, in general, not possible:
if the weight function for a program is unbounded, it is also unbounded on some $\ivtr \in \traceset$.
Then $\weightSem P^\iv(\ivtr)$ is an infinite interval, implying $\upperBound P^\traceset(I) = \infty$ (see \ifFull{\cref{rem:countable-traces}} for details).
Despite the (theoretical) need for countably infinite many interval traces, we can, in many cases, compute finite upper bounds by making use of an interval-based static approximation, formalised as a type system in the next section.


\section{Weight-aware Interval Type System}
\label{sec:5interval-analysis}

To obtain sound bounds on the denotation with only finitely many interval traces, we present an interval-based type system that can derive static bounds on a program.
Crucially, our type-system is \emph{weight-aware}: we bound not only the return value of a program but also the weight of an execution.
Our analyzer GuBPI uses it for two purposes.
First, it allows us to derive upper bounds even for areas of the sample space not covered with interval traces.
Second, we can use our analysis to derive a \emph{finite} (and sound) approximation of the infinite number of symbolic execution paths of a program (more details are given in \cref{sec:6linear}).
Note that the bounds inferred by our system are \emph{interval bounds}, which allow for seamless integration with our interval trace semantics.
In this section, we present the interval type system and sketch a constraint-based type inference method.

\begin{figure}[!t]
	\footnotesize
	\vspace{-0.3cm}
	\begin{minipage}{0.2\columnwidth}
		\vspace{6mm}
		\begin{prooftree}
			\def\ScoreOverhang{1pt}
			\AxiomC{$x:\sigma \in \Gamma$}
			\UnaryInfC{$\Gamma \vdash x: \exType{\sigma}{\mathbf{1}} $}
		\end{prooftree}
	\end{minipage}\hfill
	\begin{minipage}{0.32\columnwidth}
		\begin{prooftree}
			\def\ScoreOverhang{1pt}
			\def\defaultHypSeparation{\hskip .1in}
			\AxiomC{$\Gamma \vdash M : \calA$}
			\AxiomC{$\calA \sqsubseteq_\calA \calB$}
			\BinaryInfC{$\Gamma \vdash M : \calB$}
		\end{prooftree}
	\end{minipage}%
	\begin{minipage}{0.45\columnwidth}
		\vspace{5mm}
		\begin{prooftree}
			\def\ScoreOverhang{1pt}
			\AxiomC{$\Gamma; \varphi: \sigma \to \calA ; x:\sigma \vdash M : \calA$}
			\UnaryInfC{$\Gamma \vdash \fixLam{\varphi}{x} M : \exType{\sigma \to \calA}{\mathbf{1}} $}
		\end{prooftree}
	\end{minipage}

	\vspace{0.0cm}

	\begin{minipage}{0.3\columnwidth}
		\vspace{8mm}
		\begin{prooftree}
			\def\ScoreOverhang{1pt}
			\AxiomC{$\Gamma; x:\sigma \vdash M : \calA$}
			\UnaryInfC{$\Gamma \vdash \lambda x. M : \exType{\sigma \to \calA}{\mathbf{1}} $}
		\end{prooftree}
	\end{minipage}%
	\begin{minipage}{0.7\columnwidth}
		\begin{prooftree}
			\def\ScoreOverhang{1pt}
			\def\defaultHypSeparation{\hskip .15in}
			\AxiomC{$\Gamma\vdash M : \exType{\sigma_1 \to \exType{\sigma_2}{\myint{e, f}}}{\myint{a, b}}$}
			\AxiomC{$\Gamma \vdash N : \exType{\sigma_1}{\myint{c, d}}$}
			\BinaryInfC{$\Gamma \vdash M N : \exType{\sigma_2}{\myint{a, b} \times^\iv \myint{c, d} \times^\iv \myint{e, f}}$ }
		\end{prooftree}
	\end{minipage}

	\vspace{0.1cm}

	\begin{minipage}{0.2\columnwidth}
		\vspace{6.5mm}
		\begin{prooftree}
			\def\ScoreOverhang{1pt}
			\AxiomC{}
			\UnaryInfC{$\Gamma \vdash \lit{r} : \exType{\myint{r, r}}{\mathbf{1}} $}
		\end{prooftree}
	\end{minipage}%
	\begin{minipage}{0.8\columnwidth}
		\begin{prooftree}
			\def\defaultHypSeparation{\hskip .15in}
			\def\ScoreOverhang{1pt}
			\AxiomC{$\Gamma \vdash M : \exType{\myint{\_, \_}}{\myint{a, b}}$}
			\AxiomC{$\Gamma \vdash N :\exType{ \sigma}{\myint{c, d}}$}
			\AxiomC{$\Gamma \vdash P : \exType{\sigma}{\myint{c, d}}$}
			\TrinaryInfC{$\Gamma \vdash \ifSimple M N P : \exType{\sigma}{\myint{a,b} \times^\iv \myint{c, d}} $}
		\end{prooftree}
	\end{minipage}

	\vspace{0.1cm}

	\begin{minipage}{0.25\columnwidth}
		\vspace{6.5mm}
		\begin{prooftree}
			\def\ScoreOverhang{1pt}
			\AxiomC{}
			\UnaryInfC{$\Gamma \vdash \sample : \exType{\myint{0, 1}}{\mathbf{1}} $}
		\end{prooftree}
	\end{minipage}%
	\begin{minipage}{0.75\columnwidth}
		\begin{prooftree}
			\def\ScoreOverhang{1pt}
			\AxiomC{$\Gamma \vdash M : \exType{\myint{a, b}}{\myint{c, d}}$}
			\UnaryInfC{$\Gamma \vdash \score(M) : \exType{\myint{a, b} \sqcap \myint{0, \infty}}{\myint{c, d} \times^\iv \big(\myint{a, b} \sqcap \myint{0, \infty}\big) } $}
		\end{prooftree}
	\end{minipage}

	\vspace{0.1cm}

	\begin{minipage}{1\columnwidth}
		\begin{prooftree}
			\def\ScoreOverhang{1pt}
			\AxiomC{$\Gamma \vdash M_1 : \exType{\myint{a_1,b_1}}{\myint{c_1, d_1}}$}
			\AxiomC{$\cdots$}
			\AxiomC{$\Gamma \vdash M_{|f|} : \exType{\myint{a_{|f|}, b_{|f|}}}{\myint{c_{|f|}, d_{|f|}}}$}
			\TrinaryInfC{$\Gamma \vdash f(M_1, \dots, M_{|f|}) : \exType{f^\iv(\myint{a_1,b_1}, \dots, \myint{a_{|f|}, b_{|f|}})}{(\times^\iv)_{i=1}^{|f|} \myint{c_i, d_i}} $}
		\end{prooftree}
	\end{minipage}

	\caption{Weight-aware interval type system for SPCF. We abbreviate $\mathbf{1} := [1, 1]$.} \label{fig:typeSystemSelection}
\end{figure}

\subsection{Interval Types}

We define interval types by the following grammar:
\begin{align*}
	\sigma := I \mid \sigma \to \mathcal{A} \quad\quad	\mathcal{A} := \exType{\sigma}{I}
\end{align*}%
where $I \in \iv$ is an interval.
For readers familiar with refinement types, it is easiest to view the type $\sigma = I$ as the refinement type $\{x : \RR \mid x \in I\}$.
The definition of the syntactic category $\calA$ by mutual recursion with $\sigma$ gives a bound on the weight of the execution.
We call a type $\sigma$ \emph{weightless} and a type $\calA$ \emph{weighted}.
The following examples should give some intuition about the types.

\begin{example}\label{ex:typeExample}
    Consider the example term
    \begin{align*}
        \big(\fixLam{\varphi}{x} 5 \times x \oplus_{0.5} \mathit{sigm}(\varphi\, x + \score \,\sample) \big) (4 \times \sample)
    \end{align*}
	where $\mathit{sigm} : \real \to [0, 1]$ is the sigmoid function.
    In our type system, this term can be typed with the weighted type \scalebox{0.5}{$\exType{[0, 20]}{[0,1]}$}, which indicates that any terminating execution of the term reduces to a value (a number) within $[0, 20]$ and the weight of any such execution lies within $[0, 1]$.
\end{example}

\begin{example}\label{ex:pedestrianType}
	We consider the fixpoint subexpression of the pedestrian example in \cref{ex:pedestrian} which is
	\begin{center}
		\vspace{-2mm}
		\scalebox{0.98}{\parbox{\linewidth}{
				\begin{align*}
					\mu^\varphi_x. \ifElse{x}{0}{\big(\lambda \mathit{step}. \mathit{step} + \varphi( (x \!+\! \mathit{step})\oplus_{0.5} (x \!-\! \mathit{step}) )\big) \sample}.
				\end{align*}
		}}
	\end{center}
	Using the typing rules (defined below), we can infer the type
	\scalebox{0.5}{$\exType{[a, b] \to \exType{[0, \infty]}{[1, 1]}}{[1,1]}$}
	for any $a, b$.
    This type indicates that any terminating execution reduces to a function value (of simple type $\typeReal \to \typeReal$) with weight within $[1, 1]$.
    If this function value is then called on a value within $[a, b]$, any terminating execution reduces to a value within $[0, \infty]$ with a weight within $[1, 1]$.
\end{example}

\paragraph{Subtyping}

The partial order on intervals naturally extends to our type system.
For base types $I_1$ and $I_2$, we define $I_1 \sqsubseteq_\sigma I_2$ just if $I_1 \sqsubseteq I_2$, where $\sqsubseteq$ is interval inclusion.
We then extend this via:\\[-3mm]
\begin{minipage}{0.5\columnwidth}
    \begin{prooftree}
        \AxiomC{$\sigma_2 \sqsubseteq_\sigma \sigma_1$}
        \AxiomC{$\calA_1 \sqsubseteq_\calA \calA_2$}
        \BinaryInfC{$\sigma_1 \to \calA_1 \sqsubseteq_\sigma \sigma_2 \to \calA_2$}
    \end{prooftree}
\end{minipage}%
\begin{minipage}{0.5\columnwidth}
	\vspace{5mm}
   \begin{prooftree}
       \AxiomC{$\sigma_1 \sqsubseteq_\sigma \sigma_2$}
       \AxiomC{$I_1 \sqsubseteq I_2$}
       \BinaryInfC{$\exType{\sigma_1}{I_1} \sqsubseteq_\calA \exType{\sigma_2}{I_2}$}
   \end{prooftree}
\end{minipage}

\vspace{2mm}
\noindent
Note that in the case of weighted types, the subtyping requires not only that the weightless types be subtype-related ($\sigma_1 \sqsubseteq_\sigma \sigma_2$) but also that the weight bound be refined $I_1 \sqsubseteq I_2$.
It is easy to see that both $\sqsubseteq_\calA$ and $\sqsubseteq_\sigma$ are partial orders on types with the same underlying base type.

\subsection{Type System}

As for the interval-based semantics, we assume that every primitive operation $f : \real^n \to \real$ has an overapproximating interval abstraction $f^\iv : \mathbb{I}^n \to \mathbb{I}$ (cf. \cref{sec:intervalArith}).
Interval typing judgments have the form $\Gamma \vdash M : \calA$ where $\Gamma$ is a typing context mapping variables to types $\sigma$. They are given via the rules in \cref{fig:typeSystemSelection}.
Our system is sound in the following sense (which we here only state for first-order programs).

\begin{restatable}{theorem}{typeSystemSoundness}\label{thm:staticRes}
   Let $\vdash P : \typeReal$ be a simply-typed program.
   If $\vdash P : \text{\scalebox{0.7}{$\exType{\myint{a, b}}{\myint{c, d}}$}}$ and $\stdConf{P, \tr, 1}\to^* \stdConf{\lit{r}, \langle\rangle, w}$ for some $\tr \in \traces$ and $r, w \in \RR$, then $r \in \myint{a, b}$ and $w \in \myint{c, d}$.
\end{restatable}

Note that the bounds derived by our type system only refer to terminating executions, i.e.~they are partial correctness statements.
\Cref{thm:staticRes} formalises the intuition of an interval type, i.e.~every type derivation in our system bounds \emph{both} the returned value (in typical refinement-type fashion \cite{FreemanP91}) and the weight of this derivation.
Our type system also comes with a weak completeness statement: for each term, we can derive some bounds in our system.

\begin{restatable}{proposition}{typeSystemComp}
	Let $\vdash P : \alpha$ be a simply-typed program.
    There exists a weighted interval type $\calA$ such that $\vdash P : \calA$.\label{prop:compltnessType}
\end{restatable}

\subsection{Constraint-based Type Inference}

In this section, we briefly discuss the automated type \emph{inference} in our system, as needed in our tool GuBPI.
For space reasons, we restrict ourselves to an informal overview (see \ifFull{\cref{app:sec5}} for a full account).

Given a program $P$, we can derive the symbolic skeleton of a type derivation (the structure of which is determined by $P$), where each concrete interval is replaced by a placeholder variable.
The validity of a typing judgment within this skeleton can then be encoded as constraints.
Crucially, as we work in the fixed interval domain and the subtyping structure $\sqsubseteq_\calA$ is compositional, they are simple constraints over the placeholder variables in the abstract interval domain.
Solving the resulting constraints na\"{i}vely might not terminate since the interval abstract domain is not chain-complete.
Instead, we approximate the least fixpoint (where the fixpoint denotes a solution to the constraints) using \emph{widening}, a standard approach to ensure termination of static analysis on domains with infinite chains \cite{CousotCousot76,CousotC77}.
This is computationally much cheaper compared to, say, types with general first-order refinements where constraints are typically phrased as constrained Horn clauses (see e.g.~\cite{ChampionCKS20}).
This gain in efficiency is crucial to making our GuBPI tool practical.


\section{Symbolic Execution and GuBPI}
\label{sec:6linear}

In this section, we describe the overall structure of our tool GuBPI (\href{https://gubpi-tool.github.io/}{gubpi-tool.github.io}), which builds upon symbolic execution.
We also outline how the interval-based semantics can be accelerated for programs containing linear subexpressions.

\subsection{Symbolic Execution}

The starting point of our analysis is a \emph{symbolic exploration} of the term in question \cite{MakOPW21,GeldenhuysDV12,ChagantyNR13}.
For space reasons we only give an informal overview of the approach.
A detailed and formal discussion can be found in \ifFull{\cref{app:sec-symbolic}}.

The idea of symbolic execution is to treat outcomes of $\sample$ expressions fully symbolically: each $\sample$ evaluates to a fresh variable ($\alpha_1, \alpha_2, \dots$), called \emph{sample variable}.
The result of symbolic execution is thus a symbolic value: a term consisting of sample variables and delayed primitive function applications.
We postpone branching decisions and the weighting with $\score$ expressions because the value in question is symbolic.
During execution, we therefore explore both branches of a conditional and keep track of the (symbolic) conditions on the sample variables that need to hold in the current branch.
Similarly, we record the (symbolic) values of $\score$ expressions.
Formally, our symbolic execution operates on \defn{symbolic configurations} of the form $\psi = \symConf{\calM, n, {\Delta, \Xi}}$ where $\calM$ is a symbolic term containing sample variables instead of sample outcomes;
$n \in \natnum$ is a natural number used to obtain fresh sample variables;
$\Delta$ is a list of symbolic constraints of the form $\calV \bowtie r$, where $\calV$ is a symbolic value, $r \in \RR$ and ${\bowtie} \in {\{\leq, <, >, \geq\}}$, to keep track of the conditions for the current execution path;
and $\Xi$ is a set of values that records all symbolic values of $\score$ expressions encountered along the current path.
The symbolic reduction relation $\toSym$ includes the following key rules.
\begin{align*}
    &\symConf{\sample, n, \Delta, \Xi} \toSym \symConf{\alpha_{n+1}, n+1, \Delta, \Xi}\\
    &\symConf{\ifElse{\calV}{\calN}{\calP}), n, \Delta, \Xi} \toSym \symConf{\calN, n, \Delta \cup \{\calV \leq 0\}, \Xi}\\
    &\symConf{\ifElse{\calV}{\calN}{\calP}), n, \Delta, \Xi} \toSym \symConf{\calP, n, \Delta \cup \{\calV > 0\}, \Xi }\\
    &\symConf{\score(\calV), n, \Delta, \Xi} \toSym \symConf{\calV, n, \Delta \cup \{\calV \geq 0 \}, \Xi \cup \{\calV\}}
\end{align*}%
That is, we replace sample outcomes with fresh sample variables (first rule), explore both paths of a conditional (second and third rule), and record all score values (fourth rule).

\begin{example}\label{ex:pedestrianSymExec}
	Consider the symbolic execution of \cref{ex:pedestrian} where the first step moves the pedestrian towards their home (taking the right branch of $\oplus_{0.5}$) and the second step moves away from their home (the left branch of $\oplus_{0.5}$).
	We reach a configuration $(\calM, 5, \Delta, \Xi)$ where $\calM$ is
	\begin{center}
		\vspace{-2mm}
		\scalebox{0.99}{\parbox{\linewidth}{
				\begin{align*}
					\score \big(\pdf_{\Normal(1.1, 0.1)}\big(\alpha_2\! + \!\alpha_4\! +\! (\mu^\varphi_x. \calN) (3\alpha_1 \!-\! \alpha_2 \!+ \!\alpha_4)\big)\big) ; 3\alpha_1.
				\end{align*}
		}}
	\end{center}
	Here $\alpha_1$ is the initial sample for $\mathit{start}$; $\alpha_2, \alpha_4$ the two samples of $\mathit{step}$; and $\alpha_3, \alpha_5$ the samples involved in the $\oplus_{0.5}$ operator.
	The fixpoint $\mu^\varphi_x. \calN$ is already given in \cref{ex:pedestrianType}, $\Xi = \emptyset$ and $\Delta = \{3\alpha_1 > 0, \alpha_3 > \tfrac{1}{2}, 3\alpha_1 \!-\! \alpha_2 > 0, \alpha_5 \leq \tfrac{1}{2}\}$.
\end{example}

For a symbolic value $\calV$ using sample variables $\overline{\alpha} = \alpha_1, \allowbreak \dots, \allowbreak \alpha_n$ and $\tr \in [0,1]^n$, we write  $\calV[\tr/\overline \alpha] \in \real$ for the substitution of concrete values in $\tr$ for the sample variables.
Call a symbolic configuration of the form $\Psi = \symPath{\calV, n,\Delta, \Xi}$ (i.e.~a configuration that has reached a symbolic value $\calV$) a \defn{symbolic path}.
We write $\mathit{symPaths}(\psi)$ for the (countable) set of symbolic paths reached when evaluating from configuration $\psi$.
Given a symbolic path $\Psi = \symPath{\calV, n,\Delta, \Xi}$ and a set $U \in \Sigma_\real$, we define the denotation along $\Psi$, written $\llbracket \Psi \rrbracket (U)$, as\\
\scalebox{1}{\parbox{\linewidth}{
		\begin{align*}
			\int_{[0,1]^n} \!
			\big[\calV[\tr/\overline \alpha] \in U\big] \!\!\! \prod_{\calC \bowtie r \in \Delta} \!\!\!\! \big[\calC[\tr/\overline \alpha] \bowtie r\big]
			\prod_{\calW \in \Xi} \!\!\calW[\tr/\overline \alpha]
			 \D\tr,
		\end{align*}%
}}\\
i.e.~the integral of the product of the  $\score$ weights $\Xi$ over the traces of length $n$ where the result value is in $U$ and all the constraints $\Delta$ are satisfied.
We can recover the denotation of a program $P$ (as defined in \cref{sec:2background}) from all its symbolic paths starting from the configuration $\symConf{P, 0, \emptyset, \emptyset}$.

\begin{restatable}{theorem}{symbolicExec}\label{thm:relation-semantics-symbolic}
	Let $\vdash P : \typeReal$ be a program and $U \in \Sigma_\real$. Then
\begin{align*}
	\textstyle\measureSem P(U) = \sum_{\Psi \in \mathit{symPaths}\symConf{P, 0, \emptyset, \emptyset}} \; \llbracket \Psi \rrbracket (U).
\end{align*}
\end{restatable}

Analogously to interval SPCF (\cref{sec:3intervals}), we define \defn{symbolic interval terms} as symbolic terms that may contain intervals (and similarly for symbolic interval values, symbolic interval configurations, and symbolic interval paths).

\begin{algorithm}[!t]
	\begin{algorithmic}[1]
		\State \textbf{Input:} Program $\vdash P : \typeReal$, depth limit $D \in \NN$, and $I \in \iv$
		\State $\psi_{\mathsf{init}} := (P, 0, \emptyset, \emptyset); S := \{ (\psi_{\mathsf{init}}, 0) \}; T := \emptyset$
		\While{$\exists (\psi, \mathit{depth}) \in S$}
		\If{$\psi$ has terminated}
			\State $T := T \cup \{\psi\}; S := S \setminus \{(\psi, \mathit{depth})\}$
		\ElsIf{$\psi$ contains no fixpoints or $\mathit{depth} \leq D$}
			\State $S := S \setminus \{(\psi, \mathit{depth})\}$
			\For{$\psi'$ with $\psi \toSym \psi'$}
			\State $S := S \cup \{(\psi', \mathit{depth} + 1)\}$ \vspace{-1mm}
			\EndFor
		\Else
			\State $S := (S \setminus \{(\psi, \mathit{depth})\}) \cup  \{(\rmFix(\psi), \mathit{depth})\}$\label{line:approx}
		\EndIf
		\EndWhile
		\State \textbf{return} $\big[\sum_{\Psi \in T} \; \llbracket \Psi \rrbracket_\mathit{lb}(I), \sum_{\Psi \in T} \;  \llbracket \Psi \rrbracket_\mathit{ub}(I)\big]$\label{line:final}
	\end{algorithmic}
	\caption{Symbolic Analysis in GuBPI.}\label{alg:symRiBo}
\end{algorithm}

\subsection{GuBPI}

With symbolic execution at hand, we can outline the structure of our analysis tool GuBPI (sketched in \cref{alg:symRiBo}).
GuBPI's analysis begins with symbolic execution of the input term to accumulate a set of symbolic \emph{interval} paths $T$.
If a symbolic configuration $\psi$ has exceeded the user-defined depth limit $D$ and still contains a fixpoint, we overapproximate all paths that extend $\psi$ to ensure a finite set $T$.
We accomplish this by using the interval type system (\cref{sec:5interval-analysis}) to overapproximate all fixpoint subexpressions, thereby obtaining strongly normalizing terms (in line \ref{line:approx}).
Formally, given a symbolic configuration $\psi = \symConf{\calM, n, \Delta, \Xi}$ we derive a typing judgment for the term $\calM$ in the system from \cref{sec:5interval-analysis}.
Each first-order fixpoint subterm is thus given a (weightless) type of the form \scalebox{0.7}{$[a, b] \to \exType{[c, d]}{[e, f]}$}.
We replace this fixpoint with $\lambda \_.\big( \score([e, f])\mathbin{;} [c, d]\big)$.
We denote this operation on configurations by $\rmFix(\psi)$ (it extends to higher-order fixpoints as expected).
Note that $\rmFix(\psi)$ is a symbolic \emph{interval} configuration.

\begin{example}\label{ex:pedestrianSymbolicPath}
	Consider the symbolic configuration given in \cref{ex:pedestrianSymExec}.
	As in \cref{ex:pedestrianType} we infer the type of $\mu^\varphi_x. \calN$ to be
	\scalebox{0.65}{$[-1, 4] \to \exType{[0, \infty]}{[1, 1]}$}.
	The function $\rmFix$ replaces $\mu^\varphi_x. \calN$ with $\lambda \_. \score([1, 1]); [0, \infty]$.
	By evaluating  the resulting symbolic interval configuration further, we obtain the symbolic interval path  $\symPath{3\alpha_1, 5, \Delta, \Xi}$ where $\Delta$ is as in \cref{ex:pedestrianSymExec} and $\Xi = \{\pdf_{\Normal(1.1, 0.1)}(\alpha_2 + \alpha_4 + [0, \infty]) \}$.
	Note that, in general, the further evaluation of $\rmFix(\psi)$ can result in multiple symbolic interval paths.
\end{example}

Afterwards, we're left with a finite set $T$ of symbolic interval paths.
Due to the presence of intervals, we cannot define a denotation of such paths directly and instead define lower and upper bounds.
For a symbolic interval value $\calV$ that contains \emph{no} sample variables, we define $\ulcorner \calV \urcorner \subseteq \real$ as the set of all values that the term can evaluate to by replacing every interval $[a, b]$ with some value $r \in [a, b]$.
Given a symbolic interval path $\Psi = \symPath{\calV, n, \Delta, \Xi}$ and $U \in \Sigma_\RR$ we define $\llbracket \Psi \rrbracket_\mathit{lb}(U)$ by considering only those concrete traces that fulfill the constraints in $\Psi$ for \emph{all} concrete values in the intervals and take the infimum over all scoring expressions:\\
\scalebox{0.88}{\parbox{\linewidth}{
		\begin{align*}
			\int\!\!\!
			\big[\ulcorner\calV[\tr/\overline \alpha]\urcorner\subseteq U\big] \!\!\!\!\prod_{\calC \bowtie r \in \Delta} \!\!\!\!\big[\forall t \!\in\! \ulcorner\calC[\tr/\overline \alpha]\urcorner. t\bowtie r\big]  \!\!
			\prod_{\calW \in \Xi} \!\!\inf \ulcorner\calW[\tr/\overline \alpha]\urcorner
			\D\tr.
		\end{align*}
}}\\
Similarly, we define $\llbracket \Psi \rrbracket_\mathit{ub}(U)$ as\\
\scalebox{0.82}{\parbox{\linewidth}{
		\begin{align*}
			\int \!\!\!
			\big[\ulcorner\calV[\tr/\overline \alpha]\urcorner\cap  U \neq \emptyset \big] \!\!\!\!\prod_{\calC \bowtie r \in \Delta} \!\!\!\!\big[\exists t \!\in\! \ulcorner\calC[\tr/\overline \alpha]\urcorner.t \bowtie r\big]  \!\!
			\prod_{\calW \in \Xi} \!\!\sup \ulcorner\calW[\tr/\overline \alpha]\urcorner
			\D\tr.
		\end{align*}
}}\\
We note that, if $\Psi$ contains no intervals, $\llbracket \Psi \rrbracket$ is defined and we have $\llbracket \Psi \rrbracket_\mathit{lb} = \llbracket \Psi\rrbracket_\mathit{ub}  = \llbracket \Psi \rrbracket$.
We can now state the following theorem that formalises the observation that $\rmFix(\psi)$ soundly approximates all symbolic paths that result from $\psi$.

\begin{theorem}\label{thm:lowerUpperSymbolicDenotation}
	Let $\psi$ be a symbolic (interval-free) configuration and $U \in \Sigma_{\RR}$.
	Define $A = \mathit{symPaths}(\psi)$ as the (possibly infinite) set of all symbolic paths reached when evaluating $\psi$ and $B = \mathit{symPaths}(\rmFix(\psi))$ as the (finite) set of symbolic \emph{interval} paths reached when evaluating $\rmFix(\psi)$.
	Then
	\begin{align*}
		\textstyle\sum_{\Psi \in B} \; \llbracket \Psi \rrbracket_\mathit{lb}(U) \leq \sum_{\Psi \in A} \; \llbracket \Psi \rrbracket(U) \leq \sum_{\Psi \in B} \; \llbracket \Psi \rrbracket_\mathit{ub}(U).
	\end{align*}
\end{theorem}

The correctness of \cref{alg:symRiBo} is then a direct consequence of \cref{thm:relation-semantics-symbolic,thm:lowerUpperSymbolicDenotation}:

\begin{corollary}
	Let $T$ be the set of symbolic interval paths computed when at line \ref{line:final} of \cref{alg:symRiBo} and $U \in \Sigma_{\RR}$.
	Then
	\begin{align*}
		\textstyle\sum_{\Psi \in T} \; \llbracket \Psi \rrbracket_\mathit{lb}(U) \leq \llbracket P \rrbracket(U) \leq \sum_{\Psi \in T} \; \llbracket \Psi \rrbracket_\mathit{ub}(U).
	\end{align*}
\end{corollary}

What remains is to compute (lower bounds on) $\llbracket \Psi \rrbracket_\mathit{lb}(I)$ and (upper bounds on)  $\llbracket \Psi \rrbracket_\mathit{ub}(I)$ for a symbolic interval path $\Psi \in T$ and $I \in \iv$.
We first present the standard interval trace semantics (\cref{sec:standardInt}) and then a more efficient analysis for the case that $\Psi$ contains only linear functions (\cref{sec:linOpt}).

\subsection{Standard Interval Trace Semantics}
\label{sec:standardInt}

For any symbolic interval path $\Psi$, we can employ the semantics as introduced in \cref{sec:3intervals}.
Instead of applying the analysis to the entire program, we can restrict to the current path $\Psi$ (intuitively, by adding a $\score(0)$ to all other program paths).
The interval traces split the domain of each sample variable in $\Psi$ into intervals.
It is easy to see that for any compatible and exhaustive set of interval traces $\traceset$, we have $\lowerBound \Psi^\traceset (U) \leq \llbracket \Psi \rrbracket_\mathit{lb}(U)$ and $\llbracket \Psi \rrbracket_\mathit{ub}(U) \leq \upperBound \Psi^\traceset (U)$ for any $U \in \Sigma_\RR$ (see \cref{thm:lowerBoundsSound} and \ref{thm:upperBoundsSound}).
Applying the interval-based semantics on the level of symbolic interval paths maintains the attractive features, namely soundness and completeness (relative to the current path) of the semantics.
Note that the intervals occurring in $\Psi$ seamlessly integrate with our interval-based semantics.

\subsection{Linear Interval Trace Semantics}
\label{sec:linOpt}

In case the score values and the guards of all conditionals are linear, we can improve and speed up the interval-based semantics.

Assume all symbolic interval values appearing in $\Psi$ are interval linear functions of $\overline{\alpha}$ (i.e.~functions $\overline{\alpha} \mapsto \mathbf{w}^\intercal \overline{\alpha} +^\iv [a, b]$ for some $\mathbf{w} \in \RR^n$ and $[a, b] \in \iv$).
We assume, for now, that each symbolic value $\calW \in \Xi$ denotes an interval-free linear function (i.e.~a function $\overline{\alpha} \mapsto \mathbf{w}^\intercal \overline{\alpha} + r$).
Fix some interval $I \in \iv$.
We first note that both $\llbracket \Psi \rrbracket_\mathit{lb}(I)$ and $\llbracket \Psi \rrbracket_\mathit{ub}(I)$ are the integral of a polynomial over a convex polytope:
define\\
\scalebox{1}{\parbox{\linewidth}{
		\begin{align*}
		\mathfrak{P}_\mathit{lb} \! := \! \big\{\tr \in \real^n \mid \! \ulcorner\calV[\tr/\overline \alpha]\urcorner \!\subseteq \!I  \land \!\!\!\! \bigwedge_{\calC\bowtie r \in \Delta}\!\!\!\! \forall t \! \in \!\ulcorner\calC[\tr/\overline \alpha]\urcorner. t \bowtie r \big\}
	\end{align*}
}}\\
which is a polytope.%
\footnote{For example, if $\calC$ denotes the function $\overline{\alpha} \mapsto \mathbf{w}^\intercal \overline{\alpha} + [a, b]$ we can transform a constraint $\forall t \! \in \!\ulcorner\calC[\tr/\overline \alpha]\urcorner. t \leq r $ into the linear constraint $\mathbf{w}^\intercal \overline{\alpha} + b \leq r$.}
Then $\llbracket \Psi \rrbracket_\mathit{lb}(I)$ is the integral of the polynomial $\overline{\alpha} \mapsto \prod_{\calW \in \Xi} \calW$ over $\mathfrak{P}_\mathit{lb}$.
The definition of $\mathfrak{P}_\mathit{ub}$ (as the region of integration for $\llbracket \Psi \rrbracket_\mathit{ub}(I)$) is similar.
Such integrals can be computed exactly \cite{baldoni2011integrate}, e.g. with the LattE tool \cite{de2013software}.
Unfortunately, our experiments showed that this does not scale to interesting probabilistic programs.

Instead, we derive guaranteed bounds on the denotation by means of iterated volume computations.
This has the additional benefit that we can handle non-uniform samples and non-liner expressions in $\Xi$.
We follow an approach similar to that of the interval-based semantics in \cref{sec:4intervals-theory} but do not split/bound \emph{individual sample variables} and instead directly bound \emph{linear functions} over the sample variables.
Let $\Xi = \{\calW_1, \dots, \calW_k\}$.
We define a \defn{box} (by abuse of language) as an element $\boundbox = \langle [a_1,b_1], \dots, [a_k, b_k] \rangle$, where $[a_i, b_i]$ gives a bound on $\calW_i$.%
\footnote{Note the similarity to the interval trace  semantics. While the $i$th position in an interval trace bounds the value of the $i$th sample variable, the $i$th entry of a box bounds the $i$th score value.}
We define $\mathit{lb}(\boundbox) := \prod_{i=1}^{k} a_i$ and $\mathit{ub}(\boundbox) := \prod_{i=1}^{k} b_i$.
The box $\boundbox$ naturally defines a subset of $\mathfrak{P}_\mathit{lb}$ given by $\mathfrak{P}_\mathit{lb}^\boundbox = \big\{\tr \in \mathfrak{P}_\mathit{lb} \mid \bigwedge_{i=1}^k   \calW_i[\tr/\overline \alpha] \in [a_i, b_i] \big\}$.
Then $\mathfrak{P}_\mathit{lb}^\boundbox$ is again a polytope and we write $\volume(\mathfrak{P}_\mathit{lb}^\boundbox)$ for its volume.
The definition of $\mathfrak{P}_\mathit{ub}^\boundbox$ and $\volume(\mathfrak{P}_\mathit{ub}^\boundbox)$ is analogous.
As for interval traces, we call two boxes $\boundbox_1$, $\boundbox_2$ \emph{compatible} if the intervals are almost disjoint in at least one position.
A set of boxes $B$ is \emph{compatible} if its elements are pairwise compatible and \emph{exhaustive} if $\bigcup_{\boundbox \in B} \mathfrak{P}_\mathit{lb}^\boundbox = \mathfrak{P}_\mathit{lb}$ and $\bigcup_{\boundbox \in B} \mathfrak{P}_\mathit{ub}^\boundbox = \mathfrak{P}_\mathit{ub}$ (cf.~\cref{sec:boundsFromIntervalTraces}).

\begin{proposition}\label{prop:linearSplitting}
	Let $B$ be a compatible and exhaustive set of boxes.
	Then $\sum_{\boundbox \in B} \, \mathit{lb}(\boundbox) \cdot \volume\big(\mathfrak{P}_\mathit{lb}^\boundbox\big) \leq \llbracket \Psi \rrbracket_\mathit{lb}(I)$ and $ \llbracket \Psi \rrbracket_\mathit{ub}(I) \leq \sum_{\boundbox \in B} \, \mathit{ub}(\boundbox) \cdot \volume\big(\mathfrak{P}_\mathit{ub}^\boundbox\big)$.
\end{proposition}

As in the standard interval semantics, a finer partition into boxes yields more precise bounds.
While the volume computation involved in \cref{prop:linearSplitting} is expensive \cite{DyerF88}, the number of splits on the linear functions is much smaller than that needed in the standard interval-based semantics.
Our experiments empirically demonstrate that the direct splitting of linear functions (if applicable) is usually superior to the standard splitting.
In GuBPI we compute a set of exhaustive boxes by first computing a lower and upper bound on each $\calW_i \in \Xi$ over $\mathfrak{P}_\mathit{lb}$ (or $\mathfrak{P}_\mathit{ub}$) by solving a linear program (LP) and splitting the resulting range in evenly sized chunks.

\paragraph{Beyond uniform samples and linear scores}

We can extend our linear optimization to non-uniform samples and arbitrary symbolic values in $\Xi$.
We accomplish the former by \emph{combining} the optimised semantics (where we bound linear expressions) with the standard interval-trace semantics (where we bound individual sample variables).
For the latter, we can identify linear sub-expressions of the expressions in $\Xi$, use boxes to bound each linear sub-expression, and use interval arithmetic to infer bounds on the entire expression from bounds on its linear sub-expressions.
More details can be found in \ifFull{\cref{sec:ext-lin-splitting}}.

\begin{example}
	Consider the path $\Psi = \symPath{3\alpha_1, 5, \Delta, \Xi}$ derived in \cref{ex:pedestrianSymbolicPath}.
	We use 1-dimensional boxes to bound $\alpha_2 + \alpha_4$ (the single linear sub-expression of the symbolic values in $\Xi$).
	To obtain a lower bound on $\llbracket \Psi \rrbracket_\mathit{lb}(I)$, we sum over all boxes $\boundbox = \langle [a_1, b_1] \rangle$ and take the product of $\volume\big(\mathfrak{P}_\mathit{lb}^\boundbox\big)$ with the lower interval bound of $\pdf_{\Normal(1.1, 0.1)}([a_1,b_1] + [0, \infty])$ (evaluated in interval arithmetic).
	Analogously, for the upper bound we take the product of $\volume\big(\mathfrak{P}_\mathit{ub}^\boundbox\big)$ with the upper interval bound of $\pdf_{\Normal(1.1, 0.1)}([a_1,b_1]+ [0, \infty])$.
\end{example}


\section{Practical Evaluation}
\label{sec:7practical-evaluation}

\begin{table}[!t]
	\caption{Evaluation on selected benchmarks from \cite{SankaranarayananCG13}. We give the times (in seconds) and bounds computed by \cite{SankaranarayananCG13} and GuBPI. Details on the exact queries (the Q column) can be found in \ifFull{\cref{fig:pldi13resultsApp} in the appendix}. }
	\label{fig:pldi13results}
	\vspace{-3mm}
	\setlength\dashlinedash{0.5pt}
	\setlength\dashlinegap{1pt}
	\setlength\arrayrulewidth{0.5pt}
	\centering
	\footnotesize
	\def\arraystretch{1.2}
	\begin{tabular}{lccccc}
		\toprule
		\multicolumn{2}{@{}c@{\hspace{0mm}}}{} &\multicolumn{2}{@{}c@{}}{\textbf{Tool from \cite{SankaranarayananCG13}}} & \multicolumn{2}{@{}c@{\hspace{0mm}}}{\textbf{GuBPI}}\\
		\cmidrule[0.7pt](l{1mm}r{1mm}){3-4}
		\cmidrule[0.7pt](l{1mm}){5-6}
		 \textbf{Program}& \textbf{Q}& $\boldsymbol{t}$ & \textbf{Result} & $\boldsymbol{t}$ & \textbf{Result} \\
		\cmidrule[0.7pt](r{1mm}){1-1}
		\cmidrule[0.7pt](l{1mm}r{1mm}){2-2}
		\cmidrule[0.7pt](l{1mm}r{1mm}){3-3}
		\cmidrule[0.7pt](l{1mm}r{1mm}){4-4}
		\cmidrule[0.7pt](l{1mm}r{1mm}){5-5}
		\cmidrule[0.7pt](l{1mm}r{0mm}){6-6}
	 	tug-of-war & Q1  & 1.29  & $[0.6126,   0.6227]$  & 0.72  & $[0.6134, 0.6135]$  \\
		\hdashline
		tug-of-war & Q2 & 1.09  & $[0.5973,   0.6266]$  & 0.79  & $[0.6134, 0.6135]$  \\
		\hdashline
		beauquier-3 & Q1 & 1.15 & $[0.5000,   0.5261]$ & 22.5 & $[0.4999, 0.5001]$  \\
		\hdashline
		ex-book-s & Q1 & 8.48 & $[0.6633,   0.7234]$ & 6.52 & $[0.7417, 0.7418]$   \\
		\hdashline
		ex-book-s \!\!\!& Q2$^\star$ \!\!\!  & 10.3 & $[0.3365,   0.3848]$ & 8.01 & $[0.4137, 0.4138]$   \\
		\hdashline
		ex-cart & Q1  & 2.41 & $[0.8980,   1.1573]$ & 67.3 & $[0.9999, 1.0001]$   \\
		\hdashline
		ex-cart & Q2    & 2.40 & $[0.8897,   1.1573]$ & 68.5 & $[0.9999, 1.0001]$  \\
		\hdashline
		ex-cart & Q3  & 0.15 & $[0.0000,   0.1150]$ & 67.4 & $[0.0000, 0.0001]$   \\
		\hdashline
		ex-ckd-epi-s\!\!\! & Q1$^\star$ \!\!\! & 0.15 & $[0.5515,   0.5632]$ & 0.86 & $[0.0003, 0.0004]$   \\
		\hdashline
		ex-ckd-epi-s\!\!\! & Q2$^\star$  \!\!\! & 0.08 & $[0.3019,   0.3149]$ & 0.84 & $[0.0003, 0.0004]$   \\
		\hdashline
		ex-fig6 & Q1 & 1.31 & $[0.1619,   0.7956]$ & 21.2 & $[0.1899, 0.1903]$   \\
		\hdashline
		ex-fig6 & Q2  & 1.80 & $[0.2916,   1.0571]$ & 21.4 & $[0.3705, 0.3720]$   \\
		\hdashline
		ex-fig6 & Q3  & 1.51 & $[0.4314,   2.0155]$ & 24.7 & $[0.7438, 0.7668]$   \\
		\hdashline
		ex-fig6 & Q4  & 3.96 & $[0.4400,   3.0956]$ & 27.4 & $[0.8682, 0.9666]$   \\
		\hdashline
		ex-fig7 & Q1  &  0.04& $[0.9921, 1.0000]$ & 0.18 & $[0.9980, 0.9981]$  \\
		\hdashline
		example4 & Q1 & 0.02 & $[0.1910,   0.1966]$ & 0.31 & $[0.1918, 0.1919]$   \\
		\hdashline
		example5 & Q1  & 0.06 & $[0.4478,   0.4708]$ & 0.27 & $[0.4540, 0.4541]$   \\
		\hdashline
		herman-3 & Q1  & 0.47 & $[0.3750,   0.4091]$ & 124 & $[0.3749, 0.3751]$ \\
		\bottomrule
	\end{tabular}

\end{table}

We have implemented our semantics in the prototype GuBPI (\href{https://gubpi-tool.github.io/}{gubpi-tool.github.io}), written in F\#.
In cases where we apply the linear optimisation of our semantics, we use Vinci \cite{BuelerEF00} to discharge volume computations of convex polytopes.
We set out to answer the following questions:
\begin{enumerate}
    \item How does GuBPI perform on instances that could already be solved (e.g.~by PSI \cite{GehrMV16})?
    \item Is GuBPI able to infer useful bounds on recursive programs that could not be handled rigorously before?
\end{enumerate}

\subsection{Probability Estimation}
\label{sec:evalEsti}

We collected a suite of 18 benchmarks from \cite{SankaranarayananCG13}.
Each benchmark consists of a program $P$ and a query $\phi$ over the variables of $P$.
We bound the probability of the event described by $\phi$ using the tool from \cite{SankaranarayananCG13} and GuBPI (\cref{fig:pldi13results}).
While our tool is generally slower than the one in \cite{SankaranarayananCG13}, the completion times are still reasonable.
Moreover, in many cases, the bounds returned by GuBPI are tighter than those of \cite{SankaranarayananCG13}.
In addition, for benchmarks marked with a $\star$, the two pairs of bounds contradict each other.%
\footnote{A stochastic simulation using $10^6$ samples in Anglican \cite{TolpinMW15} yielded results that fall within GuBPI's bounds but violate those computed by \cite{SankaranarayananCG13}.}
We should also remark that GuBPI cannot handle all benchmarks proposed in \cite{SankaranarayananCG13} because the heavy use of conditionals causes our precise symbolic analysis to suffer from the well-documented path explosion problem \cite{BoonstoppelCE08,Godefroid07,CadarGPDE08}.
Perhaps unsurprisingly, \cite{SankaranarayananCG13} can handle such examples much better, as one of their core contributions is a stochastic method to reduce the number of paths considered (see \cref{sec:related}).
Also note that \cite{SankaranarayananCG13} is restricted to uniform samples, linear guards and score-free programs, whereas we tackle a much more general problem.

\begin{table}
	\caption{Probabilistic programs with discrete domains from PSI \cite{GehrMV16}. The times for PSI and GuBPI are given in seconds. }\label{tab:psi-discrete}
	\vspace{-3mm}
	\small
	\def\arraystretch{1.1}
	\begin{tabular}{lcclcc}
		\cmidrule[1pt](r{1.5mm}){1-3}
		\cmidrule[1pt](l{0mm}r{1mm}){4-6}
		\textbf{Instance} & $\boldsymbol{t}_\mathit{PSI}$ & $\boldsymbol{t}_\mathit{GuBPI}$ & \textbf{Instance} & $\boldsymbol{t}_\mathit{PSI}$ & $\boldsymbol{t}_\mathit{GuBPI}$  \\
		\cmidrule[0.7pt](r{1mm}){1-1}
		\cmidrule[0.7pt](l{0.5mm}r{1mm}){2-2}
		\cmidrule[0.7pt](l{0.5mm}r{1.5mm}){3-3}
		\cmidrule[0.7pt](l{0mm}r{1mm}){4-4}
		\cmidrule[0.7pt](l{0.5mm}r{1mm}){5-5}
		\cmidrule[0.7pt](l{0.5mm}r{1mm}){6-6}
		burglarAlarm \!\!\!\!& 0.06 & 0.21 & coins & 0.04 & 0.18   \\
		twoCoins & 0.04 & 0.21 & ev-model1 & 0.04 & 0.21 \\
		grass & 0.06 & 0.37 & ev-model2 & 0.04 & 0.20 \\
		noisyOr & 0.14 & 0.72 & murderMystery  \!\!\!\! & 0.04 & 0.19 \\
		bertrand & 0.04 & 0.22 & coinBiasSmall  & 0.13 & 1.92 \\
		coinPattern & 0.04 & 0.19 & gossip  & 0.08 & 0.24 \\
		\cmidrule[1pt](r{1.5mm}){1-3}
		\cmidrule[1pt](l{0mm}r{1mm}){4-6}
	\end{tabular}
\end{table}

\begin{figure}
	\begin{subfigure}[t]{0.23\textwidth}
		\centering
		\includegraphics[width=\textwidth]{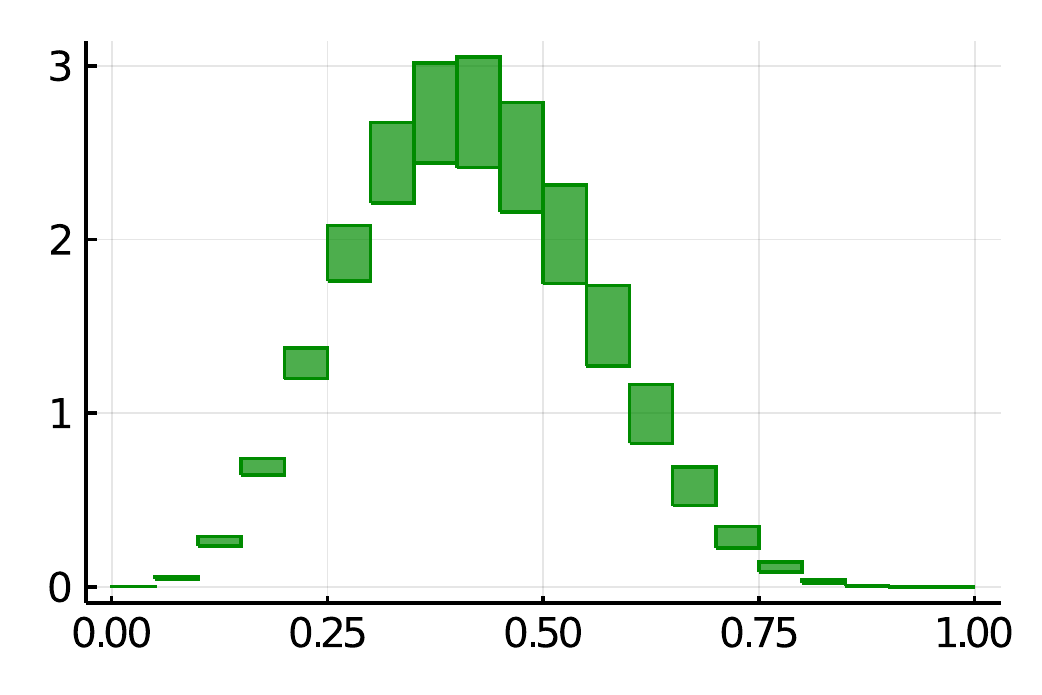}
		\vspace{-6mm}
		\subcaption{\texttt{coinBias} example from \cite{GehrMV16}. The program samples a beta prior on the bias of a coin and observes repeated coin flips (26 seconds).}
	\end{subfigure}\hfill%
	\begin{subfigure}[t]{0.23\textwidth}
		\centering
		\includegraphics[width=\textwidth]{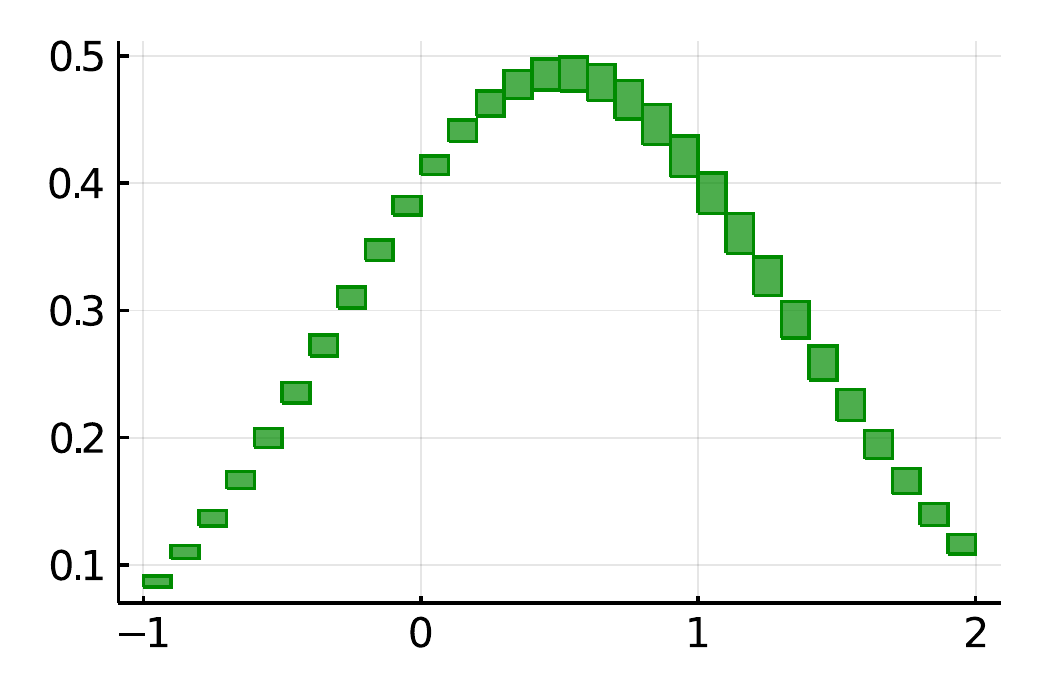}
		\vspace{-6mm}
		\subcaption{\texttt{max} example from \cite{GehrMV16}. The program compute the maximum of two i.i.d. normal samples (31 seconds).}
	\end{subfigure}

	\begin{subfigure}[t]{0.23\textwidth}
		\centering
		\includegraphics[width=\textwidth]{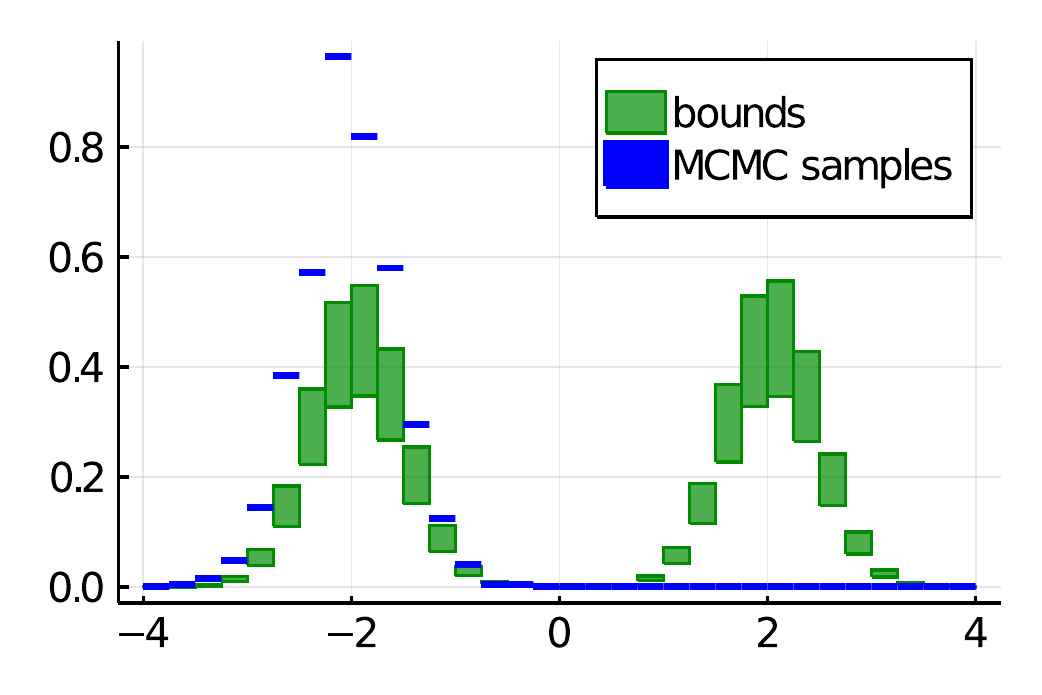}
		\vspace{-6mm}
		\subcaption{Binary Gaussian Mixture Model from \cite{ZhouGKRYW19} (39 seconds). MCMC methods usually find only one mode.}\label{fig:gmm}
	\end{subfigure}\hfill%
	\begin{subfigure}[t]{0.23\textwidth}
		\centering
		\includegraphics[width=\textwidth]{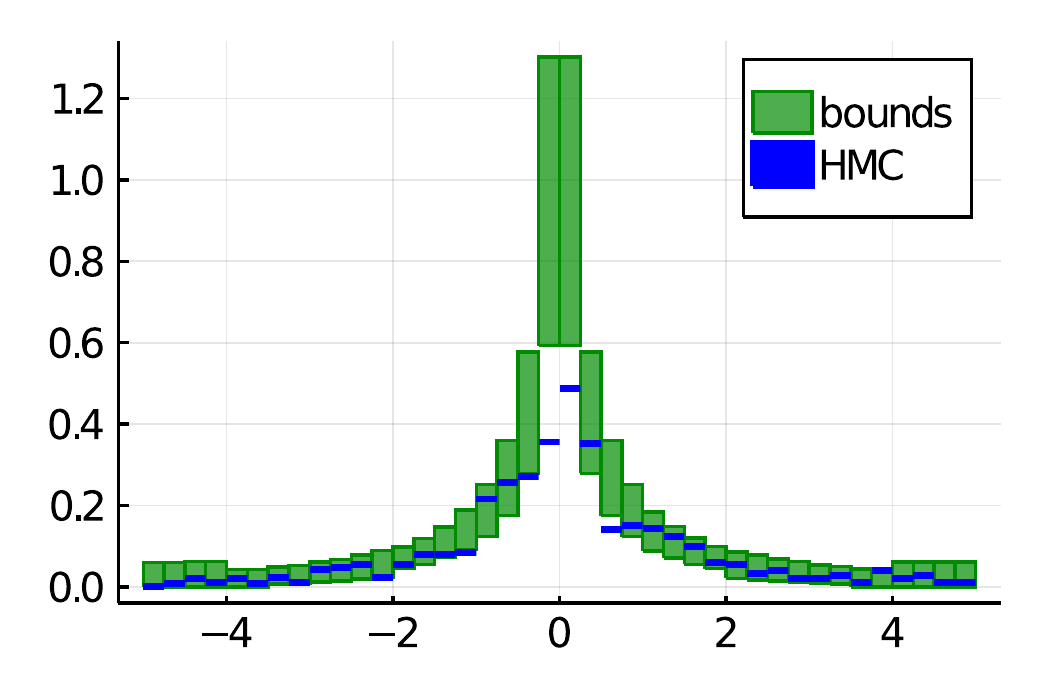}
		\vspace{-6mm}
		\subcaption{Neal's funnel from \cite{Neal03,GorinovaMH20} (2.8 seconds). HMC misses some probability mass around 0.}
	\end{subfigure}%
	\caption{Guaranteed Bounds computed by GuBPI for a selection of non-recursive models from \cite{GehrMV16,GehrSV20,ZhouGKRYW19, Neal03}.  }\label{fig:psi-examples}
\end{figure}

\begin{figure*}[!t]
	\begin{subfigure}[t]{0.30\textwidth}
		\centering
		\includegraphics[width=\textwidth]{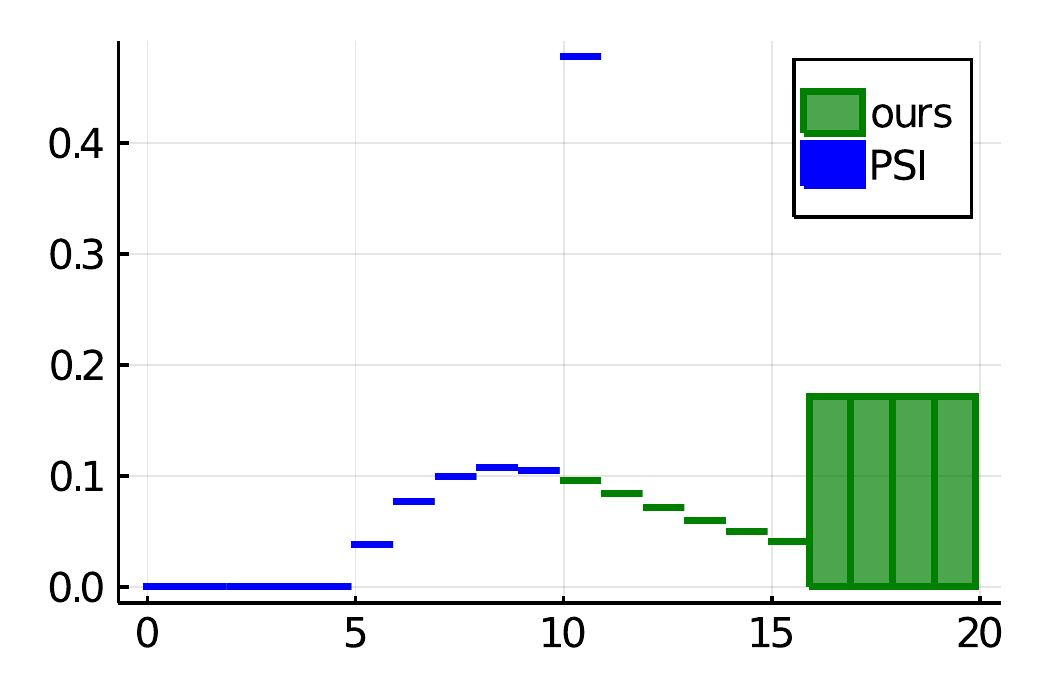}
		\vspace{-7mm}
		\subcaption{\texttt{cav-example-7}. Program taken from the PSI repository. PSI bounds the depth resulting in a spike at 10, whereas GuBPI can compute bounds on the denotation of the unbounded program (112 seconds).}\label{fig:psi-bounded-1}
	\end{subfigure}\hfill%
	\begin{subfigure}[t]{0.30\textwidth}
		\centering
		\includegraphics[width=\textwidth]{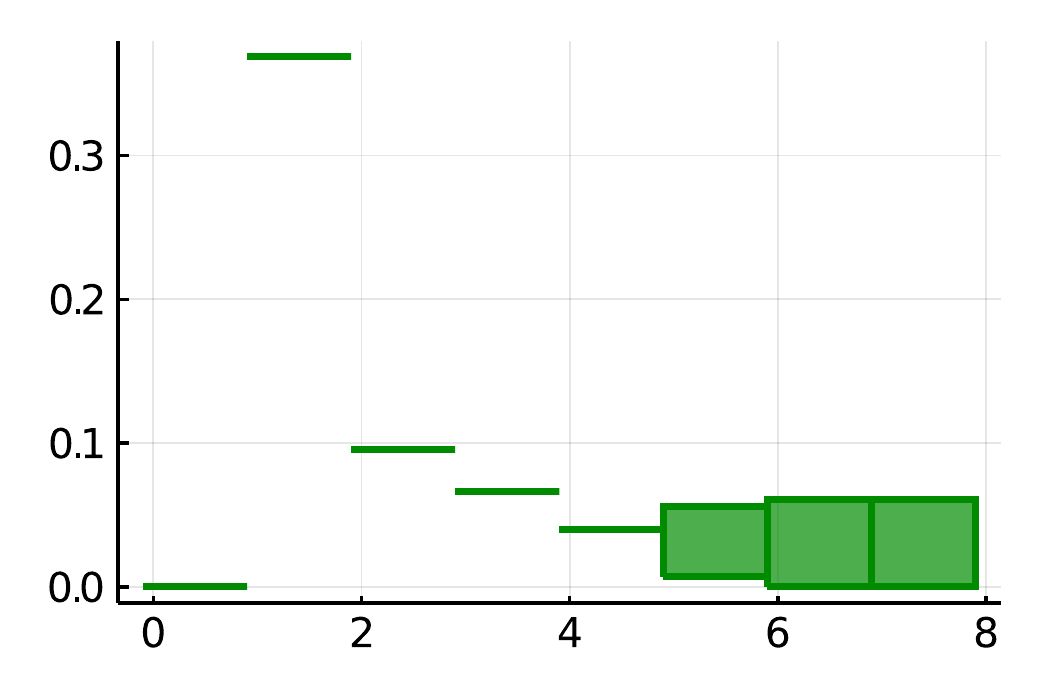}
		\vspace{-7mm}
		\subcaption{\texttt{cav-example-5}. Program taken from the PSI repository. PSI cannot handle this program  due to the unbounded loops (192 seconds).}\label{fig:psi-bounded-2}
	\end{subfigure}\hfill%
	\begin{subfigure}[t]{0.30\textwidth}
		\centering
		\includegraphics[width=\textwidth]{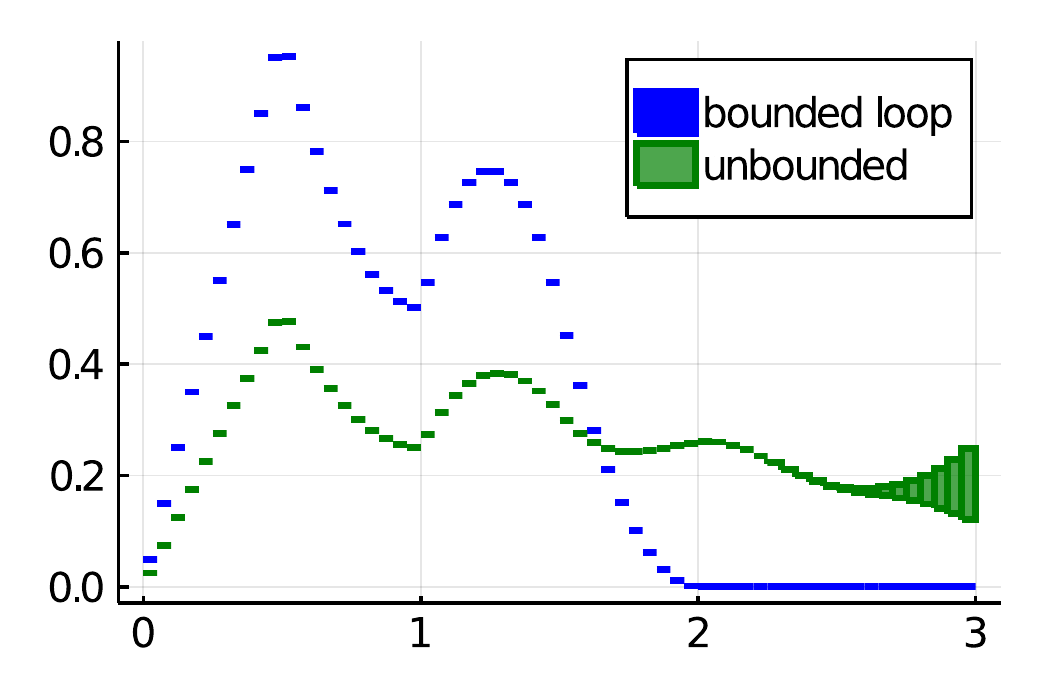}
		\vspace{-7mm}
		\subcaption{\texttt{add\_uniform\_with\_counter\_large}. Program taken from the PSI repository.
			GuBPI can handle the unbounded loop, whereas PSI unrolls to a fixed depth (21 seconds).}\label{fig:adduniforms}
	\end{subfigure}

	\begin{subfigure}[t]{0.30\textwidth}
		\centering
		\includegraphics[width=\textwidth]{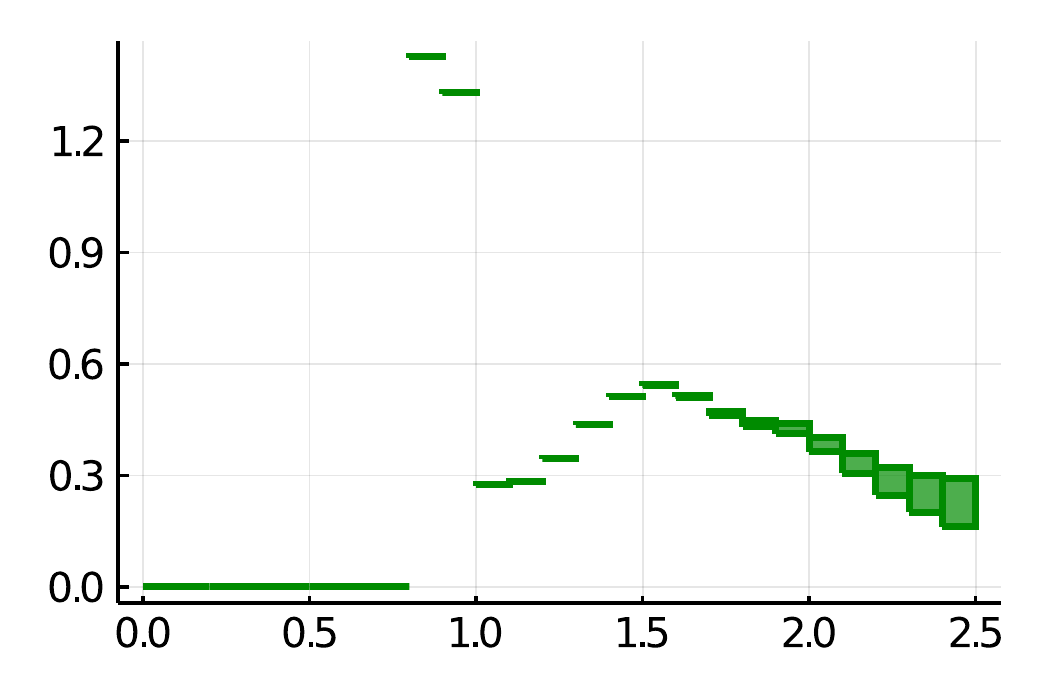}
		\vspace{-7mm}
		\subcaption{\texttt{random-box-walk}. The program models the cumulative distance traveled by a biased random walk.
			If a uniformly sampled step $s$ has length less than $\tfrac{1}{2}$, we move $s$ to the left, otherwise $s$ to the right.
			The walk stops when it crosses a threshold (167 seconds).}\label{fig:extra-recursive1}
	\end{subfigure}\hfill%
	\begin{subfigure}[t]{0.30\textwidth}
		\centering
		\includegraphics[width=\textwidth]{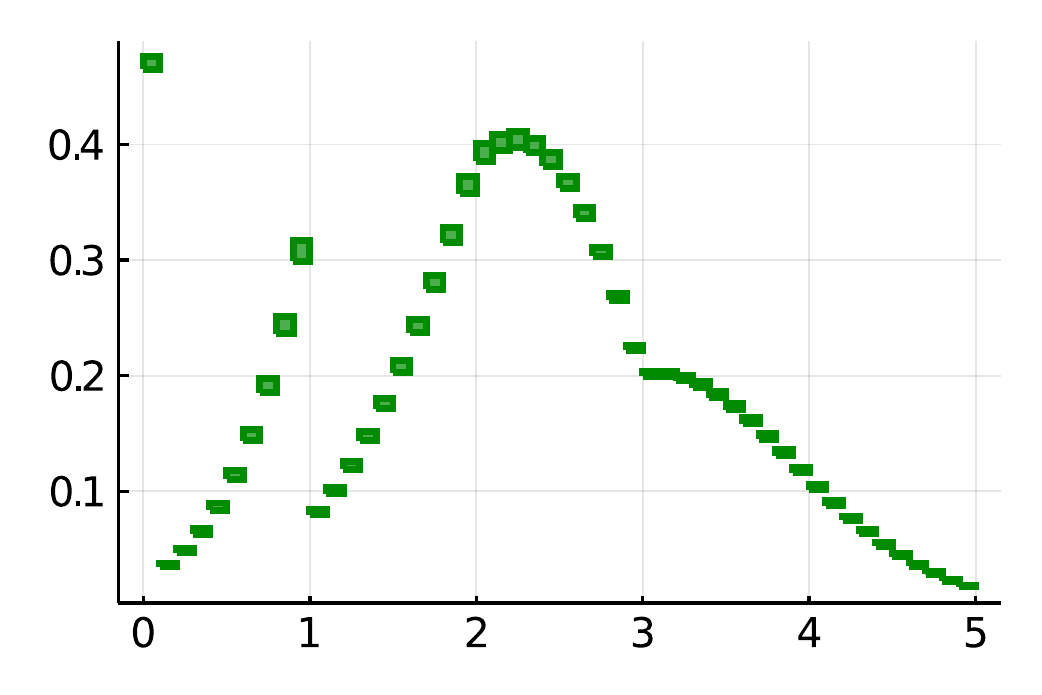}
		\vspace{-7mm}
		\subcaption{\texttt{growing-walk}. The program models a geometric random walk where (with increasing distance) the step size of the walk is increased.
			The cumulative distance is observed from a normal distribution centered at 3 (67 seconds).}\label{fig:extra-recursive2}
	\end{subfigure}\hfill%
	\begin{subfigure}[t]{0.30\textwidth}
		\centering
		\includegraphics[width=\textwidth]{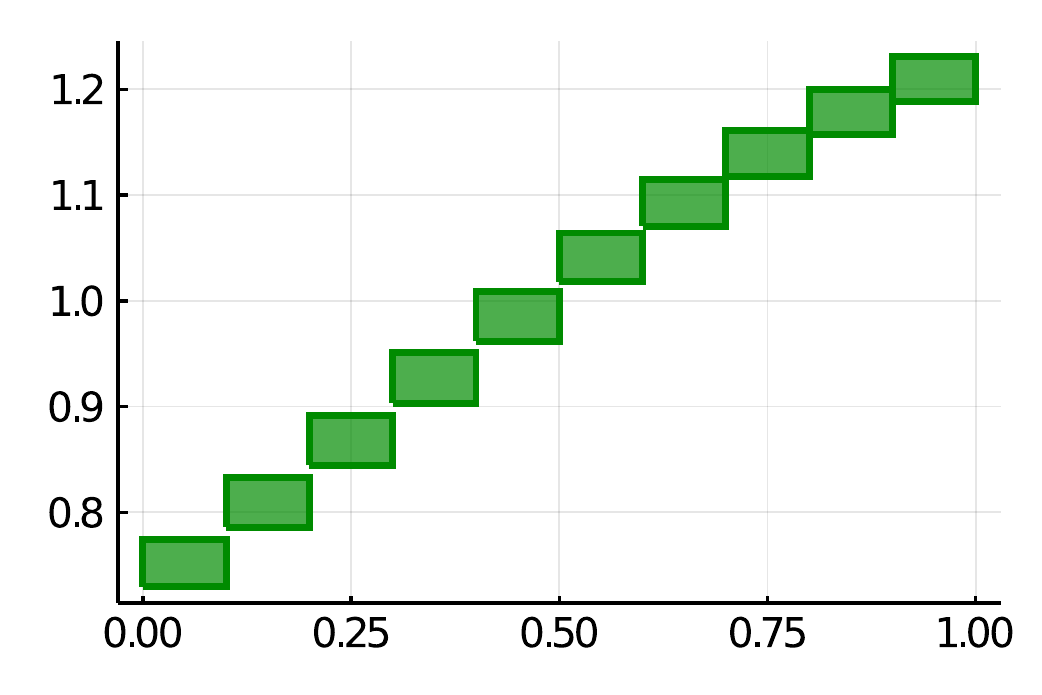}
		\vspace{-7mm}
		\subcaption{\texttt{param-estimation-recursive}. We sample a uniform prior $p$ and (in each step) travel to the left with probability $p$ and to the right with probability $(1-p)$. We observe the walk to come to a halt at location $1$  (observed from a normal) and wish to find the posterior on $p$ (162 seconds).
		}\label{fig:extra-recursive3}
	\end{subfigure}
	\vspace{-1mm}
	\caption{Guaranteed bounds computed by GuBPI for a selection of recursive models.}
	\label{fig:recursive-models}
\end{figure*}

\subsection{Exact Inference}

To evaluate our tool on instances that can be solved exactly, we compared it with PSI \cite{GehrMV16,GehrSV20}, a symbolic solver which can, in certain cases, compute a closed-form solution of the posterior.
We note that whenever exact inference is possible, exact solutions will always be superior to mere bounds and, due to the overhead of our semantics, will often be found faster.
Because of the different output formats (i.e.~exact results vs.~bounds), a direct comparison between exact methods and GuBPI is challenging.
As a consistency check, we collected benchmarks from the PSI repository where the output domain is finite and GuBPI can therefore compute \emph{exact} results (tight bounds).
They agree with PSI in all cases, which includes 8 of the 21 benchmarks from \cite{GehrMV16}.
We report the computation times in \cref{tab:psi-discrete}.

We then considered examples where GuBPI computes non-tight bounds.
For space reasons, we can only include a selection of examples in this paper.
The bounds computed by GuBPI and a short description of each example are shown in \cref{fig:psi-examples}.
We can see that, despite the relatively loose bounds, they are still useful and provide the user with a rough---and most importantly, \emph{guaranteed to be correct}---idea of the denotation.

The success of exact solvers such as PSI depends on the underlying symbolic solver (and the optimisations implemented).
Consequently, there are instances where the symbolic solver cannot compute a closed-form (integral-free) solution.
Conversely, while our method is (theoretically) applicable to a very broad class of programs, there exist programs where the symbolic solver finds solutions but the analysis in GuBPI becomes infeasible due to the large number of interval traces required.

\subsection{Recursive Models}

We also evaluated our tool on complex models that \emph{cannot} be handled by any of the existing methods.
For space reasons, we only give an overview of some examples.
Unexpectedly, we found recursive models in the PSI repository:
there are examples that are created by unrolling loops to a fixed depth.
This fixed unrolling changes the posterior of the model.
Using our method we can handle those examples \emph{without} bounding the loop.
Three such examples are shown in \cref{fig:psi-bounded-1,fig:psi-bounded-2,fig:adduniforms}. In \cref{fig:psi-bounded-1}, PSI bounds the iterations resulting in a spike at 10 (the unrolling bound).
For \cref{fig:psi-bounded-2}, PSI does not provide any solution whereas GuBPI provides useful bounds.
For \cref{fig:adduniforms}, PSI bounds the loop to compute results (displayed in blue) whereas GuBPI computes the green bounds on the unbounded program.
It is obvious that the bounds differ significantly, highlighting the impact that unrolling to a fixed depth can have on the denotation.
This again strengthens the claim that rigorous methods that can handle unbounded loops/recursion are needed.
There also exist unbounded discrete examples where PSI computes results for the bounded version that differ from the denotation of the unbounded program.
\cref{fig:extra-recursive1,fig:extra-recursive2,fig:extra-recursive3} depict further recursive examples (alongside a small description).

Lastly, as a \emph{very} challenging example, we consider the pedestrian example (\cref{ex:pedestrian}) again.
The bounds computed by GuBPI are given in \cref{fig:pedBounds} together with the two stochastic results from \cref{fig:pedestrian-stochastic}.
The bounds are clearly precise enough to rule out the HMC samples.
Since this example has infinite expected running time, it is very challenging and GuBPI takes about 1.5h (84 min).%
\footnote{While the running time seems high, we note that Pyro HMC took about an hour to generate $10^4$ samples and produce the (wrong) histogram. Diagnostic methods like simulation-based calibration took even longer (>30h) and delivered inconclusive results (see \cref{sec:sbc} for details).}
In fact, guaranteed bounds are the only method to recognise the wrong samples with certainty (see the next section for statistical methods).

\begin{figure}
    \centering
    \includegraphics[width=\columnwidth]{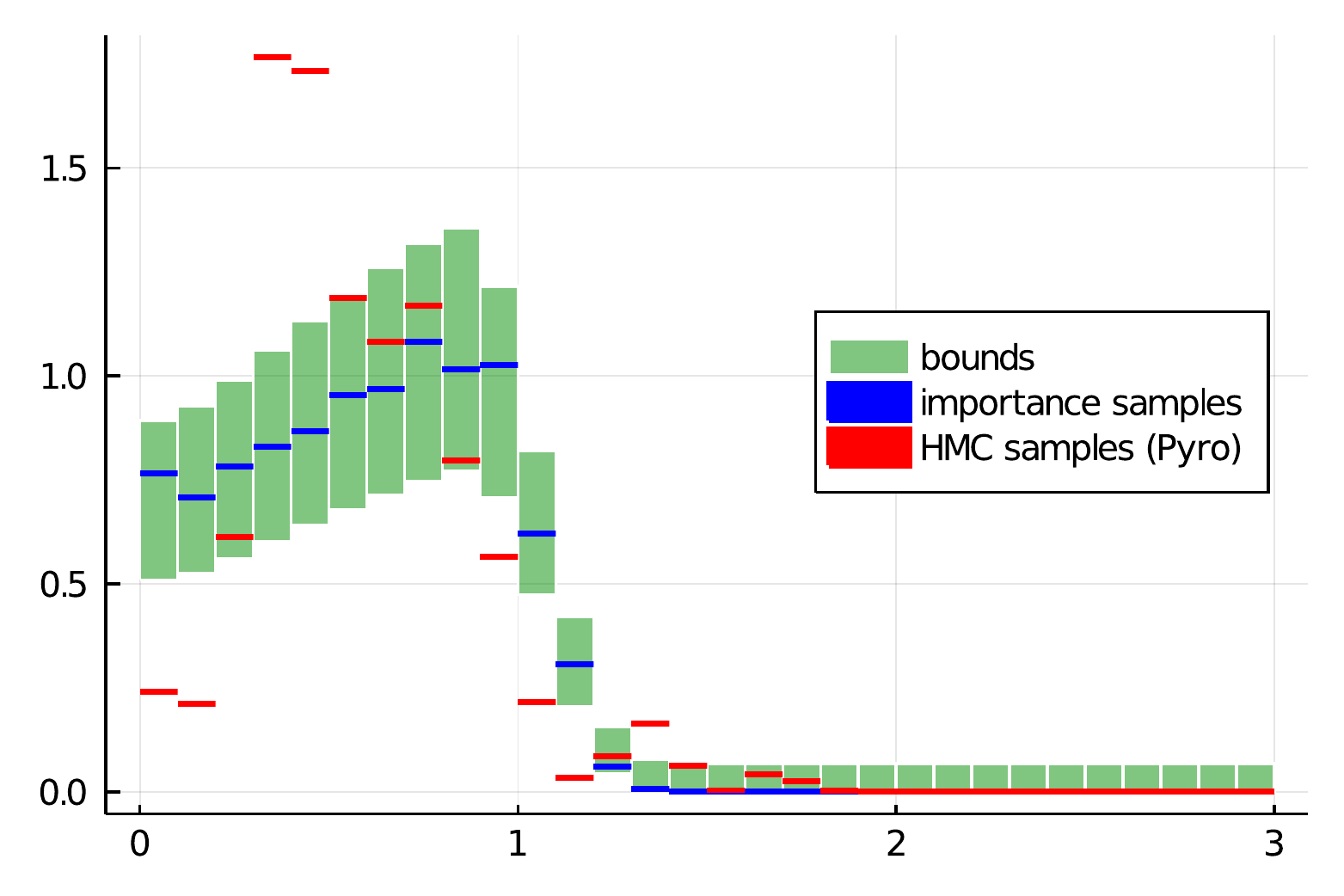}
    \vspace{-8mm}
    \caption{Bounds for the pedestrian example (\cref{ex:pedestrian}). }\label{fig:pedBounds}
\end{figure}

\begin{table}
	\caption{Running times of GuBPI and SBC for (Pyro's) HMC. Times are given in seconds (s) and hours (h).}\label{tab:SBC}
	\vspace{-3mm}
	\def\arraystretch{1.1}
	\begin{tabular}{lcc}
		\toprule
		\textbf{Instance} & $\boldsymbol{t}_\mathit{GuBPI}$  & $\boldsymbol{t}_\mathit{SBC}$  \\

		\cmidrule[0.7pt](r{1mm}){1-1}
		\cmidrule[0.7pt](l{0.5mm}r{1mm}){2-2}
		\cmidrule[0.7pt](l{0.5mm}){3-3}

		Binary GMM (1-dimensional) (\cref{fig:gmm}) & 39s & 1h \\
		Binary GMM (2-dimensional)  & 4h  & 1.5h   \\
		Pedestrian Example (\cref{fig:pedBounds})& 1.5h & >300h  \\
		\bottomrule
	\end{tabular}
\end{table}

\subsection{Comparison with Statistical Validation}
\label{sec:sbc}

A general approach to validating inference algorithms for a generative Bayesian model is \emph{simulation-based calibration} (SBC) \cite{SBC,cook2006validation}.
SBC draws a sample $\theta$ from the prior distribution of the parameters, generates data $y$ for these parameters, and runs the inference algorithm to produce posterior samples $\theta_1, \dots, \theta_L$ given $y$.
If the posterior samples follow the true posterior distribution, the rank statistic of the prior sample $\theta$ relative to the posterior samples will be uniformly distributed.
If the empirical distribution of the rank statistic after many such simulations is non-uniform, this indicates a problem with the inference.
While SBC is very general, it is computationally expensive because it performs inference in every simulation.
Moreover, as SBC is a stochastic validation approach, any fixed number of samples may fail to diagnose inference errors that only occur on a very low probability region.

We compare the running times of GuBPI and SBC for three examples where Pyro's HMC yields wrong results (\cref{tab:SBC}).
Running SBC on the pedestrian example (with a reduced sample size and using the parameters recommended in \cite{SBC}) took 32 hours and was still inconclusive because of strong autocorrelation.
Reducing the latter via thinning requires more samples, and would increase the running time to >300 hours.
Similarly, GuBPI diagnoses the problem with the mixture model in \cref{fig:gmm} in significantly less time than SBC.
However, for higher-dimensional versions of this mixture model, SBC clearly outperforms GuBPI.
We give a more detailed discussion of SBC for these examples in \ifFull{\cref{sec:sbc-experiments}}.

\subsection{Limitations and Future Improvements}

The theoretical foundations of our interval-based semantics ensure that GuBPI is applicable to a very broad class of programs (cf. \cref{sec:4intervals-theory}).
In practice, as usual for exact methods, GuBPI does not handle all examples equally well.

Firstly, as we already saw in \cref{sec:evalEsti}, the symbolic execution---which forms the entry point of the analysis---suffers from path explosion.
On some extreme loop/recursion-free programs (such as example-ckd-epi from \cite{SankaranarayananCG13}), our tool cannot compute all (finitely many) symbolic paths in reasonable time, let alone analyse them in our semantics.
Extending the approach from \cite{SankaranarayananCG13}, to sample representative program paths (in the presence of conditioning), is an interesting future direction that we can combine with the rigorous analysis provided by our interval type system.

Secondly, our interval-based semantics imposes bounds on each sampled variable and thus scales exponentially with the dimension of the model; this is amplified in the case where the optimised semantics (\cref{sec:linOpt}) is not applicable.
It would be interesting to explore whether this can be alleviated using different trace splitting techniques.

Lastly, the bounds inferred by our interval type system take the form of a single interval with no information about the exact distribution on that interval.
For example, the most precise bound derivable for the term $\mu^\varphi_x. x \oplus \big[\varphi (x + \sample) \oplus \varphi (x - \sample)\big]$ is \scalebox{0.6}{$[a, b]\to \exType{[-\infty, \infty]}{[1, 1]}$} for any $a, b$.
After unrolling to a fixed depth, the approximation of the paths not terminating within the fixed depth is therefore imprecise.
For future work, it would be interesting to improve the bounds in our type system to provide more information about the distribution by means of rigorous approximations of the denotation of the fixpoint in question (i.e.~a probabilistic summary of the fixpoint \cite{WangHR18,Muller-OlmS04,PodelskiSW05}).


\section{Related Work}
\label{sec:related}

\paragraph{Interval trace semantics and Interval SPCF}
Our interval trace semantics to compute bounds on the denotation is similar to the semantics introduced by \citet{BeutnerO21}, who study an interval approximation to obtain \emph{lower} bounds on the termination probability.
By contrast, we study the more challenging problem of bounding the program denotation which requires us to track the weight of an execution, and to prove that the denotation approximates a Lebesgue integral, which requires novel proof ideas.
Moreover, whereas the termination probability of a program is always upper bounded by $1$, here we derive both lower and \emph{upper} bounds.

\paragraph{Probability estimation}

\citet{SankaranarayananCG13} introduced a static analysis framework to infer bounds on a class of definable events in (\emph{score-free}) probabilistic programs.
The idea of their approach is that if we find a finite set $\traceset$ of symbolic traces with cumulative probability at least $1-c$, and a given event $\phi$ occurs with probability at most $b$ on the traces in $\traceset$, then $\phi$ occurs with probability at most $b + c$ on the entire program.
In the presence of conditioning, the problem becomes vastly more difficult, as the aggregate weight on the unexplored paths can be unbounded, giving $\infty$ as the only derivable upper bound (and therefore also $\infty$ as the best upper bound on the normalising constant).
In order to infer guaranteed bounds, it is necessary to analyse \emph{all} paths in the program, which we accomplish via static analysis and in particular our interval type system.
The approach from \cite{SankaranarayananCG13} was extended by \citet{AlbarghouthiDDN17} to compute the probability of events defined by arbitrary SMT constraints but is restricted to score-free and non-recursive programs.
Our interval-based approach, which may be seen as a variant of theirs, is founded on a complete semantics (\cref{thm:Completeness of interval approximations}), can handle recursive program with (soft) scoring, and is applicable to a broad class of primitive functions.

Note that we consider programs with \emph{soft} conditioning in which scoring cannot be reduced to volume computation directly.%
\footnote{For programs including only hard-conditioning (i.e.~scoring is only possible with $0$ or $1$), the posterior probability of an event $\varphi$ can be computed by dividing the probability of all traces with weight $1$ on which $\varphi$ holds by the probability of all traces with weight $1$.}
Intuitively, soft conditioning performs a (global) re-weighting of the set of traces, which cannot be captured by (local) volume computations.
In our interval trace semantics, we instead track an approximation of the weight along each interval trace.

\paragraph{Exact inference}
There are numerous approaches to inferring the exact denotation of a probabilistic program.
\citet{HoltzenBM20} introduced an inference method to efficiently compute the denotation of programs with discrete distributions.
By exploiting program structure to factorise inference, their system Dice can perform exact inference on programs with hundreds of thousands of random variables.
\citet{GehrMV16} introduced PSI, an exact inference system that uses symbolic manipulation and integration.
A later extension, $\lambda$PSI \cite{GehrSV20}, adds support for higher-order functions and nested inference.
The PPL Hakaru \cite{NarayananCRSZ16} supports a variety of inference algorithms on programs with both  discrete and continuous distributions.
Using program transformation and partial evaluation, Hakaru can perform exact inference via symbolic disintegration \cite{ShanR17} on a limited class of programs.
\citet{SaadRM21} introduced SPPL, a system that can compute exact answers to a range of probabilistic inference queries, by translating a restricted class of programs to sum-product expressions, which are highly effective representations for inference.

While exact results are obviously desirable, this kind of inference only works for a restricted family of programs: none of the above exact inference systems allow (unbounded) recursion.
Unlike our tool, they are therefore unable to handle, for instance, the challenging \cref{ex:pedestrian} or the programs in \cref{fig:recursive-models}.

\paragraph{Abstract interpretation}
There are various approaches to probabilistic abstract interpretation, so we can only discuss a selection.
\citet{Monniaux00, Monniaux01} developed an abstract domain for (score-free) probabilistic programs given by a weighted sum of abstract regions.
\citet{Smith08} considered truncated normal distributions as an abstract domain and developed analyses restricted to score-free programs with only linear expressions.
Extending both approaches to support soft conditioning is non-trivial as it requires the computation of integrals on the abstract regions.
In our interval-based semantics, we abstract the concrete traces (by means of interval traces) and not the denotation.
This allows us to derive bounds on the weight along the abstracted paths.

\citet{HuangDM21} discretise the domain of continuous samples into interval cubes and derive posterior approximations on each cube.
The resulting approximation converges to the true posterior (similarly to approximate/stochastic methods) but does not provide exact/guaranteed bounds and is not applicable to recursive programs.

\paragraph{Refinement types}
Our interval type system (\cref{sec:5interval-analysis}) may be viewed as a type system that refines not just the value of an expression but also its weight \cite{FreemanP91}.
To our knowledge, no existing type refinement system can bound the weight of program executions.
Moreover, the seamless integration with our interval trace semantics by design allows for much cheaper type inference, without resorting to an SMT or Horn constraint solver.
This is a crucial advantage since a typical GuBPI execution queries the analysis numerous times.

\paragraph{Stochastic methods}
A general approach to validating inference algorithms for a generative Bayesian model is \emph{simulation-based calibration} (SBC) \cite{SBC,cook2006validation}, discussed in \cref{sec:sbc}.
\citet{Grosse2015} introduced a method to estimate the log marginal likelihood of a model by constructing stochastic lower/upper bounds.
They show that the true value can be sandwiched between these two stochastic bounds with high probability.
In closely related work \cite{Cusumano-Towner2017b,Grosse2016a}, this was applied to measure the accuracy of approximate probabilistic inference algorithms on a specified dataset.
By contrast, our bounds are non-stochastic and our method is applicable to arbitrary programs of a universal PPL.

\section{Conclusion}

We have studied the problem of inferring guaranteed bounds on the posterior of programs written in a universal PPL.
Our work is based on the interval trace semantics, and our weight-aware interval type system gives rise to a tool that can infer useful bounds on the posterior of interesting recursive programs.
This is a capability beyond the reach of existing methods, such as exact inference.
As a method of Bayesian inference for statistical probabilistic programs, we can view our framework as occupying a useful middle ground between approximate stochastic inference and exact inference.

\balance

{
\interlinepenalty=10000 
\bibliography{references}
}

\iffullversion
\clearpage
\appendix
\allowdisplaybreaks

\section{Supplementary Material for \texorpdfstring{\cref{sec:3intervals}}{Section 3}}
\label{app:sec4}

\subsection{Intervals as a lattice}

Intervals $\iv$ form a partially ordered set under interval inclusion ($\sqsubseteq$).
We will sometimes need the meet $\sqcap$ and join $\sqcup$ operations, corresponding to the greatest lower bound and the least upper bound of two intervals.
Note that the meet of two intervals does not exist if the two intervals are disjoint.
Concretely, these two operations are given by $[a,b] \sqcap [c,d] := [\max(a,c), \min(b,d)]$ (if the two intervals overlap) and $[a,b] \sqcup [c,d] := [\min(a,c), \max(b,d)]$.

For some applications (e.g. the interval type system), we need the interval domain to be a true lattice.
To turn $\iv$ into a lattice, we add a bottom element $\bot$ (signifying an empty interval).
The definition of the meet $\sqcap$ and join $\sqcup$ is extended in the natural way.
The meet $\sqcap$ is extended by defining $I_1 \sqcap \bot = \bot \sqcap I_2 = \bot$ and $I_1 \sqcap I_2 = \bot$ if the two intervals $I_1, I_2 \in \iv$ are disjoint.
The join $\sqcup$ satisfies $I \sqcup \bot = \bot \sqcup I = I$.

\subsection{Lifting Functions to Intervals}

For constants $c \in \RR$ (i.e.~nullary functions), for common functions like $+$, $-$, $\times$, $|\cdot|$, $\min$, $\max$, for monotonically increasing functions $f_\nearrow: \RR \to \RR$, and for monotonically decreasing functions $f_\searrow: \RR \to \RR$, it is easy to describe the interval-lifted functions $+^\iv$, $-^\iv$, $\times^\iv$, ${|\cdot|}^\iv$, $\min^\iv$, $\max^\iv$, $f_\nearrow^\iv$, and $f_\searrow^\iv$:\\
\scalebox{0.951}{\parbox{\linewidth}{
\begin{align*}
c^\iv &= [c, c] \\
-^\iv [x_1,y_1] &= [-y_1, -x_1] \\
|[x_1, y_1]|^\iv &= \begin{cases}
[0, \max(|x_1|,|y_1|)] \text{ if $x_1 \le 0 \le y_1$} \\
[\min(|x_1|,|y_1|), \max(|x_1|,|y_1|)] \text{ else}
\end{cases} \\
[x_1,y_1] +^\iv [x_2,y_2] &= [x_1 + x_2, y_1 + y_2] \\
[x_1,y_1] -^\iv [x_2,y_2] &= [x_1 - y_2, y_1 - x_2] \\
[x_1,y_1] \times^\iv [x_2,y_2] &= [\min(x_1x_2, x_1y_2,y_1x_2,y_1y_2), \\
&\qquad \max(x_1x_2, x_1y_2,y_1x_2,y_1y_2)] \\
{\min}^\iv([x_1,y_1],[x_2,y_2]) &= [\min(x_1,x_2), \min(y_1,y_2)] \\
{\max}^\iv([x_1,y_1],[x_n,y_n]) &= [\max(x_1,x_2), \max(y_1,y_2)] \\
f_\nearrow^\iv([x_1,y_1]) &= [f_\nearrow(x_1), f_\nearrow(y_1)] \\
f_\searrow^\iv([x_1,y_1]) &= [f_\searrow(y_1), f_\searrow(x_1)]
\end{align*}
}}\\
where we write $f(\pm\infty)$ for $\lim_{x\to\pm\infty} f(x) \in \RR_\infty$, respectively.

\subsection{Properties of Interval Reduction}

We can define a refinement relation $M \refines M'$ (``$M$ refines $M'$'') between a standard term $M$ and an interval term $M'$, if $M$ is obtained from $M'$ by replacing every occurrence of $\lit{[a,b]}$ with some $r \in [a,b]$.

\lemIntervalApproximation*
\begin{proof}
	If the interval reduction $\to_\iv$ gets stuck, $\weightSem P^\iv$ is $[0,\infty]$ and $\valueSem P^\iv$ is $[-\infty,\infty]$, so the claim is certainly true.
	Otherwise, for each $(M_\iv, \ivtr, w_\iv) \to_\iv (M'_\iv, \ivtr', w'_\iv)$ reduction step, we can do a reduction step $(M, \tr, w) \to (M', \tr', w')$ where $M' \refines M'_\iv$ and $w' \in w'_\iv$, and $\tr' \refines \ivtr'$ if $M \refines M_\iv$, $w \in w_\iv$, and $\tr \refines \ivtr$.
	Since the reduction doesn't get stuck, we end up with a value $\lit r \refines \lit{[a,b]}$, so $\valueSem P(\tr) = r \in [a,b] = \valueSem P^\iv(\ivtr)$.
\end{proof}

\subsection{Additional Possible Reduction Rules}

\label{app:sec-additional-reduction-rules}

The interval semantics as presented has the unfortunate property that even a simple program like
$$\ifElse{\sample}{\allowbreak\score(0)}{\score(1)}$$
requires infinitely many interval traces to achieve a finite upper bound.
The reason is that the right branch $\score(1)$ is taken if the sampled value is in the open interval $(0,1]$.
To approximate this using closed intervals $[a,b]$ that our analysis supports, we need infinitely many intervals, e.g.~$\{ [2^{-n-1}, 2^{-n}] \mid n \in \NN \}$.
Adding (half\nobreakdash-)open intervals to the semantics would solve this specific problem, but not more general ones, where the guard condition is for example
$$\sample- \sample \le 0.$$
In that case, we have to approximate the set $\{ (x,y) \in [0,1]^2
 \mid x \le y \}$.
For the lower bounds, that is not an issue, but for the upper bounds, we need an infinite number of interval traces again.
We would like to use the interval traces $\langle [0,\frac12], [0, \frac12] \rangle$ and $\langle [\frac12, 1], [0, 1] \rangle$ to cover this set, but the reduction gets stuck on them because it is not clear which branch should be taken.

To remedy this, we could add the following two rules.
\begin{prooftree}
\AxiomC{$a \le 0 < b$}
\UnaryInfC{$\intConf{ \ifElse{\lit{[a, b]}}{N}{P}, \ivtr, w} \to_\iv \intConf {N, \ivtr, w \times^\iv [0,1]}$}
\end{prooftree}
\begin{prooftree}
\AxiomC{$a \le 0 < b$}
\UnaryInfC{$\intConf{ \ifElse{\lit{[a, b]}}{N}{P}, \ivtr, w} \to_\iv \intConf {P, \ivtr, w \times^\iv [0,1]}$}
\end{prooftree}
They basically express that if the interval bounds are not precise enough to decide what branch to take, we can take both, but have to allow the weight to be zero because it's not guaranteed that the taken branch can actually happen.
This change can only improve the upper bounds, not the lower bounds because the lower bound on each weight is zero if the additional rules are used.
Then the definition of upper bound can be modified in the following way:
\begin{center}
	\scalebox{0.95}{\parbox{\linewidth}{
			\begin{align*}
				\upperBound P^\traceset(U) := \sum_{\ivtr \in \traceset} \sum_{\substack{(P, \ivtr, [1,1]) \to_\iv \\ (\lit{[a,b]}, \langle\rangle, [w_1,w_2])}}\!\!\!\!\!\!\!\volume(\ivtr) \cdot w_2 \cdot \big[[a,b] \cap U \ne \emptyset\big]
			\end{align*}
	}}
\end{center}

This is the strategy we use for our implementation and is a natural extension of the existing semantics: it requires very few changes to the soundness and completeness proofs.

A downside of the previous approach is that the bounds are not always very tight: for the term $\ifCond{\dots}{\score(50)}{\score(100)}$, it returns bounds $[0, 150]$ instead of $[50, 100]$.
To improve this, we could omit the multiplication with $[0,1]$.
\begin{prooftree}
\AxiomC{$a \le 0 < b$}
\UnaryInfC{$\intConf{ \ifElse{\lit{[a, b]}}{N}{P}, \ivtr, w} \to_\iv \intConf {N, \ivtr, w}$}
\end{prooftree}
\begin{prooftree}
\AxiomC{$a \le 0 < b$}
\UnaryInfC{$\intConf{ \ifElse{\lit{[a, b]}}{N}{P}, \ivtr, w} \to_\iv \intConf {P, \ivtr, w}$}
\end{prooftree}
However, this complicates the equations of our bounds.
With this semantics, we have to compute minima and suprema instead of a simple sum:
\begin{center}
	\scalebox{0.95}{\parbox{\linewidth}{
			\begin{align*}
				\lowerBound P^\traceset (U) &:= \sum_{\ivtr \in \traceset} \min_{\substack{(P, \ivtr, [1,1]) \to_\iv \\ (\lit{[a,b]}, \langle\rangle, [w_1,w_2])}} \!\!\!\!\!\! \volume(\ivtr) \cdot w_1 \cdot \big[[a,b] \subseteq U\big] \\
				\upperBound P^\traceset(U) &:= \sum_{\ivtr \in \traceset} \sup_{\substack{(P, \ivtr, [1,1]) \to_\iv \\ (\lit{[a,b]}, \langle\rangle, [w_1,w_2])}}\!\!\!\!\!\! \volume(\ivtr) \cdot w_2 \cdot \big[[a,b] \cap U \ne \emptyset\big]
			\end{align*}
	}}
\end{center}
This is harder to implement because the sums cannot be computed incrementally, but many temporary results have to be kept in memory to compute the minima and suprema.
Proving soundness and completeness for this would require more substantial changes to the proofs.


\section{Symbolic Execution}
\label{app:sec-symbolic}

In this section we formally introduce stochastic symbolic execution.
We make use of this form of symbolic execution in two separate ways.
First, our completeness proof hinges on guarantees provided by the symbolic execution in order to identify a suitable set of interval traces.
Second, our tool GuBPI relies on the symbolic execution as a first step in the program analysis, in order to identify relevant paths and independent subexpressions, and to avoid repeated evaluation in a small-step semantics.

\paragraph{High-level idea}

The overarching idea of symbolic execution is to postpone the evaluation of $\sample$ expressions and instead use a \emph{sample variable} to symbolically represent its outcome.
As a consequence, branching and scoring steps cannot be executed concretely, so we record them symbolically instead.

\paragraph{Symbolic terms}

To postpone concrete sample decisions we introduce \emph{sample variables} $\alpha_1, \alpha_2, \dots$ into our language.
We then define \emph{symbolic terms} and \emph{symbolic values} by extending interval terms and values by adding two new constructs:
every sample variable $\alpha_j$ is a symbolic value and for every primitive function $f$ and symbolic values $\calV_1, \dots, \calV_{|f|}$, the symbolic term $\lit f(\calV_1, \dots, \calV_{|f|})$ is a symbolic value, denoting a function application that is postponed until all sample variables are instantiated.
We denote symbolic terms by $\calM, \calN, \calP$ and symbolic values by $\calV, \calW$.
Formally we define
\begin{align*}
    \calV &:= x \mid \lit r \mid \lambda x. \calM \mid \fixLam \varphi x \calM \mid \alpha_i \mid \lit f(\calV_1, \dots, \calV_{|f|})\\
    \calM, \calN, \calP &:= \calV \mid \calM \calN \mid \ifSimple \calM \calN \calP \mid \lit f(\calM_1, \dots, \calM_{|f|})\\
    &\quad\quad\mid \sample \mid \score(\calM)
\end{align*}%
The definition of redex and evaluation context extends naturally (recall that we regard $\alpha_j$ as a value).


\begin{figure}[!t]
    \small
    \begin{prooftree}
        \AxiomC{}
        \UnaryInfC{$\symConf{(\lambda x.  \calM) \calV, n, \Delta, \Xi} \toSym \symConf{\calM[\calV/x], n, \Delta, \Xi}$}
    \end{prooftree}

    \begin{prooftree}
        \AxiomC{}
        \UnaryInfC{$\symConf{\score(\calV), n, \Delta, \Xi} \toSym \symConf{\calV, n, \Delta \cup \{\calV \geq 0 \}, \Xi \cup \{\calV\}}$}
    \end{prooftree}

    \begin{prooftree}
        \AxiomC{}
        \UnaryInfC{$\symConf{(\fixLam{\varphi}{x} \calM) \calV, n, \Delta, \Xi} \toSym \symConf{\calM[\calV/x, (\fixLam{\varphi}{x} \calM)/\varphi], n, \Delta, \Xi}$}
    \end{prooftree}

    \begin{prooftree}
        \AxiomC{}
        \UnaryInfC{$\symConf{\sample, n, \Delta, \Xi} \toSym \symConf{\alpha_{n+1}, n+1, \Delta, \Xi}$}
    \end{prooftree}

    \begin{prooftree}
        \AxiomC{}
        \UnaryInfC{$\symConf{\ifElse{\calV}{\calN}{\calP}), n, \Delta, \Xi} \toSym \symConf{\calN, n, \Delta \cup \{\calV \leq 0\}, \Xi}$}
    \end{prooftree}

    \begin{prooftree}
        \AxiomC{}
        \UnaryInfC{$\symConf{\ifElse{\calV}{\calN}{\calP}), n, \Delta, \Xi} \toSym \symConf{\calP, n, \Delta \cup \{\calV > 0\}, \Xi }$}
    \end{prooftree}

    \begin{prooftree}
        \AxiomC{$\symConf{\calR, n, \Delta, \Xi} \toSym \symConf{\calM, n', \Delta', \Xi'}$}
        \UnaryInfC{$\symConf{\calE[\calR], n, \Delta, \Xi} \toSym \symConf{\calE[\calM], n', \Delta', \Xi'} $}
    \end{prooftree}
    \caption{Reduction rules for symbolic execution} \label{fig:symbolicSemantics}
\end{figure}

\paragraph{Symbolic execution}

A symbolic constraint is a pair $(\calV \bowtie r)$ where $\calV$ is a symbolic value, ${\bowtie} \in \{\leq, < , \geq, >\}$ and $r \in \real$.
A \defn{symbolic configuration} has the form $\psi = \symConf{\calM, n, {\Delta, \Xi}}$ where $\calM$ is a symbolic term,
$n \in \natnum$ a natural number (used to obtain fresh sample variables),
$\Delta$ a set of symbolic constraints (which track the symbolic conditions on the current execution path),
and $\Xi$ is a set of symbolic values (which records all symbolic values scored on the current path).
When executing symbolically:
(1) we evaluate each $\sample$ to a fresh sample variable, (2) we postpone function application, (3) for each conditional, we explore both branches (our reduction is nondeterministic) and record the symbolic inequalities that must hold along the current path, and (4) we record the symbolic values that we scored with.
We give the reduction rules in \cref{fig:symbolicSemantics}.

We call a tuple $\Psi = \symPath{\calV, n,\Delta, \Xi}$ (a symbolic configuration where the symbolic term is a value) a \defn{symbolic path}.
For a symbolic configuration $\psi$, we write $\mathit{symPaths}(\psi)$ for the set of symbolic paths reached when evaluating from $\psi$.
Note that $\mathit{symPaths}(\psi)$ is countable.

Let $\calV$ be a symbolic value of type $\typeReal$ (no $\lambda$-abstraction or fixed point) with sample variables within $\{\alpha_1, \dots, \alpha_n\}$.
For a trace $\tr = \langle r_1,\dots,r_n \rangle \in [0,1]^n$, we define $\calV[\tr/\overline \alpha]$ as the value (in $\RR$) obtained by replacing the sample variables in $\overline{\alpha}$ with $\tr$ and evaluate the postponed primitive function applications.
For a symbolic path $\Psi = \symPath{\calV, n,\Delta, \Xi}$, we define $\llbracket \Psi \rrbracket(U)$ as\\
\scalebox{1}{\parbox{\linewidth}{
		\begin{align*}
			\int_{[0,1]^n} \!
			\big[\calV[\tr/\overline \alpha] \in U\big] \!\!\! \prod_{\calC \bowtie r \in \Delta} \!\!\!\! \big[\calC[\tr/\overline \alpha] \bowtie r\big]
			\prod_{\calW \in \Xi} \!\!\calW[\tr/\overline \alpha]
			\D\tr.
		\end{align*}%
}}

\paragraph{Solution to symbolic constraints}

To simplify notation (and avoid extensive use of Iverson brackets) we introduce notation for the set of traces that satisfy a set of symbolic constraints.
Given a set of symbolic constraints $\Delta$ with sample variables contained in $\{ \alpha_1, \dots, \alpha_n \}$ we define
\[ \Sat_n(\Delta) := \bigcap_{(\calV \bowtie r) \in \Delta} \{ \tr \in [0,1]^n \mid \calV[\tr/\overline \alpha] \bowtie r \} \]
as the set of actual values for the sample variables that satisfy all constraints.
It follows immediately from the definitions that we can replace the Iverson brackets in $\llbracket \Psi \rrbracket$ by directly restricting the integral to the traces in $\Sat_n(\Delta)$.

\begin{lemma}
	For any symbolic path $\symPath{\calV, n,\Delta, \Xi}$ and any $U \in \Sigma_\RR$ we have
	\begin{align*}
		\llbracket \symPath{\calV, n,\Delta, \Xi} \rrbracket(U) = \int_{\Sat_n(\Delta)}
		[\calV[\tr/\overline \alpha] \in U]
		\prod_{\calW \in \Xi} \calW[\tr/\overline \alpha]
		\D \tr.
	\end{align*}
	\label{lem:path-denotation-vs-satn}
\end{lemma}

\paragraph{Correctness of symbolic execution}

We can now establish a correspondence between symbolic execution and the ordinary reduction.
If we wish to symbolically analyse a term $P$, we consider the (symbolic) reductions starting from  $\symConf{P, 0,\emptyset, \emptyset}$, resulting in the symbolic paths $\mathit{symPaths}\symConf{P, 0,\emptyset, \emptyset}$.

\begin{lemma}\label{lem:helper1}
Let $\vdash P : \typeReal$ and suppose we have $(\calV, n, \Delta, \Xi) \in \mathit{symPaths}\symConf{P, 0,\emptyset, \emptyset}$, where $P$ is interpreted as a symbolic term.
  Then for any $\tr \in Sat_n(\Delta)$, we have
  \[ (P, \tr, 1) \to^* (\calV[\tr/\overline \alpha], \langle \rangle, \prod_{\calW \in \Xi} \calW[\tr/\overline \alpha]). \]
\end{lemma}
\begin{proof}
	A similar proof can be found in \citep[Theorem~1]{MakOPW21}.
\end{proof}

\begin{lemma}\label{lem:helper2}
  Let $\vdash P : \typeReal$ and suppose $(P, \tr, 1) \to^* (\lit{r}, \langle\rangle, w)$ for some $r \in \RR$.
  Then there exists a \emph{unique} $(\calV, n, \Delta, \Xi) \in \mathit{symPaths}\symConf{P, 0,\emptyset, \emptyset}$ such that $\tr \in Sat_n(\Delta)$.
  For this unique symbolic path we have $w = \prod_{\calW \in \Xi} \calW[\tr/\overline \alpha]$ and $r = \calV[\tr/\overline \alpha]$.
\end{lemma}
\begin{proof}[Proof sketch]
  Choose the same branches in the $\toSym$-reduction of $P$ as in its $\to$-reduction.
  Then it is straightforward to see that this correspondence holds at every symbolic reduction step:
  if $(P, 0, \emptyset, \emptyset) \toSym^* (\calP', n, \Delta, \Xi)$ then the corresponding $\to$-reduction steps yield $(P, \tr\tr', 1) \to^* (P', \tr', w)$ where $\tr$ has length $n$, $P'$ is $\calP'[\tr/\overline{\alpha}]$ (after evaluating delayed primitve function applications), $\Delta$ records the guards $\calC[\tr/\overline{\alpha}] \le 0$ or $\calC[\tr/\overline{\alpha}] > 0$ that need to hold for the trace $\tr$, and finally, the weight $w$ is given by $\prod_{\calW \in \Xi} \calW[\tr/\overline \alpha]$ at any point.
\end{proof}

\symbolicExec*
\begin{proof}
  \begin{align*}
    \measureSem P(U) &= \sum_{n \in \NN} \int_{[0,1]^n} [\valueSem P (\tr) \in U] \weightSem P (\tr) \D \tr \\
    &= \sum_{n \in \NN}
      \int_{[0,1]^n}
      \sum_{\symPath{\calV, n,\Delta, \Xi}} \\
    &\quad
            \left(
            [\tr \in \Sat_n(\Delta)]
            [\calV[\tr/\overline \alpha] \in U]
            \prod_{\calW \in \Xi} \calW[\tr/\overline \alpha] \D \tr
            \right) \\
    &= \sum_{\symPath{\calV, n,\Delta, \Xi}}
    \int_{\Sat_n(\Delta)}
      [\calV[\tr/\overline \alpha] \in U]
      \prod_{\calW \in \Xi} \calW[\tr/\overline \alpha]
    \D \tr\\
    &= \sum_{\symPath{\calV, n,\Delta, \Xi}} \llbracket \symPath{\calV, n,\Delta, \Xi} \rrbracket(U)
  \end{align*}
	where the sum ranges over symbolic paths $\symPath{\calV, n,\Delta, \Xi} \in \mathit{symPaths}\symConf{P, 0,\emptyset, \emptyset}$.
  The first equality is by definition, the second one by \cref{lem:helper1,lem:helper2}, the third by noting that $\Sat_n(\Delta) \subseteq [0,1]^n$ and exchanging the infinite sum and integral (which is allowed because everything is nonnegative) and the fourth by \cref{lem:path-denotation-vs-satn}.
\end{proof}


\section{Supplementary Material for \texorpdfstring{\cref{sec:4intervals-theory}}{Section~4}}

\subsection{Infinite Trace Semantics}
\label{sec:inf-trace-sem}

\newcommand{\itr}{\boldsymbol{u}}

A convenient alternative to the (finite) trace semantics is using infinite traces $\traces_\infty := [0,1]^\NN$ with a suitable $\sigma$-algebra and measure $\mu_{\traces_\infty}$ \cite{CulpepperC17,Kenyon-RobertsO21}.
The $\sigma$-algebra on $\traces_\infty$ is defined as the smallest $\sigma$-algebra that contains all sets $U \times \traces_\infty$ where $U \in \Sigma_{[0, 1]^n}$ for some $n \in \NN$.
The measure $\mu_{\traces_\infty}$ is the unique measure with $\mu_{\traces_\infty}(U \times \traces_\infty) = \lambda_n(U)$ for $U \in \Sigma_{[0, 1]^n}$.
We use the symbol $\itr$ for an infinite trace in $\traces_\infty$.
For a finite trace $\tr$ and infinite trace $\itr$ we write $\tr\itr \in \traces_\infty$ for their concatenation.
For any infinite trace $\itr \in \traces_\infty$, there is at most one prefix $\tr \in \traces$ with $\weightSem P(\tr) > 0$ since the reduction is deterministic.
We can therefore define $\weightSem P^\infty(\itr) := \weightSem P(\tr)$ and $\valueSem P^\infty(\itr) := \valueSem P(\tr)$ if such a prefix $\tr$ exists, and $\weightSem P^\infty(\itr) := 0$ and $\valueSem P^\infty(\itr)$ is undefined otherwise.
The infinite trace semantics of a term is then defined as
\[
\measureSem P(U) := \int_{(\valueSem P^\infty)^{-1}(U)} \weightSem P^\infty(\itr) \, \mu_{\traces_\infty}(\D \itr).
\]

\begin{lemma}
  \label{lem:finite-infinite-trace-semantics-equal}
  The finite and infinite trace semantics agree, that is:
  \begin{align*}
  	\int_{(\valueSem P)^{-1}(U)} \weightSem P(\tr) \, \mu_{\traces}(\D \tr) = \int_{(\valueSem P^\infty)^{-1}(U)} \weightSem P^\infty(\itr) \, \mu_{\traces_\infty}(\D \itr).
  \end{align*}
\end{lemma}
\begin{proof}
Observe that $\valueSem P^\infty(\tr\itr) = \valueSem P(\tr)$ and $\weightSem P^\infty(\tr\itr) = \weightSem P(\tr)$ for all $\itr \in \traces_\infty$ if $\weightSem P(\tr) > 0$.
Then we get:
\begin{align*}
  &\int_{(\valueSem P)^{-1}(U)} \weightSem P(\tr) \, \mu_{\traces}(\D \tr) \\
  = &\int_{(\weightSem P)^{-1}(\RR_{> 0})} [\valueSem P(\tr) \in U]\weightSem P(\tr) \int_{\traces_\infty} \mu_{\traces_\infty}(\D \itr) \mu_{\traces}(\D \tr) \\
  = &\int_{(\weightSem P)^{-1}(\RR_{> 0})} \int_{\traces_\infty} [\valueSem P^\infty(\tr\itr) \in U] \weightSem P^\infty(\tr\itr) \mu_{\traces_\infty}(\D \itr) \mu_{\traces}(\D \tr) \\
  = &\int_{(\weightSem P)^{-1}(\RR_{> 0}) \times \traces_\infty} [\valueSem P^\infty(\itr) \in U] \weightSem P^\infty(\itr) \mu_{\traces_\infty}(\D \itr) \\
  = &\int_{(\valueSem P^\infty)^{-1}(U)} \weightSem P^\infty(\itr) \mu_{\traces_\infty}(\D \itr)
\end{align*}
where we used the fact that the sets $\{\tr\} \times \traces_\infty$ are disjoint for different $\tr \in \weightSem P^{-1}(\RR_{> 0})$ because otherwise we would be able to find a trace $\tr$ as a prefix of $\tr'$ and both having positive weight, which is impossible due to the deterministic reduction.
Therefore $(\weightSem P)^{-1}(\RR_{> 0}) \times \traces_\infty = (\weightSem P^\infty)^{-1}(\RR_{> 0})$ and everything works as desired.
Note that the second to last equality follows from Fubini's theorem and the fact that the product measure of $\mu_\traces$ and $\mu_{\traces_\infty}$ is $\mu_{\traces_\infty}$ again.
\end{proof}

\subsection{Exhaustivity and Soundness}

\begin{example}[more examples of exhaustive sets]
  Here are more examples and counterexamples for exhaustivity.
  \begin{asparaenum}[(i)]
    \item $\{ \langle \rangle \}$ is an (uninteresting) exhaustive set and only useful for deterministic programs.
    \item $\{ \langle [2^{-n - 1}, 2^{-n}] \rangle \mid n \in \NN \}$ is exhaustive because only the trace $\langle 0 \rangle$ (with measure 0) is not covered.
    \item Let $a_n \ge 0$ be a converging series, i.e.~$\sum_{i=1}^n a_n < \infty$, for example $a_n = n^{-2}$.
    Define
    \[ \traceset := \{ \langle [0, e^{-a_1}], \dots, [0, e^{-a_n}], [e^{-a_{n+1}}, 1] \rangle \mid n \in \NN \}. \]
    This is not an exhaustive set since it doesn't cover any of the traces in
    \[ [0, e^{-a_1}) \times [0, e^{-a_2}, 1) \times \cdots \]
    which has measure $\prod_{i=1}^\infty e^{-a_i} = \exp\left(-\sum_{i=1}^\infty a_i\right) > 0$.
  \end{asparaenum}
  \end{example}

We also note that exhaustivity can be expressed just in terms of finite traces as well, at the cost of a more complicated definition.

\begin{lemma}
A set of interval traces $\traceset$ is exhaustive if and only if
\[ \mu_\traces \left( [0,1]^n \setminus \left(\bigcup_{\langle I_1, \dots, I_m \rangle \in \traceset, m \le n} I_1 \times \cdots \times I_m \times [0,1]^{n-m}\right) \right) \to 0 \]
as $n \to \infty$.
\label{lem:exhaustive-infinite-traces}
\end{lemma}
\begin{proof}
  Let $S := \traces_\infty \setminus \bigcup_{\ivtr \in \traceset} \mathit{cover}(\ivtr)$.
  By the definition of exhaustivity, $\mu_{\traces_\infty}(S) = 0$.
  Let
  \[ S_n = [0,1]^n \setminus \left(\bigcup_{\langle I_1, \dots, I_m \rangle \in \traceset, m \le n} I_1 \times \cdots \times I_m \times [0,1]^{n-m}\right). \]
  It's easy to see that $S = \bigcap_{n = 0}^\infty S_n \times \traces_\infty$ where $S_n \times \traces_\infty$ is a decreasing sequence of sets: $S_1 \times \traces_\infty \supseteq S_2 \times \traces_\infty \supseteq \cdots$.
  Since measures are continuous from above, we have
  $\lim_{n \to \infty} \mu_{\traces}(S_n) = \lim_{n \to \infty} \mu_{\traces_\infty}(S_n \times \traces_\infty) = \mu_{\traces_\infty}(S) = 0$, as desired.
\end{proof}

The following lemma establishes a correspondence between infinite trace semantics and interval trace semantics.

\begin{lemma}
  For any interval trace $\ivtr$ and infinite trace $\tr_\infty$ with a prefix $\tr$ such that $\tr \refines \ivtr$, we have $\weightSem P^\infty(\tr_\infty) \in \weightSem P^\iv(\ivtr)$ and $\valueSem P^\infty(\tr_\infty) \in \valueSem P^\iv(\ivtr)$.
  \label{lem:interval-approximation-infinite}
\end{lemma}
\begin{proof}
Follows directly from the definition of infinite trace semantics and \cref{lem:interval-approximation}.
\end{proof}

Using the previous results, we can prove soundness of upper bounds.

\thmSoundUpper*
\begin{proof}
For any $U \in \Sigma_\RR$, we have
\begin{align}
& \upperBound P^\traceset(U) \nonumber\\
  &= \sum_{\ivtr \in \traceset} \volume(\ivtr) (\sup \weightSem P^\iv(\ivtr)) [\valueSem P^\iv(\ivtr) \cap U \ne \emptyset] \nonumber\\
  &= \sum_{\ivtr \in \traceset} \int_{\tracesin{\ivtr}} (\sup \weightSem P^\iv(\ivtr)) [\valueSem P^\iv(\ivtr) \cap U \ne \emptyset] \D \tr \nonumber\\
  &\ge \sum_{\ivtr \in \traceset} \int_{\tracesin{\ivtr}}  \int_{\traces_\infty} \weightSem P^\infty(\tr\itr)[\valueSem P^\infty(\tr\itr) \in U] \D \itr \D \tr \label{eq:4th complete}\\
  &\ge \int_{\bigcup_{\ivtr \in \traceset} \tracesin{\ivtr} \times \traces_\infty} \weightSem P^\infty(\itr)[\valueSem P^\infty(\itr) \in U] \D \tr \nonumber\\
  &\ge \int_{\traces_\infty} \weightSem P^\infty(\itr)[\valueSem P^\infty(\itr) \in U] \D \itr  \label{eq:6th complete}\\
  &= \measureSem P(U)  \label{eq:7th complete}
\end{align}
where \cref{eq:4th complete} follows from \cref{lem:interval-approximation-infinite}, \cref{eq:6th complete} from exhaustivity and \cref{eq:7th complete} from \cref{lem:finite-infinite-trace-semantics-equal}.
\end{proof}

\subsection{Assumptions for Completeness}
\label{sec:qtt}

\paragraph{Remarks on Assumption 1}
We can formally express Assumption 1 from \cref{sec:4intervals-theory} about a given program $\vdash P: \typeReal$ as follows.
For each symbolic path $\Psi = \symPath{\calV, n,\Delta, \Xi}$, we require that $\calV$, each $\calC$ with $\calC \bowtie 0 \in \Delta$, and each $\calW \in \Xi$ contain each sample variable $\alpha_i$ at most once.

\begin{example}
The pedestrian example (\cref{ex:pedestrian}) satisfies Assumption 1 because the symbolic paths have the form $\Psi = \symPath{\calV, n,\Delta, \Xi}$ with:
\begin{align*}
  \calV &= 3\alpha_1  \\
  n &= 2k + 1 \\
  \Delta &= \{ \alpha_3 - \tfrac12 \bowtie 0, \alpha_5 - \tfrac12 \bowtie 0, \dots, \alpha_{2k+1} - \tfrac12 \bowtie 0 \} \\
    &\quad \begin{aligned}
      {} \cup \{ &3\alpha_1 > 0, \\
      &3\alpha_1 \pm \alpha_2 > 0, \\
      &\dots, \\
      &3\alpha_1 \pm \alpha_2 \pm \alpha_4 \cdots \pm \alpha_{2k-2} > 0, \\
      &3\alpha_1 \pm \alpha_2 \pm \alpha_4 \cdots \pm \alpha_{2k} \le 0 \}
    \end{aligned} \\
  \Xi &= \{ \pdf_{\Normal(1.1, 0.1)}(\alpha_2 + \alpha_4 + \cdots + \alpha_{2k}) \}
\end{align*}
As we can see, none of the symbolic values contains a sample variable twice, so the assumption is satisfied.
\end{example}

\paragraph{Remarks on Assumption 2}
We first prove the sufficient condition for interval separability from \cref{sec:completeness}.
\begin{lemma}
  If a function $f: \RR^n \to \RR$ is boxwise continuous and preimages of points are null sets then $f$ is interval separable.
  \label{lem:continuous-preimage-null-interval-separable}
\end{lemma}
\begin{proof}
We decompose $f^{-1}([a,b]) = f^{-1}((a,b)) \cup f^{-1}(\{a,b\})$ and deal with the former set first.
By boxwise continuity, $f = \bigcup_i f|_{B_i}$ where $\bigcup_i B_i = \RR^n$ and each $f|_{B_i}$ is continuous on $B_i$.
To show that the preimage $f^{-1}((a,b))$ can be tightly approximated by a countable set of boxes, it suffices to show this for each $(f|_{B_i})^{-1}((a,b))$.
This set is open in $B_i$ by continuity of $f|_{B_i}$, so it can be written as a countable union of boxes (e.g.~by taking a box within $B_i$ around each rational point, which exists because it's an open set).
By the assumption, the preimage $f^{-1}(\{a,b\})$ is a null set.
Hence $f^{-1}([a,b])$ can be approximated by a null set.
\end{proof}

Note that a composition of interval separable functions need \emph{not} be interval separable.
This is an incorrect assumption made in the completeness proof of \cite{BeutnerO21}.
(To fix their Theorem 3.8, one needs to make the additional assumption that the set of primitive functions be closed under composition.)
To see this, let $f, g: \RR \to \RR$ be interval separable functions and $I$ an interval.
By definition, there are intervals $B_i$ such that $\bigcup_i B_i \cup N = f^{-1}(I)$ where $N$ is a null set.
Then by interval separability, the preimage $g^{-1}(\bigcup_i B_i)$ can be tightly approximated by interals $B'_j$, but the preimage $g^{-1}(N)$ need not be a null set.
It is also not clear at all whether one can approximate the preimage $(f \circ g)^{-1}(I)$ tightly using intervals without further restrictions on $f$ and $g$.
For this reason, we require the assumption that the set of primitive functions be closed under composition.

It is not immediately obvious that such a set of functions exists.
One example is given by the following.
A function $f: \RR^n \to \RR$ is called a \defn{submersion} if it is continuously differentiable and its gradient is nonzero everywhere.
\begin{lemma}
  The set $\mathcal F_\mathsf{subm}$ of submersions is closed under composition and each of its functions is boxwise continuous and interval separable.
\end{lemma}
\begin{proof}
  Boxwise continuity is obvious given that the functions are even continuously differentiable.
  For interval separability, we use \cref{lem:continuous-preimage-null-interval-separable}.
  Let $f: \RR^n \to \RR \in \mathcal F_\mathsf{subm}$.
  Since $f$ is a submersion, the preimage $f^{-1}(x)$ of any point $x \in \RR$ is an $(n-1)$-dimensional submanifold of $\RR^n$ by the preimage theorem (a variation of the implicit function theorem).
  Submanifolds of codimension $> 1$ have measure zero.
  (This well-known fact can be shown by writing the submanifold as a countable union of graphs and applying Fubini's theorem to each of them.)
  Therefore, the lemma applies.

  For closure under composition, let $f: \RR^m \to \RR$ and $f_i: \RR^{n_i} \to \RR$ for $i \in \{ 1, \dots, m \}$, all in $\mathcal F_\mathsf{subm}$.
  The composition $g := f \circ (f_1 \times \dots \times f_m)$ is clearly $C^1$ again, so we just have to check the submersion property.
  By the chain rule, we find that the gradient of the composition
  \[ \nabla g(x_1, \dots, x_m) = \begin{pmatrix}
      \partial_1 f(f_1(x_1), \dots, f_m(x_m)) \cdot \nabla f_1(x_1) \\
      \vdots \\
      \partial_m f(f_1(x_1), \dots, f_m(x_m)) \cdot \nabla f_m(x_m)
    \end{pmatrix} \]
  is nonzero because at least one of the $\partial_i f$ is nonzero and $\nabla f_i(x_i) \ne 0$ by assumption.
  Hence the composition is a submersion again.
\end{proof}

Unfortunately, the set of submersions does not contain constant functions.
This is a problem because then it is not guaranteed that partially applying a primitive function to a constant is still an admissible primitive function.
(For example, this would break \cref{lem:constraint-set-interval-approx}.)
Hence we need to assume that \emph{all constant functions be primitive functions}.
Luckily, the set $\mathcal F_\mathsf{subm}$ of submersions can be easily extended to accommodate this.
\begin{lemma}
  Let $\mathcal F_\mathsf{subm}^*$ be the set of functions $f: \RR^n \to \RR$ (for all $n \in \NN$) such that whenever the partial derivative $\partial_i f(x)$ is zero for some $i \in \{1, \dots, n \}$ and $x \in \RR^n$ then $f$ is constant in its $i$-th argument, i.e. there is a function $f^*: \RR^{n-1} \to \RR$ such that $f(x_1, \dots, x_n) = f^*(x_1, \dots, x_{i-1}, x_{i+1}, \dots x_n)$.
  This set satisfies all the assumptions about sets of primitive functions: it is closed under composition, contains all constant functions, and all its functions are boxwise continuous and interval separable.
\end{lemma}
\begin{proof}
  Boxwise continuity is obvious given that the functions are even continuously differentiable.
  Similarly, it is clear that $\mathcal{F}_\mathsf{subm}^*$ contains all constant functions.

  For interval separability, let $f: \RR^n \to \RR \in \mathcal F_\mathsf{subm}^*$ and $J \subseteq \{ 1, \dots, n \}$ be the set of indices in which $f$ is not constant, and $J'$ its complement.
  Hence there is a submersion $f_J: \RR^{|J|} \to \RR$ such that $f(x) = f_J(x_J)$ where $x_J$ stands for the vector of coordinates of $x$ with index in $J$.
  The preimage of $f^{-1}(U) \subseteq \RR^{|J|}$ of any set $U \subseteq \RR$ can be tightly approximated by boxes if and only if $f_J^{-1}(U)$ can because $f^{-1}(U)$ is a Cartesian product of $f_J^{-1}(U)$ and $\RR^{|J'|}$.
  Since $f_J$ is interval separable by the previous lemma, this shows that $f$ is as well.

  For closure under composition, let $f: \RR^m \to \RR$ and $f_i: \RR^{n_i} \to \RR$ for $i \in \{ 1, \dots, m \}$, all in $\mathcal F_\mathsf{subm}^*$.
  The composition $g := f \circ (f_1 \times \dots \times f_m)$ is clearly $C^1$ again, so we just have to check the property of the partial derivatives.
  By the chain rule, the partial derivatives of the composition are
  \[ \partial_i g(x_1, \dots, x_m) = \partial_j f(f_1(x_1), \dots, f_m(x_m)) \partial_k f_j(x_{j1}, \dots, x_{jn_j}) \]
  for some $j \in \{ 1, \dots m \}$ and $k \in \{ 1, \dots n_j \}$, and for all $x_1 \in \RR^{n_1}, \dots, x_m \in \RR^{n_m}$.
  So if this partial derivative is zero, there are two cases.
  First, if $\partial_j f(f_1(x_1), \dots, f_m(x_m)) = 0$ then $f$ must be constant in its $j$-th argument (because $f \in \mathcal{F}_\mathsf{subm}^*$) and thus $g$ is constant in $x_j$, and in particular the $i$-th argument (which is an entry of $x_j$).
  Second, if $\partial_k f_j(x_{j1}, \dots, x_{jn_j}) = 0$ then $f_j$ must be constant in its $k$-th argument (because $f_j \in \mathcal{F}_\mathsf{subm}^*$) and thus $g$ is constant in its corresponding $i$-th argument as well.
  This proves $g \in \mathcal{F}_\mathsf{subm}^*$, as desired.
\end{proof}
Note that $\exp$, $\sinh$, $\arctan$, $n$-th roots for $n$ odd, and all linear functions are in $\mathcal F_\mathsf{subm}^*$. So this is a useful set of primitive functions already.
Unfortunately, it does not include multiplication because the gradient of $(x,y) \mapsto xy$ is zero at $(0,0)$.
To fix this issue, we need to restrict the domain.%
\footnote{Handling these issues at the level of primitive functions directly (without restricting the domain) seems challenging: even if a function has only one point with zero gradient, e.g. multiplication, its preimage under other primitive functions can become very complicated.
We tried to handle this by allowing the primitive functions to be submersions except on a null set given by a union of lower-dimensional manifolds.
However, the preimages of such manifolds need not be manifolds again.
Hence it seems difficult to come up with a broader class of primitive functions satisfying the assumptions without restricting the domain.}

For simplicity, we required primitive functions to be defined on all of $\RR^n$.
Suppose we allow for primitive functions to be defined only on an open subset of $\RR^n$, and applying them to a value outside their domain is disallowed in SPCF programs.
Then we can also include multiplication (on $\RR^2 \setminus \{(0,0)\}$), logarithms (on $(0,\infty)$), non-constant univariate polynomials (on the complement of their stationary points), quantile functions of continuous distributions with nonzero density (on $(0,1)$), and probability density functions (on the complement of their stationary points).
The fact that some points in the domain are missing is inconvenient for functions that can be continuously extended to the these points, but one can work around this in a program by checking for the points that are not in the domain and returning the function values as constants for those cases.

\subsection{Completeness Proof}

This section uses the definitions of \emph{boxwise continuity} and \emph{interval separability} from \cref{sec:completeness}.
As discussed there, we assume that for each symbolic path $\Psi = \symPath{\calV, n,\Delta, \Xi}$, we have that $\calV$, each $\calC$ with $\calC \bowtie 0 \in \Delta$, and each $\calW \in \Xi$ contains each sample variable $\alpha_i$ at most once (Assumption 1).
We also assume that the primitive functions are boxwise continuous, interval separable and closed under composition (Assumption 2).
Furthermore, we say that $\traceset$ is a \defn{subdivision} of $\ivtr$ if $\traceset$ is compatible and $\bigcup_{\ivtr' \in \traceset} \tracesin{\ivtr'} = \tracesin{\ivtr}$.

\thmCompleteness*
\begin{proof}
  First, we give a brief outline of how the proof works.
  The idea is to cover $\{ \tr \in \traces \mid \valueSem P(\tr) \in I \}$ using boxes (interval traces).
  We can achieve this using symbolic execution: for a fixed path through the program, the result value is just a composition of primitive functions applied to the samples.
  Similarly, the weight function is a product of such functions, hence boxwise continuous.
  By passing to smaller boxes, we can assume that it is continuous on each box.
  In order to approximate the integral of the weight function, we use Riemann sums (as used in the definition of the Riemann integral).
  We partition the domain into smaller and smaller boxes such that the lower bound and the upper bound of the weight function come arbitrarily close (by continuity).
  Then by properties of the Riemann integral, the bounds arising from the interval traces representing the boxes in this partition converge to the desired integral of the weight function.
  The details of the proof are as follows.

  \textbf{Step 1:} approximating the branching inequalities.
  Let $\Psi = \symPath{\calV, n, \Delta, \Xi}$ a symbolic path of $P$.
  To find a countable set $\traceset_\Psi \subseteq \traces_\iv$ such that $\bigcup_{\ivtr \in \traceset_\Psi}\tracesin{\ivtr} \tightsubset \Sat_n(\Delta)$.
  Note that $\Sat_n(\Delta)$ is a finite intersection of sets of the form $\{ \tr \in [0,1]^n \mid \calC[\tr/\overline \alpha] \bowtie 0 \}$ where $\bowtie \in \{\leq, >\}$.
  In the $\leq$ case, we can write this constraint as $\calC[\tr/\overline \alpha] \in \bigcup_{n \in \NN} [-n, 0]$ and in the $>$ case as $\calC[\tr/\overline \alpha] \in \bigcup_{n \in \NN} [1/n, n]$.
  By applying \cref{lem:constraint-set-interval-approx} to each of the compact intervals in these unions, we obtain a countable union of boxes that is a tight subset of $\{ \tr \in [0,1]^n \mid \calC[\tr/\overline \alpha] \bowtie 0 \}$.
  Since the intersection of two boxes is a box and since $\Sat_n(\Delta)$ is a finite intersection of such countable unions of boxes, it can be rewritten as a countable union of boxes.
  This yields $\traceset_\Psi$, such that $\bigcup_{\ivtr \in \traceset_\Psi}\tracesin{\ivtr} \tightsubset \Sat_n(\Delta)$.

  \textbf{Step 2:} handling the result value.
  By applying \Cref{lem:constraint-set-interval-approx} and intersecting the obtained interval traces with $\traceset_\Psi$, we obtain a countable set $\traceset'_{\Psi,I} \subseteq \traces_\iv$ such that $\bigcup_{\ivtr \in \traceset'_{\Psi,I}} \tracesin{\ivtr} \tightsubset \{ \tr \in \Sat(\Delta) \mid \calV[\tr/\overline \alpha] \in I \}$.
  By the same lemma, we find a countable set $\traceset'_{\Psi,I^c} \subseteq \traces_\iv$ such that $\bigcup_{\ivtr \in \traceset'_{\Psi,I^c}} \tracesin{\ivtr} \tightsubset \{ \tr \in \Sat(\Delta) \mid \calV[\tr/\overline \alpha] \notin I \}$ because the complement of $I$ can be written as a countable union of intervals.
  By \cref{lem:relationship-symbolic-interval}, we find subdivisions $\traceset_{\Psi,I}$ and $\traceset_{\Psi,I^c}$ that even satisfy $\bigcup_{\ivtr \in \traceset_{\Psi,I}} \tracesin{\ivtr} \tightsubset \valueSem P^{-1}(I) \cap \Sat(\Delta)$ and $\bigcup_{\ivtr \in \traceset_{\Psi,I^c}} \tracesin{\ivtr} \tightsubset \valueSem P^{-1}(\RR \setminus I) \cap \Sat(\Delta)$.
  By \cref{lem:make-boxes-disjoint}, we can assume that the interval traces $\traceset_{\Psi, I}$ are almost disjoint.
  Because of the almost sure termination assumption, the set of traces $\bigcup_{\symPath{\calV,n,\Delta,\Xi}} \Sat(\Delta) = \valueSem P^{-1}(\RR)$ where the union ranges over all symbolic paths of $P$ has measure 1.
  As a consequence, $\bigcup_{\Psi \in \mathit{symPaths}\symConf{P, 0, \emptyset, \emptyset}} (\traceset_{\Psi,I} \cup \traceset_{\Psi,I^c})$ is a compatible and exhaustive set of interval traces.
  Now
  \begin{align*}
    \measureSem P(I) &= \sum_{\symPath{\calV,n,\Delta,\Xi}} \; \llbracket \symPath{\calV,n,\Delta,\Xi} \rrbracket (I)\\
    &= \sum_{ \symPath{\calV,n,\Delta,\Xi}}
    \int_{\Sat_n(\Delta)}
      [\calV[\tr/\overline \alpha] \in I]
      \prod_{\calW \in \Xi} \calW[\tr/\overline \alpha]
    \D \tr \\
    &= \sum_{\symPath{\calV,n,\Delta,\Xi}}
    \sum_{\ivtr \in \traceset_{\symPath{\calV,n,\Delta,\Xi},I}}
    \int_{\tracesin{\ivtr}}
      \prod_{\calW \in \Xi} \calW[\tr/\overline \alpha]
    \D \tr
  \end{align*}
	where the outer sum ranges over the symbolic paths $\symPath{\calV,\allowbreak n,\allowbreak \Delta,\allowbreak \Xi} \in \mathit{symPaths}\symConf{P, 0, \emptyset, \emptyset}$.
	The first equality holds by \cref{thm:relation-semantics-symbolic}, the second one by \cref{lem:path-denotation-vs-satn} and the last one by the construction of $\traceset_{\symPath{\calV,n,\Delta,\Xi},I}$.

  \textbf{Step 3:} approximating the weight function.
  Let
  \begin{align*}
  	\textstyle\traceset' := \bigcup_{\Psi \in \mathit{symPaths}(P, 0, \emptyset, \emptyset)} \traceset_{\Psi,I}.
  \end{align*}
  For each $\ivtr \in \traceset'$, fix some $\epsilon_{\ivtr} > 0$, such that $\sum_{\ivtr \in \traceset'} \epsilon_{\ivtr} = \epsilon$.
  This can be achieved, for example, by enumerating $\traceset'$ as $\ivtr^{(1)}, \ivtr^{(2)}, \dots $ and choosing $\epsilon_{\ivtr^{(i)}} = 2^{-i}\epsilon$.
  By \cref{lem:approximate-weight-integral-sum}, we can find for each $\ivtr \in \traceset'$ a countable set $\mathcal S_{\ivtr}$ of interval traces such that:
  \begin{align*}
    &\sum_{\ivtr' \in \mathcal S_{\ivtr}} \volume(\ivtr') \sup_{\tr \in \tracesin{\ivtr'}} \prod_{\calW \in \Xi} \calW[\tr/\overline \alpha] - \epsilon_{\ivtr} / 2 \\
    &\le \int_{\tracesin{\ivtr}} \prod_{\calW \in \Xi} \calW[\tr/\overline \alpha] \D \tr \\
    &\le \sum_{\ivtr' \in \mathcal S_{\ivtr}} \volume(\ivtr') \min_{\tr \in \tracesin{\ivtr'}} \prod_{\calW \in \Xi} \calW[\tr/\overline \alpha] + \epsilon_{\ivtr} / 2
  \end{align*}
  Next, choose $\epsilon_{\ivtr'} > 0$ for each $\ivtr'$ in such a way that $\sum_{\ivtr' \in \mathcal S'_{\ivtr}} \epsilon_{\ivtr'} < \epsilon_{\ivtr} / 2$.
  By \cref{lem:approximate-weight-intervals}, we can find for each $\ivtr'$ a finite subdivision $\mathcal S'_{\ivtr'}$ such that for all $\ivtr'' \in \mathcal S'_{\ivtr'}$, we have
  \begin{align*}
    \sup_{\tr \in \tracesin{\ivtr''}} \weightSem P(\tr) = \sup_{\tr \in \tracesin{\ivtr''}} \prod_{\calW \in \Xi} \calW[\tr/\overline \alpha] &\ge \sup \weightSem P^\iv(\ivtr'') - \epsilon_{\ivtr'} \\
    \min_{\tr \in \tracesin{\ivtr''}} \weightSem P(\tr) = \min_{\tr \in \tracesin{\ivtr''}} \prod_{\calW \in \Xi} \calW[\tr/\overline \alpha] &\le \min \weightSem P^\iv(\ivtr'') + \epsilon_{\ivtr'}.
  \end{align*}
  Multiplying by $\volume(\ivtr'')$ and summing over all $\ivtr''$, we find together with the previous inequality
  \begin{align*}
    &\sum_{\ivtr' \in \mathcal S_{\ivtr}} \sum_{\ivtr'' \in \mathcal S'_{\ivtr'}} \volume(\ivtr'') \sup \weightSem P^\iv(\ivtr'') -  \epsilon_{\ivtr} \\
    &\le \sum_{\ivtr' \in \mathcal S_{\ivtr}} \sum_{\ivtr'' \in \mathcal S'_{\ivtr'}} \volume(\ivtr'') (\sup \weightSem P^\iv(\ivtr'') - \epsilon_{\ivtr'}) -  \epsilon_{\ivtr} / 2 \\
    &\le \int_{\tracesin{\ivtr}} \prod_{\calW \in \Xi} \calW[\tr/\overline \alpha] \D \tr \\
    &\le \sum_{\ivtr' \in \mathcal S_{\ivtr}} \sum_{\ivtr'' \in \mathcal S'_{\ivtr'}} \volume(\ivtr'') (\min \weightSem P^\iv(\ivtr'') + \epsilon_{\ivtr'}) + \epsilon_{\ivtr} / 2 \\
    &\leq \sum_{\ivtr' \in \mathcal S_{\ivtr}} \sum_{\ivtr'' \in \mathcal S'_{\ivtr'}} \volume(\ivtr'') \min \weightSem P^\iv(\ivtr'') + \epsilon_{\ivtr}
  \end{align*}
  because $\sum_{\ivtr'' \in \mathcal S'_{\ivtr'}} \volume(\ivtr'') \le 1$ and thus the contribution of all the $\epsilon_{\ivtr'}$ is at most $\epsilon_{\ivtr} / 2$.

  Overall, the desired trace set is given by
  \[ \traceset := \bigcup_{\Psi = \symPath{\calV,n,\Delta,\Xi}} \left( \traceset_{\Psi,I^c} \cup \bigcup_{\ivtr \in \traceset_{\Psi, I}} \bigcup_{\ivtr' \in \mathcal S_{\ivtr}} \mathcal S'_{\ivtr'}\right), \]
  is compatible and exhaustive because it is a subdivision of $\traceset_{\Psi,I^c}$ and $\traceset_{\Psi,I}$.
  By construction, we have $\valueSem P^\iv(\ivtr) \subseteq I$ for $\ivtr \in \traceset_{\Psi,I}$ and $\valueSem P^\iv(\ivtr) \cap I = \emptyset$ for $\ivtr \in \traceset_{\Psi,I^c}$.
  Hence the $\traceset_{\Psi,I^c}$-summands vanish in the sum for the bounds and we obtain
  \begin{align*}
    &\upperBound P^\traceset(I) - \epsilon \\
    &= \sum_{\ivtr \in \traceset} \volume(\ivtr) (\sup \weightSem P^\iv(\ivtr)) [\valueSem P^\iv(\ivtr) \cap I \ne \emptyset] - \epsilon \\
    &= \sum_{\Psi = \symPath{\calV,n,\Delta,\Xi}}
    \sum_{\ivtr \in \traceset_{\Psi,I}} \\
    &\qquad
    \left(
    \sum_{\ivtr' \in \mathcal S_{\ivtr}}
    \sum_{\ivtr'' \in \mathcal S'_{\ivtr'}}
    \volume(\ivtr'') \sup \weightSem P^\iv(\ivtr'') - \epsilon_{\ivtr}
    \right) \\
    &\le \underbrace{\sum_{\Psi = \symPath{\calV,n,\Delta,\Xi}}
    \sum_{\ivtr \in \traceset_{\Psi,I}}
    \int_{\tracesin{\ivtr}} \prod_{\calW \in \Xi} \calW[\tr/\overline \alpha] \D \tr}_{=\measureSem P(I)} \\
    &\le \sum_{\Psi = \symPath{\calV,n,\Delta,\Xi}}
    \sum_{\ivtr \in \traceset_{\Psi,I}} \\
    &\qquad \left(
    \sum_{\ivtr' \in \mathcal S_{\ivtr}}
    \sum_{\ivtr'' \in \mathcal S'_{\ivtr'}}
    \volume(\ivtr'') \min \weightSem P^\iv(\ivtr'') + \epsilon_{\ivtr}
    \right) \\
    &= \sum_{\ivtr \in \traceset} \volume(\ivtr) (\min \weightSem P^\iv(\ivtr)) [\valueSem P^\iv(\ivtr) \subseteq I] + \epsilon \\
    &= \lowerBound P^\traceset(I) + \epsilon. \qedhere
  \end{align*}
\end{proof}

\begin{lemma}\label{lem:make-boxes-disjoint}
Given a countable set of interval traces $\traceset \subseteq \iv^n$, there is a countable set of interval traces $\traceset' \subseteq \iv^n$ that is compatible and satisfies $\bigcup_{\ivtr \in \traceset} \tracesin{\ivtr} = \bigcup_{\ivtr \in \traceset'} \tracesin{\ivtr}$.
\end{lemma}
\begin{proof}
  Let $A: \NN \to \traceset$ be an enumeration.
  Define $A': \NN \to \Sigma_{\RR^n}$ by $m \mapsto \overline{A(m) \setminus \bigcup_{i=0}^{m-1} A(i)}$ where $\overline{S}$ denotes the closure of $S$.
  Then the collection $\{ A'(m) \mid m \in \NN \}$ is pairwise almost disjoint, and each $A'(m)$ can be written as a finite union of boxes, proving the claim.
\end{proof}

\begin{lemma}\label{lem:constraint-set-interval-approx}
  Let $\calV$ a symbolic value of ground type containing each sample variable $\alpha_1,\dots,\alpha_n$ at most once and $[x,y]$ an interval.
  Then there is a countable set of pairwise disjoint interval traces $\traceset \subset \iv_{[0,1]}^n$ such that
  \[ \bigcup_{\ivtr \in \traceset} \tracesin{\ivtr} \tightsubset \{ \tr \in [0,1]^n \mid \calV[\tr/\overline \alpha] \in [x,y] \}. \]
\end{lemma}
\begin{proof}
  If $\calV$ is of ground type, then it is simply a composition of primitive functions applied to sample variables and literals.
  Since the set of primitive functions is closed under composition and since no sample variable occurs twice, this composition is still an interval separable function $f$ of the sample variables.
  By definition of interval separability, there is a countable set of interval traces $\mathcal J$ such that $\bigcup_{\ivtr \in \mathcal J} \tracesin{\ivtr} \tightsubset f^{-1}([x,y])$, as desired.
\end{proof}

\begin{lemma}\label{lem:approximate-weight-integral-sum}
  Let $\ivtr \in \iv^n$ be an interval trace and $\Xi$ a set of symbolic values with sample variables from $\overline \alpha = \alpha_1,\dots,\alpha_n$.
  Then for any $\epsilon > 0$, there is a countable subdivision $\traceset$ of $\ivtr$ such that
  \begin{align*}
    &\sum_{\ivtr' \in \traceset} \volume(\ivtr') \sup_{\tr \in \tracesin{\ivtr'}} \prod_{\calW \in \Xi} \calW[\tr/\overline \alpha] - \epsilon \\
    &\le \int_{\tracesin{\ivtr}} \prod_{\calW \in \Xi} \calW[\tr/\overline \alpha] \D \tr \\
    &\le \sum_{\ivtr' \in \traceset} \volume(\ivtr') \min_{\tr \in \tracesin{\ivtr'}} \prod_{\calW \in \Xi} \calW[\tr/\overline \alpha] + \epsilon.
  \end{align*}
\end{lemma}
\begin{proof}
  Values are simply boxwise continuous functions applied to the sample variables.
  Intersecting the boxes for each $\calW \in \Xi$, we see that the function
  \[ f: \tracesin{\ivtr} \to \RR, \quad \tr \mapsto \prod_{\calW \in \Xi} \calW[\tr/\overline \alpha] \]
  is boxwise continuous.
  We can thus find a countable subdivision $\traceset_\mathsf{cont}$ of $\ivtr$ such that $f$ is continuous on each $\ivtr' \in \traceset_\mathsf{cont}$.

  Since we can sum over the $\ivtr' \in \traceset_\mathsf{cont}$, it suffices to prove that each integral $\int_{\tracesin{\ivtr'}} f(\tr) \D \tr$ can be approximated arbitrarily closely.
  Note that each such integral is finite because a continuous function is bounded on a compact set and the measure of $\tracesin{\ivtr'}$ is finite.
  But then such approximations are given by Riemann sums, i.e.~the sums that are used to define the Riemann integral.
  As a concrete example, one can consider the subdivision $\traceset_m$ of $\ivtr'$ in $m$ equidistant sections in each dimension (consisting of $m^n$ parts overall).
  Then $\sum_{\ivtr'' \in \traceset_m} \volume(\ivtr'') \min_{\tr \in \tracesin{\ivtr''}} f(\tr)$ converges to the Riemann integral $\int_{\tracesin{\ivtr'}} f(\tr) \D \tr$ as $m \to \infty$ (and similarly for the supremum).
  Since it is known that the Riemann integral and the Lebesgue integral have the same value for continuous functions on a Cartesian product of compact intervals, the claim follows immediately.
\end{proof}

\begin{lemma}[Relationship between symbolic execution and interval semantics]
  Let $\Psi = \symPath{\calV, n, \Delta, \Xi}$ be a symbolic path of $P$ and $\ivtr$ an interval trace with $\tracesin{\ivtr} \subseteq \Sat_n(\Delta)$.
  Suppose furthermore that all the symbolic values contain each of the sample variables $\overline \alpha = \alpha_1, \dots, \alpha_n$ at most once.
  Then there is a subdivision $\traceset$ of $\ivtr$ such that for all $\ivtr' \in \traceset$, the interval semantics for the value is precise:
  \[ \valueSem P^\iv(\ivtr') = \{ \valueSem P(\tr) \mid \tr \in \tracesin{\ivtr'} \}. \]

  For each symbolic score value $\calW \in \Xi$, let $[\calW_{\ivtr'}^{-}, \calW_{\ivtr'}^{+}]$ be its interval approximation, i.e. $\calW_{\ivtr'}^{-} = \min_{\tr \in \tracesin{\ivtr'}} \calW[\tr/\overline \alpha]$ and $\calW_{\ivtr'}^{+} = \sup_{\tr \in \tracesin{\ivtr'}} \calW[\tr/\overline \alpha]$. Then
  \[ \weightSem P^\iv(\ivtr') = \left[ \prod_{\calW \in \Xi} \calW_{\ivtr'}^{-}, \prod_{\calW \in \Xi} \calW_{\ivtr'}^{+} \right]. \]
  \label{lem:relationship-symbolic-interval}
\end{lemma}
\begin{proof}
  The symbolic value $\calV$ is a composition of primitive functions applied to $\alpha$'s.
  Hence the are boxwise continuous functions of the $\alpha$'s.
  We pick a suitable subdivision $\traceset$ such that all these functions are continuous when restricted to any $\ivtr' \in \traceset$.
  For any such function $f$, we have
  \[ f_\iv([x_1,y_1], \dots, [x_m,y_m]) = [\inf F, \sup F] \]
  where $F := f([x_1,y_1] \times \dots \times [x_m,y_m]) \subseteq \RR$, by definition.
  Then continuity  implies that the image of any box is a compact and path-connected subset of $\RR$, i.e. an interval.
  Hence we even have $f_\iv([x_1,y_1], \dots, [x_m,y_m]) = F$, i.e. the image of any box equals its interval approximation, proving the claim about the value semantics.

  For the interval semantics of the weight, note that the previous argument applies to every symbolic score value $\calW \in \Xi$, proving that
  \[ \{ \calW[\tr/\overline \alpha] \mid \tr \in \tracesin{\ivtr'} \} = [\calW_{\ivtr'}^{-}, \calW_{\ivtr'}^{+}]. \]
  Since the interval semantics multiplies the interval approximation of each score value in interval arithmetic, this implies the claim.
\end{proof}

Note that the interval approximation of the weight is imprecise in the following sense:
  \[ \weightSem P^\iv(\ivtr') \ne \left\{ \prod_{\calW \in \Xi} \calW[\tr/\overline \alpha] \,\middle\vert\, \tr \in \tracesin{\ivtr'} \right\}. \]
As an example, if $\Xi = \{ \alpha_1, 1 - \alpha_1 \}$ and $\ivtr'' = \langle [0,1] \rangle$ then the left-hand side is $[0,1]$ because each of the weights is approximated by $[0,1]$, but the right-hand side is $[0,1/4]$ because the function $\alpha_1(1-\alpha_1)$ attains its maximum at $1/4$, not $1$.

\begin{lemma}
  Let $\Psi = \symPath{\calV, n, \Delta, \Xi}$ be a symbolic path of $P$ and $\ivtr$ an interval trace with $\tracesin{\ivtr} \subseteq \Sat_n(\Delta)$.
  Suppose furthermore that all the symbolic values contain each of the sample variables $\overline \alpha = \alpha_1, \dots, \alpha_n$ at most once.
  Then for all $\epsilon > 0$, there is a subdivision $\traceset$ of $\ivtr$ such that for all $\ivtr' \in \traceset$, we have $\min_{\tr \in \tracesin{\ivtr'}} \weightSem P(\tr) \le \min \weightSem P^\iv(\ivtr') + \epsilon$ and $\sup_{\tr \in \tracesin{\ivtr'}} \weightSem P(\tr) \ge \sup \weightSem P^\iv(\ivtr') - \epsilon$.
  \label{lem:approximate-weight-intervals}
\end{lemma}
\begin{proof}
  Since for each $\calW \in \Xi$, the function $f: \tracesin{\ivtr} \to \RR, \tr \mapsto \calW[\tr/\overline \alpha]$ is boxwise continuous (a property of primitive functions), we can find a countable subdivision $\traceset'$ of $\ivtr$, such that for all $\ivtr' \in \traceset'$, $f$ is continuous on $\tracesin{\ivtr'}$.
  Hence it suffices to prove the statement for each $\ivtr'$.

  Since $\tracesin{\ivtr'}$ is compact (because it's closed and bounded), $f$ attains a maximum $W < \infty$ on $\tracesin{\ivtr'}$ and is even uniformly continuous on $\ivtr'$.
  Hence there is a $\delta > 0$ such that whenever $||\tr - \tr'|| < \delta$ then $|f(\tr) - f(\tr')| < \epsilon' := \frac{\epsilon}{W^{|\Xi|-1}}$.

  Let $\traceset$ be a subdivision where each interval trace $\ivtr \in \traceset$ has diameter less than $\delta$.
  For $\ivtr \in \traceset$ and $\calW \in \Xi$, let $\calW_{\ivtr}^{-} := \min_{\tr \in \tracesin{\ivtr}} \calW[\tr/\alpha]$ and $\calW_{\ivtr}^{+} := \sup_{\tr \in \tracesin{\ivtr}} \calW[\tr/\alpha]$.
  By the choice of $\traceset$, we have $\calW_{\ivtr}^{+} \le \calW_{\ivtr}^{-} + \epsilon'$ for $\ivtr \in \traceset$.
  By \cref{lem:relationship-symbolic-interval}, we find that $\sup \weightSem P^\iv(\ivtr) = \prod_{\calW \in \Xi} \calW_{\ivtr}^{+}$ and $\min \weightSem P^\iv(\ivtr) = \prod_{\calW \in \Xi} \calW_{\ivtr}^{-}$.
  As a consequence, we have
  \begin{align*}
    \sup \weightSem P^\iv(\ivtr) - \min \weightSem P^\iv(\ivtr) &= \prod_{\calW \in \Xi} \calW_{\ivtr}^{+} - \prod_{\calW \in \Xi} \calW_{\ivtr}^{-} \\
    &< \prod_{\calW \in \Xi} (\calW_{\ivtr}^{-} + \epsilon') - \prod_{\calW \in \Xi} \calW_{\ivtr}^{-} \\
    &< \epsilon' W^{|\Xi| - 1} = \epsilon
  \end{align*}

  So the interval $\weightSem P^\iv(\ivtr)$ has diameter less than $\epsilon$.
  Since the interval $\{ \weightSem P(\tr) \mid \tr \in \tracesin{\ivtr'} \}$ is contained in it (by soundness), the claim follows.
\end{proof}

\corollaryCompleteness*
\begin{proof}
By \cref{thm:completeness}, we can find for each $n \in \NN$ a set of interval traces $\traceset'_n$ such that $\lowerBound P^{\traceset'_n}(I) > \measureSem P(I) - 1/n$.
Since the lower bound is defined as a sum over $\traceset'_n$, there is a finite subset $\traceset_n$ such that $\lowerBound P^{\traceset_n}(I) > \measureSem P(I) - 2/n$.
Since $\traceset'_n$ is still compatible, the soundness result yields $\lowerBound P^{\traceset_n}(I) \le \measureSem P(I)$, implying the claim.
\end{proof}

Finitely many interval traces are not enough for complete upper bounds, if the weight function is unbounded.
This issue arises even if we can compute the tightest possible bounds on the weight function, as the following program illustrates.

\begin{example}
Consider the following probabilistic program expressed in pseudocode.

\algdef{SE}[SUBALG]{Indent}{EndIndent}{}{\algorithmicend\ }%
\algtext*{Indent}
\algtext*{EndIndent}

\parbox{\linewidth-10pt}{
	\begin{algorithmic}
		\State $\mathit{threshold}$ := $1$
		\State \textbf{while} $(\sample \leq \mathit{threshold})$ \textbf{do}
		\Indent
		\State $\mathit{threshold}$ := $\tfrac{\mathit{threshold}}{2}$
		\State \textbf{score}$(2)$
		\EndIndent
	\end{algorithmic}
}
\noindent
The program only requires addition and scalar multiplication.
It can even be implemented using call-by-name (CbN) semantics (which allows each sampled value to be used at most once).
For example in SPCF we can write
\[ P \equiv \big(\fixLam \varphi s \ifElse{\sample - s}{\score(2); \varphi (s / 2)} 1\big) \, 1 \]
The program $P$ has the weight function
\[
  \weightSem{P}(\langle t_0, \dots, t_n \rangle) = \begin{cases}
    2^n &\text{if } t_n > 2^{-n} \land \\
    &\quad \forall i \in \{0, \dots, n-1\}: t_i \le 2^{-i} \\
    0 &\text{otherwise.}
  \end{cases}
\]
$P$ is integrable because the normalizing constant is
\begin{align*}
  Z &= \int_\traces \weightSem{P}(\tr) \D \tr
  = \sum_{n = 1}^\infty 2^n \times (1-2^{-n}) \prod_{i=0}^{n-1} 2^{-i} \\
  &= \sum_{n = 1}^\infty 2^n(1-2^{-n}) 2^{-n(n-1)/2}
  < \infty.
\end{align*}
We claim that $P$ requires infinitely many interval traces for the upper bound to converge to the true denotation.
Define the sets of traces $T_n$ for $n \ge 1$ by
\[ T_n = [0, 2^0] \times [0, 2^{-1}] \times \dots \times [0, 2^{-n+1}] \times (2^{-n}, 1]. \]
Suppose we are given an arbitrary finite exhaustive set of interval traces.
This set needs to cover all of the $T_n$'s, so one interval trace, say $\ivtr$, must cover infinitely many $T_n$'s.
Since $\weightSem P(\tr) = 2^n$ for $\tr \in T_n$, the weight function on $\ivtr$ is unbounded.
Therefore, the only possible upper bound for $\ivtr$ is $\infty$, \emph{even if our semantics could compute the set $\{ \weightSem P(\tr) \mid \tr \in \tracesin{\ivtr} \}$ exactly}.
Hence any finite exhaustive interval trace has upper bound $\infty$, while the true denotation is finite.
As we have seen, this is not because of imprecision of interval analysis, but an inherent problem if the weight function is unbounded.
So we cannot hope for complete upper bounds with finitely many interval traces.
\label{rem:countable-traces}
\end{example}


\section{Supplementary Material for \texorpdfstring{\cref{sec:5interval-analysis}}{Section~5}}
\label{app:sec5}

\begin{figure*}[!t]
	\small
	\begin{minipage}{0.25\linewidth}
		\vspace{6mm}
		\begin{prooftree}
			\def\ScoreOverhang{1pt}
			\AxiomC{$x:\sigma \in \Gamma$}
			\RightLabel{\small(\textsc{Var})}
			\UnaryInfC{$\Gamma \vdash x: \exType{\sigma}{\mathbf{1}} $}
		\end{prooftree}
	\end{minipage}%
	\begin{minipage}{0.25\linewidth}
		\begin{prooftree}
			\def\ScoreOverhang{1pt}
			\def\defaultHypSeparation{\hskip .1in}
			\AxiomC{$\Gamma \vdash M : \calA$}
			\AxiomC{$\calA \sqsubseteq_\calA \calB$}
			\RightLabel{\small(\textsc{Sub})}
			\BinaryInfC{$\Gamma \vdash M : \calB$}
		\end{prooftree}
	\end{minipage}%
	\begin{minipage}{0.25\linewidth}
		\vspace{8mm}
		\begin{prooftree}
			\def\ScoreOverhang{1pt}
			\AxiomC{$\Gamma; x:\sigma \vdash M : \calA$}
			\RightLabel{\small(\textsc{Abs})}
			\UnaryInfC{$\Gamma \vdash \lambda x. M : \exType{\sigma \to \calA}{\mathbf{1}} $}
		\end{prooftree}
	\end{minipage}%
	\begin{minipage}{0.25\linewidth}
		\vspace{5mm}
		\begin{prooftree}
			\def\ScoreOverhang{1pt}
			\AxiomC{$\Gamma; \varphi: \sigma \to \calA ; x:\sigma \vdash M : \calA$}
			\RightLabel{\small(\textsc{Fix})}
			\UnaryInfC{$\Gamma \vdash \fixLam{\varphi}{x} M : \exType{\sigma \to \calA}{\mathbf{1}} $}
		\end{prooftree}
	\end{minipage}

	\vspace{0.0cm}

	\vspace{0.1cm}

	\begin{minipage}{0.4\linewidth}
		\vspace{6.5mm}
		\begin{prooftree}
			\def\ScoreOverhang{1pt}
			\AxiomC{}
			\RightLabel{\small(\textsc{Sample})}
			\UnaryInfC{$\Gamma \vdash \sample : \exType{\myint{0, 1}}{\mathbf{1}} $}
		\end{prooftree}
	\end{minipage}%
	\begin{minipage}{0.6\linewidth}
		\begin{prooftree}
			\def\ScoreOverhang{1pt}
			\def\defaultHypSeparation{\hskip .15in}
			\AxiomC{$\Gamma\vdash M : \exType{\sigma_1 \to \exType{\sigma_2}{\myint{e, f}}}{\myint{a, b}}$}
			\AxiomC{$\Gamma \vdash N : \exType{\sigma_1}{\myint{c, d}}$}
			\RightLabel{\small(\textsc{App})}
			\BinaryInfC{$\Gamma \vdash M N : \exType{\sigma_2}{\myint{a, b} \times^\iv \myint{c, d} \times^\iv \myint{e, f}}$ }
		\end{prooftree}
	\end{minipage}

	\begin{minipage}{0.4\linewidth}
		\vspace{6.5mm}
		\begin{prooftree}
			\def\ScoreOverhang{1pt}
			\AxiomC{}
			\RightLabel{\small(\textsc{Lit})}
			\UnaryInfC{$\Gamma \vdash \lit{r} : \exType{\myint{r, r}}{\mathbf{1}} $}
		\end{prooftree}
	\end{minipage}%
	\begin{minipage}{0.6\linewidth}
		\begin{prooftree}
			\def\defaultHypSeparation{\hskip .15in}
			\def\ScoreOverhang{1pt}
			\AxiomC{$\Gamma \vdash M : \exType{\myint{\_, \_}}{\myint{a, b}}$}
			\AxiomC{$\Gamma \vdash N :\exType{ \sigma}{\myint{c, d}}$}
			\AxiomC{$\Gamma \vdash P : \exType{\sigma}{\myint{c, d}}$}
			\RightLabel{\small(\textsc{If})}
			\TrinaryInfC{$\Gamma \vdash \ifSimple M N P : \exType{\sigma}{\myint{a,b} \times^\iv \myint{c, d}} $}
		\end{prooftree}
	\end{minipage}

	\vspace{0.1cm}

	\begin{minipage}{0.5\linewidth}
		\begin{prooftree}
			\def\ScoreOverhang{1pt}
			\AxiomC{$\Gamma \vdash M : \exType{\myint{a, b}}{\myint{c, d}}$}
			\RightLabel{\small(\textsc{Score})}
			\UnaryInfC{$\Gamma \vdash \score(M) : \exType{\myint{a, b} \sqcap \myint{0, \infty}}{\myint{c, d} \times^\iv \big(\myint{a, b} \sqcap \myint{0, \infty}\big) } $}
		\end{prooftree}
	\end{minipage}%
	\begin{minipage}{0.5\linewidth}
		\begin{prooftree}
			\def\ScoreOverhang{1pt}
			\AxiomC{$\Gamma \vdash M_1 : \exType{\myint{a_1,b_1}}{\myint{c_1, d_1}}$}
			\AxiomC{$\cdots$}
			\AxiomC{$\Gamma \vdash M_{|f|} : \exType{\myint{a_{|f|}, b_{|f|}}}{\myint{c_{|f|}, d_{|f|}}}$}
			\RightLabel{\small(\textsc{Prim})}
			\TrinaryInfC{$\Gamma \vdash f(M_1, \dots, M_{|f|}) : \exType{f^\iv(\myint{a_1,b_1}, \dots, \myint{a_{|f|}, b_{|f|}})}{(\times^\iv)_{i=1}^{|f|} \myint{c_i, d_i}} $}
		\end{prooftree}
	\end{minipage}

	\caption{Weight-aware interval type system for SPCF with typing rule names. The rules agree with those in \cref{fig:typeSystemSelection}. } \label{fig:typeSystemSelection2}
\end{figure*}

We provide additional proofs and material for \cref{sec:5interval-analysis}.
To have access to named rules, we give the type system in \cref{fig:typeSystemSelection2} which agrees with the one in \cref{fig:typeSystemSelection} in everything but the labels.

\subsection{Soundness}

To show soundness (\cref{thm:staticRes}), we establish a (weight-aware) subject reduction property for our type system as follows.
For an interval $[a, b] \in \iv$ and $r \in \RR_{\ge 0}$, we define $r \cdot [a, b] := [r \cdot a, r \cdot b]$.
To simplify notation, we use a modified transition relation that omits the concrete trace (which is irrelevant in \cref{thm:staticRes}).
We write $P \to_w P'$ if $\stdConf{P, \tr, 1} \to \stdConf{P', \tr', w}$ for some $w \in \RR$ and $\tr, \tr' \in \traces$.
Note that we could define $\to_w$ as a dedicated reduction system by adapting the rules from $\to$ in \cref{fig:reductionRules}.

\begin{lemma}[Substitution]\label{lem:type-system-subst}
	If $\Gamma; \{x_i : \sigma_i\}_{i=1}^n \vdash P : \calA$ for distinct variables $x_i$ and if $\Gamma \vdash M_i : \exTypeS{\sigma_i}{\mathbf{1}}$ for all $i \in \{1, \dots, n\}$ then $\Gamma \vdash P[M_i/x_i]_{i=1}^n : \calA$.
\end{lemma}
\begin{proof}
	By a standard induction on $M$.
\end{proof}

\begin{lemma}[Weighted Subject Reduction]\label{lem:type-system-subject-reduction}
	Let $P$ be any program such that $\vdash P : \exTypeS{\sigma}{J}$ and $P \to_w P'$ for some $w > 0$.
	Then $\vdash P' : \exTypeS{\sigma}{\tfrac{1}{w} \cdot J}$.
\end{lemma}
\begin{proof}
	We prove this by induction on the structure of $P$.

	\textbf{Case $P = \sample$:} then $P \to_1 \lit r$ for some $r \in [0, 1]$.
	As multiple consecutive application of (Sub) can be replaced by a single one since (as $\sqsubseteq_\calA$ is transitive), we can assume w.l.o.g. that the last step in $\vdash P : \exTypeS{\sigma}{J}$ was:
	\begin{prooftree}
		\def\defaultHypSeparation{\hskip .15in}
		\def\ScoreOverhang{1pt}
		\AxiomC{}
		\RightLabel{\small(\textsc{Sample})}
		\UnaryInfC{$\vdash \sample : \exTypeS{[0, 1]}{\mathbf{1}}$}
		\RightLabel{\small(\textsc{Sub})}
		\UnaryInfC{$\vdash \sample : \exTypeS{I}{J}$}
	\end{prooftree}
	By subtyping we have $[0,1] \sqsubseteq I$ and $\mathbf{1} \sqsubseteq J$.
	So $[r, r] \sqsubseteq I$ and we can type $\vdash \lit r : \exTypeS{I}{J}$ using (\textsc{Lit}) and (\textsc{Sub}) as required.

	\textbf{Case $P = f(\lit{r_1}, \dots,\lit{ r_{|f|}})$:} then $P \to_1 \lit{f(r_1, \dots, r_{|f|})}$.
	W.l.o.g., we can assume that the last step in $\vdash P : \exTypeS{\sigma}{J}$ was
	\begin{prooftree}
		\def\defaultHypSeparation{\hskip .15in}
		\def\ScoreOverhang{1pt}
		\AxiomC{}
		\RightLabel{\small(\textsc{Lit})}
		\UnaryInfC{$\vdash \lit{r_1} : \exTypeS{[r_1, r_1]}{\mathbf{1}}$}
		\RightLabel{\small(\textsc{Sub})}
		\UnaryInfC{$\vdash \lit{r_1} : \exTypeS{I_1}{J_1}$}
		\AxiomC{$\cdots$}
		\AxiomC{}
		\RightLabel{\small(\textsc{Lit})}
		\UnaryInfC{$\vdash \lit{r_{|f|}} : \exTypeS{[r_{|f|}, r_{|f|}]}{\mathbf{1}}$}
		\RightLabel{\small(\textsc{Sub})}
		\UnaryInfC{$\vdash \lit{r_{|f|}} : \exTypeS{I_{|f|}}{J_{|f|}}$}
		\RightLabel{\small(\textsc{Prim})}
		\TrinaryInfC{$\vdash f(\lit{r_1}, \dots,\lit{ r_{|f|}}) : \exTypeS{f^\iv(I_1, \dots, I_{|f|})}{(\times^\iv)_{i=1}^{|f|} J_i}$}
		\RightLabel{\small(\textsc{Sub})}
		\UnaryInfC{$\vdash f(\lit{r_1}, \dots,\lit{ r_{|f|}}) : \exTypeS{I}{J}$}
	\end{prooftree}
	So by subtyping $r_i \in I_i$ and $1 \in J_i$ for all $i$.
	By definition of $f^\iv$ we thus have $f(r_1, \dots, r_{|f|}) \in f^\iv(I_1, \dots, I_{|f|})$ and (again by subtyping) we have $f(r_1, \dots, r_{|f|}) \in I$. Similarly, by definition of $\times^\iv$ we have $1 \in (\times^\iv)_{i=1}^{|f|} J_i$ and thus $1 \in J$.
	We can type $\vdash \lit{f(r_1, \dots, r_{|f|})} : \exTypeS{I}{J}$ using (\textsc{Lit}) and (\textsc{Sub}) as required.

	\textbf{Case $P = \ifElse{\lit r}{M}{N}$ and $r \leq 0$:} then $P \to_1 M$.
	W.l.o.g., we can assume that the last step in $\vdash P : \exTypeS{\sigma}{J}$ is
	\begin{prooftree}
		\def\defaultHypSeparation{\hskip .15in}
		\def\ScoreOverhang{1pt}
		\AxiomC{}
		\RightLabel{\small(\textsc{Lit})}
		\UnaryInfC{$\vdash \lit{r} : \exTypeS{[r, r]}{\mathbf{1}}$}
		\RightLabel{\small(\textsc{Sub})}
		\UnaryInfC{$\vdash \lit{r} : \exTypeS{I}{J''}$}
		\AxiomC{$\vdash M : \exTypeS{\sigma'}{J'}$}
		\AxiomC{$\vdash N : \exTypeS{\sigma'}{J'}$}
		\RightLabel{\small(\textsc{If})}
		\TrinaryInfC{$\vdash \ifElse{\lit r}{M}{N} : \exTypeS{\sigma'}{J' \times^\iv J''}$}
		\RightLabel{\small(\textsc{Sub})}
		\UnaryInfC{$\vdash \ifElse{\lit r}{M}{N} : \exTypeS{\sigma}{J}$}
	\end{prooftree}
	and we have $1 \in J''$, $\sigma' \sqsubseteq_\sigma \sigma$ and $J' \times^\iv J'' \sqsubseteq J$ by subtyping.
	As $1 \in J''$, we get $J' \sqsubseteq J' \times^\iv J''$.
	Thus we obtain $J' \sqsubseteq J$ and $\sigma' \sqsubseteq_\sigma \sigma$, and we can type $\vdash M : \exTypeS{\sigma}{J}$ using (\textsc{Sub}).

	\textbf{Case $P = \ifElse{\lit r}{M}{N}$ and $r > 0$:} then $P \to_1 N$. Analogous to the previous case.

	\textbf{Case $P = \score(\lit{r})$ and $r \geq 0$:} then $P \to_r \lit{r}$.
	W.l.o.g., we can assume that the last step in $\vdash P : \exTypeS{\sigma}{J}$ is
	\begin{prooftree}
		\def\defaultHypSeparation{\hskip .15in}
		\def\ScoreOverhang{1pt}
		\AxiomC{}
		\RightLabel{\small(\textsc{Lit})}
		\UnaryInfC{$\vdash \lit r : \exTypeS{[r,r]}{\mathbf{1}}$}
		\RightLabel{\small(\textsc{Sub})}
		\UnaryInfC{$\vdash \lit r : \exTypeS{I'}{J'}$}
		\RightLabel{\small(\textsc{Score})}
		\UnaryInfC{$\vdash \score(\lit r) : \exTypeS{I' \sqcap [0, \infty]}{J' \times^\iv (I' \sqcap [0, \infty])}$}
		\RightLabel{\small(\textsc{Sub})}
		\UnaryInfC{$\vdash \score(\lit r) : \exTypeS{I}{J}$}
	\end{prooftree}
	By subtyping we have $r \in I'$ and even $r \in I' \sqcap [0, \infty]$ because $r \geq 0$.
	Thus $r \in I$, again by subtyping.
	Similarly, $1 \in J'$ and by definition of $\times^\iv$, we have $r \in J' \times^\iv (I' \sqcap [0, \infty])$.
	Thus $r \in J$ by subtyping.
	This already implies $1 \in \tfrac{1}{r} \cdot J$ and we can thus type $\vdash \lit{r} : \exTypeS{I}{\tfrac{1}{r} \cdot J}$ by using (\textsc{Lit}) and (\textsc{Sub}).

	\textbf{Case $P = (\lambda x. M) V$:} then $P \to_1 M[V/x]$.
	W.l.o.g., the last step in $\vdash P : \exTypeS{\sigma}{J}$ is
	\begin{prooftree}
		\def\defaultHypSeparation{\hskip .15in}
		\def\ScoreOverhang{1pt}
		\AxiomC{$\{x : \tilde{\sigma}' \} \vdash M : \exTypeS{\tilde{\sigma}''}{\tilde{J}''}$}
		\RightLabel{\small(\textsc{Abs})}
		\UnaryInfC{$\vdash \lambda x. M : \exTypeS{\tilde{\sigma}' \to \exTypeS{\tilde{\sigma}''}{\tilde{J}''}}{\mathbf{1}}$}
		\RightLabel{\small(\textsc{Sub})}
		\UnaryInfC{$\vdash \lambda x. M : \exTypeS{\sigma' \to \exTypeS{\sigma''}{J''}}{J'}$}
		\AxiomC{$\vdash V : \exTypeS{\sigma'}{J'''}$}
		\RightLabel{\small(\textsc{App})}
		\BinaryInfC{$\vdash (\lambda x. M) V: \exTypeS{\sigma''}{J' \times^\iv J'' \times^\iv J'''}$}
		\RightLabel{\small(\textsc{Sub})}
		\UnaryInfC{$\vdash (\lambda x. M) V: \exTypeS{\sigma}{J}$}
	\end{prooftree}
	By subtyping we get that $\sigma' \sqsubseteq_\sigma \tilde{\sigma}'$.
	It is easy to see that we can also type $\vdash V : \exTypeS{\sigma'}{\mathbf{1}}$ because $V$ is a value and we get $\vdash V : \exTypeS{\tilde{\sigma}'}{\mathbf{1}}$ by (\textsc{Sub}).
	Using \cref{lem:type-system-subst}, we can thus type $\vdash M[V/x] : \exTypeS{\tilde{\sigma}''}{\tilde{J}''}$.

	We have $1 \in J'$ by subyptying and as $V$ is a value, it is easy to see that $1 \in J'''$.
	Hence $J'' \sqsubseteq J' \times^\iv J'' \times^\iv J'''$.
	Also by subtyping, we find $\tilde{\sigma}'' \sqsubseteq_\sigma \sigma'' \sqsubseteq_\sigma \sigma$, $\tilde{J}'' \sqsubseteq J''$, and $J' \times^\iv J'' \times^\iv J''' \sqsubseteq J$.
	This implies $\tilde{\sigma}'' \sqsubseteq_\sigma \sigma$ and $\tilde{J}'' \sqsubseteq J$.
	By subtyping $\vdash M[V/x] : \exTypeS{\sigma}{J}$, as required.

	\textbf{Case $P = (\mu^\varphi_x. M) V$:} then $P \to_1 M[V/x, (\mu^\varphi_x. M)/\varphi]$. Analogous to the previous case for abstractions.

	\textbf{Case $P = E[P']$ for an evaluation context $E \neq [\cdot]$:} then $P' \to_r P''$ and $P \to_r E[P'']$.
	All such cases follow easily by case analysis on $E$.
	As an example, consider the context $E = [\cdot] N$.
	In this situation, we have $P = E[P'] = P' N$ with $P' \to_r P''$, so $P \to_r P'' N$.
	W.l.o.g., the last step in $\vdash P : \exTypeS{\sigma}{J}$ is
	\begin{prooftree}
		\def\defaultHypSeparation{\hskip .15in}
		\def\ScoreOverhang{1pt}
		\AxiomC{$\vdash P' : \exTypeS{\sigma' \to \exTypeS{\sigma''}{J''}}{J'}$}
		\AxiomC{$\vdash N  : \exTypeS{\sigma'}{J'''}$}
		\RightLabel{\small(\textsc{App})}
		\BinaryInfC{$\vdash P' N: \exTypeS{\sigma''}{J' \times^\iv J'' \times^\iv J'''}$}
		\RightLabel{\small(\textsc{Sub})}
		\UnaryInfC{$\vdash P' N: \exTypeS{\sigma}{J}$}
	\end{prooftree}
	By the inductive assumption for $P' \to_r P''$, we get $\vdash P'' : \exTypeS{\sigma' \to \exTypeS{\sigma''}{J''}}{\tfrac{1}{r} \cdot J'}$ and can then type
	\begin{prooftree}
		\def\defaultHypSeparation{\hskip .15in}
		\def\ScoreOverhang{1pt}
		\AxiomC{$\vdash P'' : \exTypeS{\sigma' \to \exTypeS{\sigma''}{J''}}{\tfrac{1}{r} \cdot J'}$}
		\AxiomC{$\vdash N  : \exTypeS{\sigma'}{J'''}$}
		\RightLabel{\small(\textsc{App})}
		\BinaryInfC{$\vdash P' N: \exTypeS{\sigma''}{\tfrac{1}{r} \cdot J' \times^\iv J'' \times^\iv J'''}$}
		\RightLabel{\small(\textsc{Sub})}
		\UnaryInfC{$\vdash P' N: \exTypeS{\sigma}{\tfrac{1}{r} \cdot J}$}
	\end{prooftree}
	because if $J' \times^\iv J'' \times^\iv J''' \sqsubseteq J$ then $\tfrac{1}{r} \cdot J' \times^\iv J'' \times^\iv J''' \sqsubseteq \tfrac{1}{r} \cdot J$.

	The proof for the other evaluation contexts, i.e., where $P = \ifElse{P'}{M}{N}$, $P = V P'$, $P = \score(P')$, or $P = \lit f(r_1, \dots, r_{i-1},\allowbreak P',\allowbreak N_{i+1}, \dots, N_{|f|})$ for some $P' \to_r P''$, are all analogous to the above.
\end{proof}

\begin{lemma}[Zero-Weighted Subject Reduction]\label{lem:type-system-subject-reduction-zero}
	Let $P$ be any program such that $\vdash P : \exTypeS{\sigma}{J}$ and $P \to_0 P'$.
	Then
	\begin{enumerate}
		\item $\vdash P' : \exTypeS{\sigma}{J'}$ for some $J'$, and\label{lem:type-system-subject-reduction-zero1}
		\item $0 \in J$ \label{lem:type-system-subject-reduction-zero2}
	\end{enumerate}
\end{lemma}
\begin{proof}
	The proof that $\vdash P' : \exTypeS{\sigma}{J'}$ for some $J'$ is analogous to the proof of \cref{lem:type-system-subject-reduction} with fewer restrictions on the weight.
	The claim $0 \in J$ follows by observing that $P \to_0 P'$ is only possible if the redex in $P$ is $\score (\lit 0)$.
	In case $P = \score (\lit 0)$, the claim follows directly from (\textsc{Score}).
	If $P = E[\score (\lit 0)]$, it is a simple induction on the structure of the evaluation context $E$.
\end{proof}

\typeSystemSoundness*
\begin{proof}
	Let
	\begin{align*}
		\stdConf{P, \tr, 1} = \stdConf{P_0, \tr_0, w_0} \to \cdots \to \stdConf{P_n, \tr_n, w_n} = \stdConf{\lit r, \langle \rangle, w}
	\end{align*}
	be the reduction sequence of $\stdConf{P, \tr, 1}$.
	By definition of $\to_w$ it is easy to see that we get
	\begin{align*}
		P = P_0 \to_{\tilde{w}_1} P_1 \to_{\tilde{w}_2} \cdots \to_{\tilde{w}_n} P_n = \lit r
	\end{align*}
	for unique $\tilde{w}_1, \dots, \tilde{w}_n$.
	Note that $w_i = \prod_{j=1}^{i} \tilde{w}_j$.

	We first assume that $w \neq 0$ (so $\tilde{w}_i \neq 0$ for all $i$).
	We claim that for each $0 \leq i \leq n$, we can type
	\begin{align}\label{eq:type-system-induction}
		\vdash P_i : \exType{[a,b]}{\tfrac{1}{w_i} \cdot [c, d]}
	\end{align}
	\Cref{eq:type-system-induction} follows by simple induction:
	the base case $i = 0$ holds by the assumption $\vdash P_0 : \exTypeS{[a,b]}{[c, d]}$ and because $w_0 = 1$.
	For the inductive case, we assume that $\vdash P_i : \exTypeS{[a,b]}{\tfrac{1}{w_i} \cdot [c, d]}$.
	We apply \cref{lem:type-system-subject-reduction} to $P_i \to_{\tilde{w}_i} P_{i+1}$ and finish the induction step using $w_{i+1} = w_i \tilde{w}_{i+1}$:
	\begin{align*}
		\vdash P_{i+1} : \exTypeS{[a,b]}{\tfrac{1}{\tilde{w}_{i+1}} \cdot \tfrac{1}{w_i} \cdot [c, d]} = \exTypeS{[a,b]}{\tfrac{1}{w_{i+1}} \cdot [c, d]}
	\end{align*}

	\Cref{eq:type-system-induction} thus implies $\vdash \lit r : \exTypeS{[a, b]}{\tfrac{1}{w}\cdot [c, d]}$.
	W.l.o.g., we can assume that this type judgment has the form

	\begin{prooftree}
		\def\defaultHypSeparation{\hskip .15in}
		\def\ScoreOverhang{1pt}
	\AxiomC{}
	\RightLabel{\small(\textsc{Const})}
	\UnaryInfC{$\vdash \lit{r} : \exTypeS{[r, r]}{\mathbf{1}}$}
	\RightLabel{\small(\textsc{Sub})}
	\UnaryInfC{$\vdash \lit{r} : \exTypeS{[a, b]}{\tfrac{1}{w}\cdot [c, d]}$}
	\end{prooftree}

	This implies $r \in [a, b]$ and $w \in [c, d]$ (because $1 \in \tfrac{1}{w}\cdot [c, d]$), as required.

	Now consider the case $w = 0$, so at least one $\tilde{w}_i = 0$.
	The claim that $r \in [a, b]$ follows easily as above (now using \cref{lem:type-system-subject-reduction-zero1} of \cref{lem:type-system-subject-reduction-zero} to handle the weight 0 reduction steps).
	Let $i^*$ be the smallest index such that $\tilde{w}_{i^*} = 0$.
	Using the same argument as above on the (possibly empty) prefix up to index $i^* - 1$ (where all $\tilde{w}_j$ are non-zero) we find
	\begin{align*}
		\vdash P_{i^* - 1} : \exType{[a,b]}{\tfrac{1}{w_{i^*-1}} \cdot [c, d]}
	\end{align*}
	Note that this is well defined because $\tilde{w}_j > 0$ for $j < i^*$.
	\Cref{lem:type-system-subject-reduction-zero2} of \cref{lem:type-system-subject-reduction-zero} shows $0 \in \tfrac{1}{w_{i^* - 1}} \cdot [c, d]$, which already implies $w = 0 \in [c, d]$, as required.
\end{proof}

\subsection{Weak Completeness}

\typeSystemComp*
\begin{proof}
	For every \emph{simple} type $\alpha$, we define a weighted type $\calA_\alpha$ and weightless type $\sigma_\alpha$ by mutual recursion as follows.
	\begin{align*}
		\sigma_\typeReal &:= [-\infty, \infty]\\
		\sigma_{\alpha \to \beta} &:= \sigma_\alpha \to \calA_\beta\\
		\calA_\alpha &:= \exTypeS{\sigma_\alpha}{[0, \infty]}
	\end{align*}
	That is, we insert $[-\infty, \infty]$ for values and $[0,\infty]$ for weights in all locations.
	We claim that if $\vdash P : \alpha$ in the simple type system, then $\vdash P : \calA_\alpha$ in the weight-aware interval type system.

	For the proof, we strengthen the induction hypothesis by claiming: if $\Gamma \vdash P : \alpha$ in the simple type system then $\Gamma^{\uparrow \sigma} \vdash P : \calA_\alpha$ where
	\begin{align*}
		\Gamma^{\uparrow \sigma} := \{ x : \sigma_\alpha \mid x : \alpha \in \Gamma \}.
	\end{align*}
	The claim can be proved by simple induction on the derivation of $\Gamma \vdash P : \alpha$ using the respective typing rule for the interval type system possibly followed by (\textsc{Sub}) for typing rules that yield proper subtypes of $\calA_\alpha$.

    Note that this typing derivation does not contain any useful information to improve the precision of GuBPI.
\end{proof}

\subsection{Constraint-based Type Inference}

In this section we formalize the constraint-based type inference algorithm and sketch our constraints-solving method based on worklist and widening.
The overarching idea is to substitute intervals with variables $\nu_i$, called interval variables, and to encode typability as a constraint system.
As we work in the restricted interval domain (as opposed to e.g.~full first-order refinements), the resulting constraints can be solved very efficiently, which is crucial to the practicality of our tool.

\paragraph{Symbolic types}

Symbolic types are defined by the following grammar:
\[ \kappa := \nu_i \mid \kappa \to \mathscr{A} \qquad	\mathscr{A} := \exType{\kappa}{\nu_j} \]
where $\nu_i, \nu_j$ are interval variables.
Symbolic types are identical to interval types but use interval variables instead of intervals as first-order types and in the weight bound.

\paragraph{Constraints}

Constraints on interval variables come in three forms:
\[ c := \;\nu_n \equiv \myint{a, b} \Bigm\vert \kappa_1 \sqsubseteq \kappa_2 \Bigm\vert \nu_n \equiv \lit{f}(\nu_{n_1}, \dots, \nu_{n_{|f|}}) \]
where $f$ is a primitive function and $[a, b]$ an interval.
That is, a constraint can either equate an interval variable to a particular interval, require a symbolic type $\kappa_1$ to be a subtype of a type $\kappa_2$, or equate an interval variable to the result of a function applied to interval variables.
Note that due to the compositional nature of our subtype relation $\sqsubseteq$ (which extends to symbolic types) we can restrict ourself to constraints of the form $\nu_1 \sqsubseteq \nu_2$ because each constraint of the form $\kappa_1 \sqsubseteq \kappa_2$ or $\mathscr{A}_1 \sqsubseteq \mathscr{A}_2$ (with identical base types) can be reduced to an equivalent set of constraints on interval variables by the definition of the subtype relation.

\paragraph{Symbolic type system}

In the presentation of our symbolic type inference system, we aim to stay as close as possible to the implementation.
Thus we describe it as an impure type system, meaning that our typing rules have side effects.
In our case, typing rules can generate fresh interval variables.

For a simple type $\alpha$, we write $\mathit{fresh}(\alpha)$ for the symbolic weightless type obtained by replacing every base type $\typeReal$ with a fresh interval variable $\nu_n$ (and adding weights given by fresh interval variables where needed).
We write $\mathit{fresh}()$ for $\mathit{fresh}(\typeReal)$.
For a symbolic type $\kappa$ we write $\mathit{base}(\kappa)$ for the underlying simple type (defined in the obvious way).

Our constraint generation system is given in \cref{fig:symbolicTypeSystem}.
Judgments have the form $\Gamma \vdash M :  \mathscr{A}, \mathcal{C}$ where $\Gamma$ maps variables to weightless symbolic types, $\mathscr{A}$ is a weighted symbolic type and $\calC$ a list of constraints on the interval variables.
The rules follow the structure of the system in \cref{fig:typeSystemSelection} but replace all operations on intervals with interval variables and constraints.
The term structure directly determines the symbolic typing derivation; there are no choices to be made, contrary to \cref{fig:typeSystemSelection}, which requires ``cleverness'', for example to find a suitable interval for an argument in the fixpoint rule.
Note that in our system, we assume that the \emph{simple} types of arguments of abstractions and fixpoints are given.
These types can be determined by a simple prior run of any standard type inference algorithm.

\begin{figure*}[!t]
    \small
    \begin{minipage}{0.25\textwidth}
        \begin{prooftree}
            \AxiomC{$x:\kappa \in \Gamma$}
            \AxiomC{$\nu = \mathit{fresh}()$}
            \BinaryInfC{$\Gamma \vdash x: \exType{\kappa}{\nu}, \{\nu \equiv \mathbf{1}\}  $}
        \end{prooftree}
    \end{minipage}%
    \begin{minipage}{0.45\textwidth}
        \begin{prooftree}
            \AxiomC{$\kappa = \mathit{fresh}(\alpha)$}
            \AxiomC{$\nu = \mathit{fresh}()$}
            \AxiomC{$\Gamma; x:\kappa \vdash M : \mathscr{A}, \calC$}
            \TrinaryInfC{$\Gamma \vdash \lambda x^\alpha. M : \exType{\kappa \to\mathscr{A}}{\nu}, \calC \cup \{\nu \equiv \mathbf{1}\} $}
        \end{prooftree}
    \end{minipage}%
    \begin{minipage}{0.3\textwidth}
        \begin{prooftree}
            \AxiomC{$\Gamma \vdash M : \exType{\nu_1}{\nu_2}, \calC$}
            \AxiomC{$\nu = \mathit{fresh}()$}
            \BinaryInfC{$\Gamma \vdash \score(M) : \exType{\nu_1}{\nu}, \calC \cup \{\nu \equiv \nu_1 \times \nu_2\} $}
        \end{prooftree}
    \end{minipage}

    \begin{minipage}{0.3\textwidth}
        \begin{prooftree}
            \AxiomC{$\nu_1 = \mathit{fresh}()$}
            \AxiomC{$\nu_2 = \mathit{fresh}()$}
            \BinaryInfC{$\Gamma \vdash \lit{r} : \exType{\nu_1}{\nu_2}, \{\nu_1 \equiv [r, r], \nu_2 \equiv \mathbf{1}\} $}
        \end{prooftree}
    \end{minipage}%
    \begin{minipage}{0.7\textwidth}
        \begin{prooftree}
            \AxiomC{$\kappa = \mathit{fresh}(\alpha)$}
            \AxiomC{$\kappa_1 = \mathit{fresh}(\beta)$}
            \AxiomC{$\nu, \nu_1 = \mathit{fresh}()$}
            \AxiomC{$\Gamma; \varphi: \kappa \to \exType{\kappa_1}{\nu_1} ; x:\kappa \vdash M : \exType{\kappa_2}{\nu_2}, \calC$}
            \QuaternaryInfC{$\Gamma \vdash \fixLam{\varphi: \alpha \to \beta}{x} M : \exType{\kappa \to \exType{\kappa_2}{\nu_2}}{\nu}, \calC \cup \{\nu \equiv \mathbf{1}, \kappa_2 \sqsubseteq \kappa_1, \nu_2 \sqsubseteq \nu_1 \} $}
        \end{prooftree}
    \end{minipage}

    \begin{minipage}{0.35\textwidth}
        \begin{prooftree}
            \AxiomC{$\nu, \nu' = \mathit{fresh}()$}
            \UnaryInfC{$\Gamma \vdash \sample : \exType{\nu}{\nu'}, \{\nu \equiv [0, 1], \nu' \equiv \mathbf{1}\} $}
        \end{prooftree}
    \end{minipage}%
    \begin{minipage}{0.65\textwidth}
        \begin{prooftree}
            \AxiomC{$\Gamma\vdash M : \exType{\kappa_1 \to \exType{\kappa_2}{\nu_2}}{\nu_1}, \calC_1$}
            \AxiomC{$\Gamma \vdash N : \exType{\kappa_3}{\nu_3}, \calC_2$}
            \AxiomC{$\nu = \mathit{fresh()}$}
            \TrinaryInfC{$\Gamma \vdash M N : \exType{\kappa_2}{\nu}, \calC_1 \cup \calC_2 \cup \{\kappa_3 \sqsubseteq \kappa_1, \nu \equiv \nu_1 \times \nu_2 \times \nu_3\}$ }
        \end{prooftree}
    \end{minipage}

    \begin{prooftree}
        \AxiomC{$\Gamma \vdash M : \exType{\nu_1}{\nu_2}, \calC_M$}
        \AxiomC{$\Gamma \vdash N :\exType{\kappa_1}{\nu_3}, \calC_N$}
        \AxiomC{$\Gamma \vdash P : \exType{\kappa_2}{\nu_4}, \calC_P$}
        \AxiomC{$\kappa = \mathit{fresh}(\mathit{base}(\kappa_1))$}
        \AxiomC{$\nu, \nu' = \mathit{fresh}()$}
        \QuinaryInfC{$\Gamma \vdash \ifElse M N P : \exType{\kappa}{\nu}, \calC_M \cup \calC_N \cup \calC_P \cup \{\kappa_1 \sqsubseteq \kappa, \kappa_2 \sqsubseteq \kappa, \nu \equiv \nu_2 \times \nu', \nu_3 \sqsubseteq  \nu', \nu_4 \sqsubseteq \nu' \} $}
    \end{prooftree}

    \begin{prooftree}
        \AxiomC{$\Gamma \vdash M_1 : \exType{\nu_1}{\nu'_1}, \calC_1$}
        \AxiomC{$\cdots$}
        \AxiomC{$\Gamma \vdash M_{|f|} : \exType{\nu_{|f|}}{\nu'_{|f|}}, \calC_{|f|}$}
        \AxiomC{$\nu, \nu' = \mathit{fresh}()$}
        \QuaternaryInfC{$\Gamma \vdash f(M_1, \dots, M_{|f|}) : \exType{\nu}{\nu'}, \bigcup_{i=1}^{|f|} \calC_i \cup \{\nu \equiv \lit f(\nu_1, \dots, \nu_{|f|}), \nu' \equiv \prod_{i=1}^{|f|} \nu'_i\} $}
    \end{prooftree}
    \caption{Symbolic weight-aware type system.} \label{fig:symbolicTypeSystem}
\end{figure*}

\paragraph{From symbolic to concrete types}

An assignment $A$ is a mapping from interval variables to concrete intervals.
Given a symbolic type $\kappa$, we define the interval type $\kappa^A$ by replacing every interval variable in $\kappa$ with the concrete interval assigned to it in $A$.
For a weighted symbolic type $\mathscr{A}$, we define $\mathscr{A}^A$ in the same way.
Given a set of constraints $\calC$, we say that $A$ satisfies $\calC$, written $A \models \calC$ if all constraints in $\calC$ are satisfied (defined in the obvious way).

\begin{theorem}
    If $\vdash M : \mathscr{A}, \calC$ and $A \models \calC$ then $\vdash M : \mathscr{A}^A$.
\end{theorem}
\begin{proof}
This can be shown by induction on the structure of the term, which also determines the symbolic typing derivation.
From this, we obtain a valid interval typing derivation by replacing interval variables in the derivation with the concrete intervals assigned to them in $A$, and by applying the (\textsc{Sub}) rule in places where subtyping constraints are introduced.
\end{proof}

This theorem states that solutions to our constraints directly give us valid judgments in our interval type system, which allows us to invoke \cref{thm:staticRes}.

\paragraph{Solving Constraints}

To solve the resulting constraints, we employ known techniques from abstract interpretation \cite{CousotC77}.
Again, note that due to the simplicity of our constraints, our approach avoids expensive calls to a theorem prover.
When solving a set of constraints $\calC$, we are interested in the smallest solution, i.e.~an assignment $A$ with $A \models \calC$ where the intervals in $\calA$ are smallest possible w.r.t.~$\sqsubseteq$.

\paragraph{Na\"{i}ve algorithm}

A na\"{i}ve attempt to find a satisfying assignment for a set of constraints would be to iterate over the constraints and to extend the current assignment (initially chosen to map all elements to the bottom element $\bot$ in the interval domain, i.e.~an empty interval) whenever needed.
For example, if $\nu_i \sqsubseteq \nu_j$ is not satisfied by assignment $A$, we can update $A$ by mapping $\nu_j$ to $A(\nu_j) \sqcup A(\nu_i)$.
As is well known from abstract interpretation, this na\"{i}ve approach may not terminate because the interval domain is not chain-complete.
For instance, consider the constraints $\calC = \{\nu_1 = \myint{0, 0}, \nu_2 = \myint{1, 1}, \nu_1 \sqsubseteq \nu_3, \nu_3 \equiv \nu_3 + \nu_2 \}$.
The minimal solution is $\{\nu_1 \mapsto \myint{0,0}, \nu_2 \mapsto \myint{1, 1}, \nu_3 \mapsto \myint{0, \infty}\}$, but the algorithm never terminates and instead assigns the ascending chain $\myint{0,0}, \myint{0, 1}, \myint{0, 2}, \dots$ to $\nu_3$.

\paragraph{Widening}

To remedy the above problem, we use \defn{widening}, a standard approach to ensure termination of abstract interpretation on domains with infinite chains \cite{CousotC77}.
Let $\nabla$ be a \emph{widening operator} for intervals.
This means that $I_1 \sqcup I_2 \sqsubseteq I_1 \nabla I_2$ for all intervals $I_1, I_2$ and for every chain $I_0 \sqsubseteq I_1 \sqsubseteq I_2 \sqsubseteq \cdots$, the chain $(I^\nabla_i)_{i \in \mathbb{N}}$ defined by $I^\nabla_0 := I_0$ and $I^\nabla_i := I^\nabla_{i-1} \nabla I_i$ for $i \geq 1$ stabilises eventually.
A trivial widening operator is given by only allowing intervals to extend to infinity, defined as follows:
\begin{align*}
    \bot \nabla I &:= I \nabla \bot := I&&\\
    \myint{a, b} \nabla \myint{c, d} &:= \myint{a, b} \quad &&\text{if } a \leq c \land d \leq b\\
    \myint{a, b} \nabla \myint{c, d} &:= \myint{a, \infty} &&\text{if } a \leq c \land d > b\\
    \myint{a, b} \nabla \myint{c, d} &:= \myint{-\infty, b} &&\text{if } d \leq b \land c < a\\
    \myint{a, b} \nabla \myint{c, d} &:= \myint{-\infty, \infty} &&\text{if } c < a \land d > b \\
\end{align*}
As soon as the upper bound increases or lower bound decreases, the bound is directly set to $=\infty$ or $-\infty$ respectively.
By using the widening operator in each update step of our na\"ive algorithm, we break infinite increasing chains and the resulting algorithm is guaranteed to converge to a satisfying assignment (if one exists).

GuBPI solves constraints by using a standard worklist algorithm \cite{CousotC02, FechtS99}, combined with the previous widening operator.


\section{Supplementary Material for \texorpdfstring{\cref{sec:6linear}}{Section~6}}
\label{app:sec6}

\subsection{Extensions to Linear Splitting}
\label{sec:ext-lin-splitting}

In this section we give additional information about how our linear optimisation of the interval-trace semantics can be extended to non-uniform samples and non-linear scoring values.

\paragraph{Beyond uniform samples}

To allow for non-uniform samples, we can combine the standard interval trace semantics with the linear optimisation.
That is, in addition to bounding linear score functions, we also split and bound each non-uniform sample (as in the standard interval trace semantics).
Suppose that $\alpha_i$ is sampled from some continuous distribution $D$.
We then split the real line into chunks (the size and number of which is selected by a heuristic depending on $D$).
For each such chunk $[a, b]$, we compute the volume and multiply by the lower and upper bounds of the pdf of $D$ on $[a, b]$.
In this way, we can even approximate integrals where not all variables are sampled from a uniform distribution (without needing to resort to inverse cumulative distribution functions).
Our experiments show, that as long as the guards are still linear, this approach is advantageous compared to the naive interval-based semantics.

\paragraph{Beyond linear functions}

While guards on conditionals are often linear, this is rarely the case for score expressions (as one usually observes values from some non-uniform distribution with a non-linear pdf).
Consider the pedestrian example again.
While all guards are linear, the score expression has the form $\pdf_{\Normal(1.1, 0.1)}(\calV)$.
We can handle such \emph{non-linear} score values by applying linear optimisation to the linear subexpressions and interval arithmetic for the nonlinear parts.
Formally, we assume that each $\calW_i \in \Xi$ (each symbolic value we score with) has the form $\calW_i = f_i(\calZ_i^1, \dots, \calZ_i^{m_i})$ where $\calZ_i^j$ for $1 \leq j \leq m_i$ denote \emph{linear} functions of the sample variables and $f : \real^{m_i} \to \real$ is a possibly non-linear function.
Every score expression can be written in this way.
For instance, in the pedestrian example, we have $f = \pdf_{\Normal(1.1, 0.1)}$.

Let $\Xi' = \{\calZ_i^j \mid i \in \{1, \dots, n\}, j \in \{1, \dots, m_i\} \}$ be the set of all such linear functions.
We bound each linear function in $\Xi'$ using linear optimisation as before.
We obtain a box $b$ (which now has dimension $|\Xi'|$ instead of $|\Xi|$) and define the weight $\mathit{weight}(b)$ by applying the interval liftings $f_i^\iv$ of the non-linear functions $f_i$ to the bounds for each argument.
Formally, $\mathit{weight}(b) = \prod_{i=1}^{k} f_i^\iv (b_i^1, \dots, b_i^{m_i})$ where $b_i^j$ is the interval bound on $\calZ_i^j$.
Note that this strictly generalizes the approach outlined before since we can choose $f_i$ as the identity if $\calW_i$ is already linear.
The definition of $\mathit{approx}(b)$ with the new weight definition still satisfies \cref{prop:linearSplitting}.
This way, we can even approximate integrals over non-linear functions by means of simple volume computations.
As our experiments (e.g.~on the pedestrian example) show, this approximation is precise enough to obtain useful bounds on the posterior.
It is important to note that while we can deal with non-linear score values, we cannot handle non-linear guards and instead use the standard semantics for such cases.


\section{Supplementary Material for \texorpdfstring{\cref{sec:7practical-evaluation}}{Section~7}}
\label{app:sec7}

Our experiments were performed on a server running Ubuntu 18.04 with an 8core Intel(R) Xeon(R) CPU E3-1271 v3 @ 3.60GHz CPU with 32Gbp of RAM.
The current version of GuBPI is not parallelised and makes no use of the additional cores.
The running times on a Macbook Pro with Apple M1 were comparable, and sometimes even faster.

\subsection{Pyro's HMC samples for the pedestrian example}

The HMC samples plotted in \cref{fig:pedestrian-stochastic,fig:pedBounds} were generated with the probabilistic programming system Pyro \cite{pyro19}.
Since the original pedestrian program has infinite expected running time, we introduced a stopping condition in the random walk: if the distance traveled exceeded 10, the loop was exited.
(This has a negligible effect on the posterior distribution because the weight of such a trace is at most $\pdf_{\Normal(1.1, 0.1)}(10) < 10^{-1700}$.)

We used Pyro's HMC sampler to compute 10 Markov chains with 1000 samples each for this program.
We set the hyperparamaters to 0.1 for the step size and 50 for the number of steps.
We also tried the NUTS sampler, which aims to automatically estimate good values for the hyperparameters, but it performed worse than the manually chosen values.
The running time for the chains varied significantly: some took around one minute, others almost an hour.
This depended on whether the Markov chain got ``stuck'' in a long trace.
(The length of the traces varied between 2 and about 200.)

We discarded chains with very low acceptance rates (under 1\%), aggregated the remaining chains, which had acceptance rates of over 50\%, and used their histogram in \cref{fig:pedestrian-stochastic,fig:pedBounds}.

\begin{table*}[!t]
	\caption{Evaluation on selected benchmarks from \cite{SankaranarayananCG13}. We give the times (in seconds) and bounds computed by \cite{SankaranarayananCG13} and GuBPI. The table agrees with \Cref{fig:pldi13results} but spells out the full problem name and the exact query.}\label{fig:pldi13resultsApp}

	\vspace{-3mm}

	\setlength\dashlinedash{0.5pt}
	\setlength\dashlinegap{1pt}
	\setlength\arrayrulewidth{0.5pt}
	\centering
	\small
	\def\arraystretch{1.2}
	\begin{tabular}{lcccccc}
		\toprule
		\multicolumn{3}{@{}c@{\hspace{0mm}}}{} &\multicolumn{2}{@{}c@{}}{\textbf{Tool from \cite{SankaranarayananCG13}}} & \multicolumn{2}{@{}c@{\hspace{0mm}}}{\textbf{GuBPI}}\\
		\cmidrule[0.7pt](l{1mm}r{1mm}){4-5}
		\cmidrule[0.7pt](l{1mm}){6-7}
		\textbf{Program}& \textbf{Q} & \textbf{Query}& $\boldsymbol{t}$ & \textbf{Result} & $\boldsymbol{t}$ & \textbf{Result} \\
		\cmidrule[0.7pt](r{1mm}){1-1}
		\cmidrule[0.7pt](l{1mm}r{1mm}){2-2}
		\cmidrule[0.7pt](l{1mm}r{1mm}){3-3}
		\cmidrule[0.7pt](l{1mm}r{1mm}){4-4}
		\cmidrule[0.7pt](l{1mm}r{1mm}){5-5}
		\cmidrule[0.7pt](l{1mm}r{1mm}){6-6}
		\cmidrule[0.7pt](l{1mm}r{0mm}){7-7}
		tug-of-war & Q1 & $\mathit{total}\_a\_b < \mathit{total}\_t\_s$ & 1.29  & $[0.6126,   0.6227]$  & 0.72  & $[0.6134, 0.6135]$  \\
		\hdashline
		tug-of-war & Q2 & $\mathit{total}\_a\_s < \mathit{total}\_b\_t$ & 1.09  & $[0.5973,   0.6266]$  & 0.79  & $[0.6134, 0.6135]$  \\
		\hdashline
		beauquier-etal-3 & Q1 & $\mathit{count} < 1$ & 1.15 & $[0.5000,   0.5261]$ & 22.5 & $[0.4999, 0.5001]$  \\
		\hdashline
		example-book-simple & Q1 & $\mathit{count} \geq 2$ & 8.48 & $[0.6633,   0.7234]$ & 6.52 & $[0.7417, 0.7418]$   \\
		\hdashline
		example-book-simple & Q2$^\star$  & $\mathit{count} \geq 4$  & 10.3 & $[0.3365,   0.3848]$ & 8.01 & $[0.4137, 0.4138]$   \\
		\hdashline
		example-cart & Q1 & $\mathit{count} \geq 1$  & 2.41 & $[0.8980,   1.1573]$ & 67.3 & $[0.9999, 1.0001]$   \\
		\hdashline
		example-cart & Q2 & $\mathit{count} \geq 2$  & 2.40 & $[0.8897,   1.1573]$ & 68.5 & $[0.9999, 1.0001]$  \\
		\hdashline
		example-cart & Q3 & $\mathit{count} \geq 4$  & 0.15 & $[0.0000,   0.1150]$ & 67.4 & $[0.0000, 0.0001]$   \\
		\hdashline
		example-ckd-epi-simple & Q1$^\star$  & $f_1 \leq 4.4 \land f \geq 4.6$  & 0.15 & $[0.5515,   0.5632]$ & 0.86 & $[0.0003, 0.0004]$   \\
		\hdashline
		example-ckd-epi-simple & Q2$^\star$  & $f_1 \geq 4.6 \land f \leq 4.4$ & 0.08 & $[0.3019,   0.3149]$ & 0.84 & $[0.0003, 0.0004]$   \\
		\hdashline
		example-fig6 & Q1 & $c \leq 1$ & 1.31 & $[0.1619,   0.7956]$ & 21.2 & $[0.1899, 0.1903]$   \\
		\hdashline
		example-fig6 & Q2 & $c \leq 2$  & 1.80 & $[0.2916,   1.0571]$ & 21.4 & $[0.3705, 0.3720]$   \\
		\hdashline
		example-fig6 & Q3 & $c \leq 5$  & 1.51 & $[0.4314,   2.0155]$ & 24.7 & $[0.7438, 0.7668]$   \\
		\hdashline
		example-fig6 & Q4 & $c \leq 8$  & 3.96 & $[0.4400,   3.0956]$ & 27.4 & $[0.8682, 0.9666]$   \\
		\hdashline
		example-fig7 & Q1 & $x \leq 1000$  &  0.04& $[0.9921, 1.0000]$ & 0.18 & $[0.9980, 0.9981]$  \\
		\hdashline
		example4 & Q1  & $x + y > 10$ & 0.02 & $[0.1910,   0.1966]$ & 0.31 & $[0.1918, 0.1919]$   \\
		\hdashline
		example5 & Q1 & $x + y > z + 10$ & 0.06 & $[0.4478,   0.4708]$ & 0.27 & $[0.4540, 0.4541]$   \\
		\hdashline
		herman-3 & Q1 & $\mathit{count} < 1$ & 0.47 & $[0.3750,   0.4091]$ & 124 & $[0.3749, 0.3751]$ \\
		\bottomrule
	\end{tabular}

\end{table*}

\subsection{Details on Probability Estimation}

In \cref{fig:pldi13results} (results of our tool for the probability estimation benchmark), we omitted the query for space reasons.
Complete information including the query can be found in \cref{fig:pldi13resultsApp}.

\subsection{Simulation-based Calibration}
\label{sec:sbc-experiments}

We implemented SBC for both likelihood-weighted importance sampling and Pyro's HMC.
As hyperparameters for SBC, we picked $L=63$ samples per simulation (following the suggestion in \cite{SBC} to take one less than a power of two) and $N$ simulations with $N = 10L$ (also following the paper's suggestion).
Note that the number of samples is much less than the 10000 samples used for \cref{fig:pedestrian-stochastic,fig:pedBounds} (10 chains with 1000 samples each).
But setting $L=1000$ would be at least 100 times slower because $N$ has to increase proportionally to $L$.
Also note that for Pyro, we ran HMC with $L$ warmup steps before generating $L$ samples.
Both importance samples and HMC samples exhibited significant autocorrelation.
As suggested in \cite{SBC}, we applied thinning to reduce its effect, choosing a thinning factor of around $\frac{L}{L_\mathit{eff}}$ where $L_\mathit{eff}$ is the effective sample size.

\begin{figure*}
	\begin{subfigure}[t]{0.24\textwidth}
		\centering
		\includegraphics[width=\textwidth]{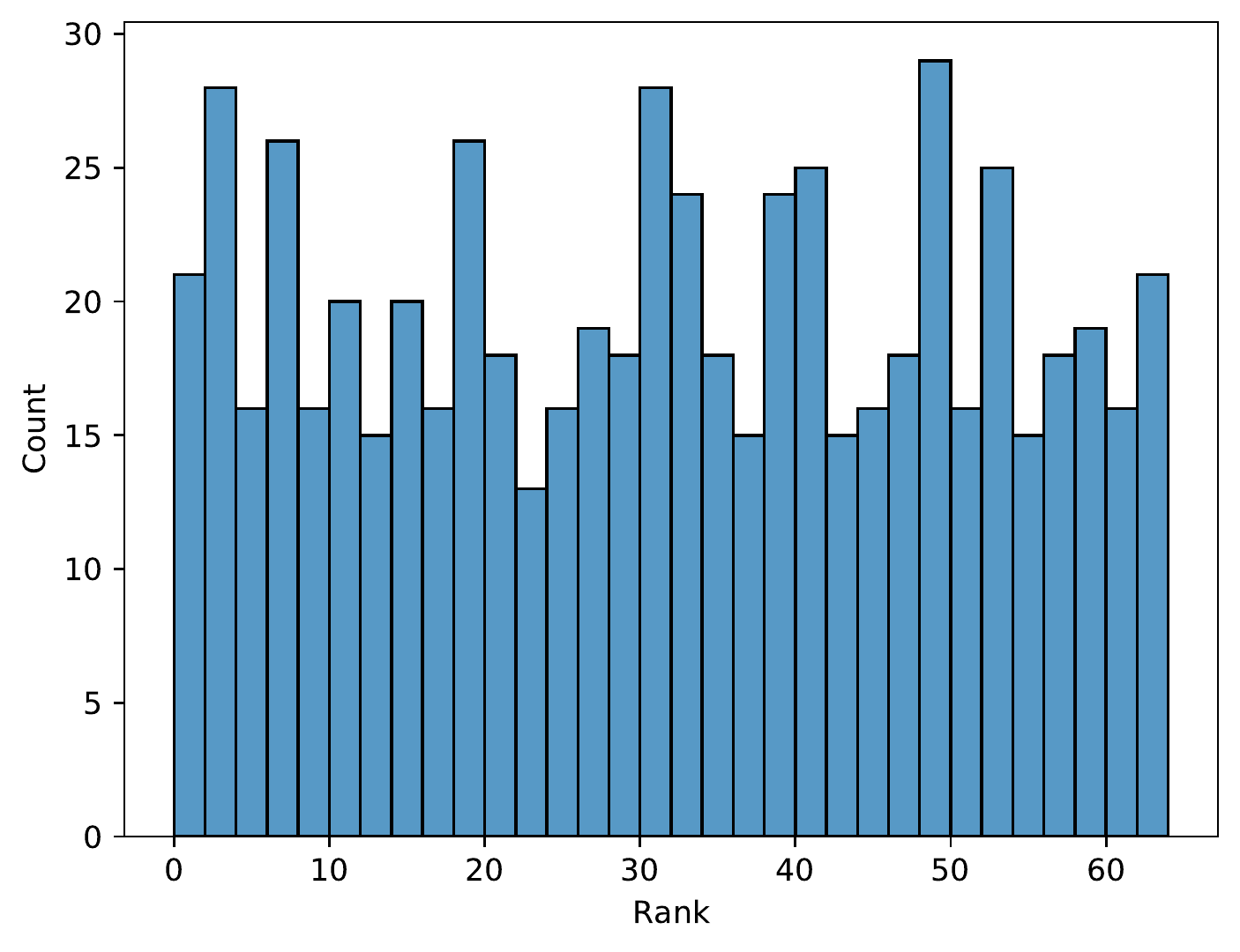}
		\subcaption{Importance sampling for the pedestrian example, thinning factor 100.}
		\label{fig:sbc-pedestrian-is}
	\end{subfigure}\hfill%
	\begin{subfigure}[t]{0.24\textwidth}
		\centering
		\includegraphics[width=\textwidth]{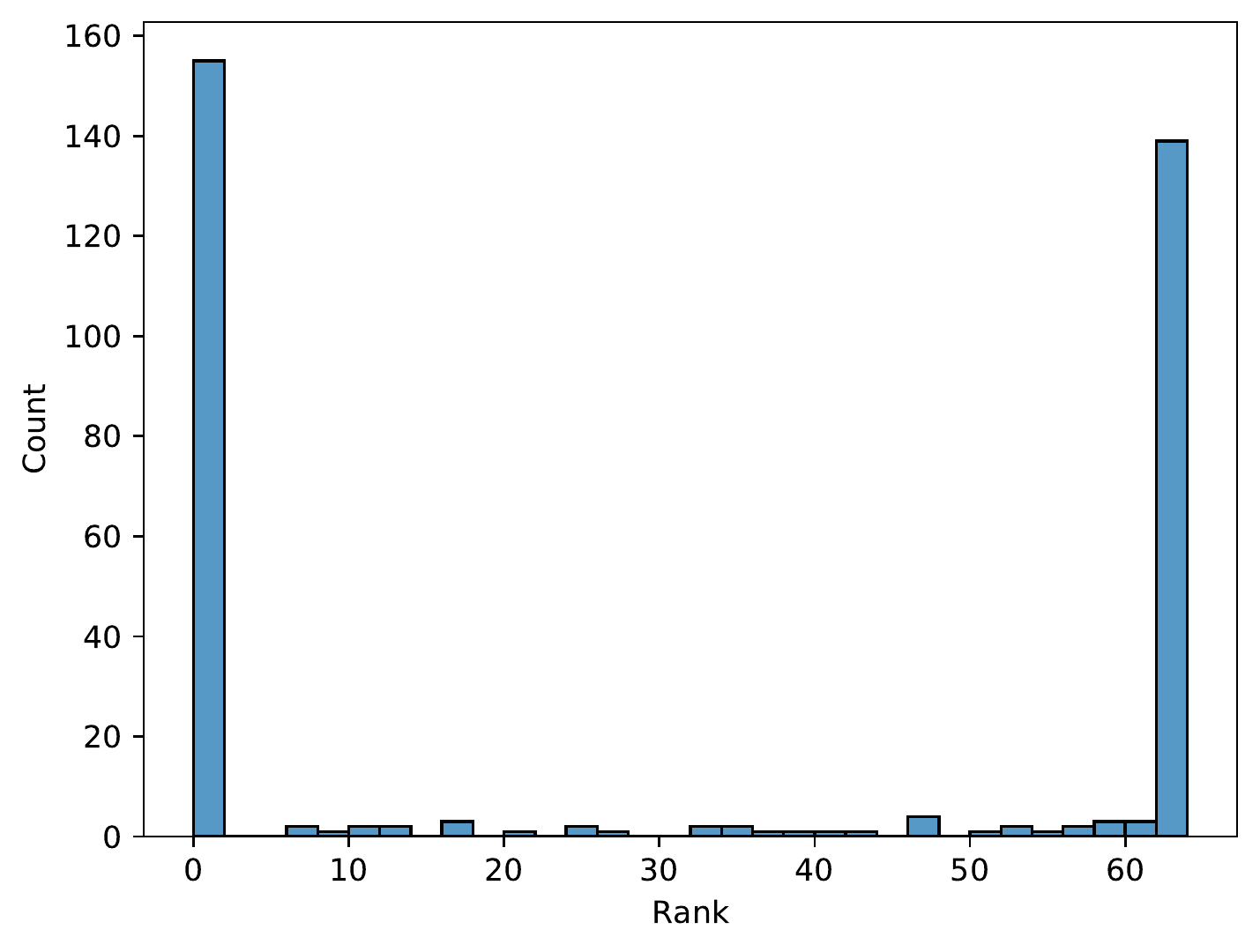}
		\subcaption{HMC for the pedestrian example, no thinning.
		(Only 332 simulations were used because the rest had too low acceptance rates.)}
		\label{fig:sbc-pedestrian-hmc}
	\end{subfigure}\hfill%
	\begin{subfigure}[t]{0.24\textwidth}
		\centering
		\includegraphics[width=\textwidth]{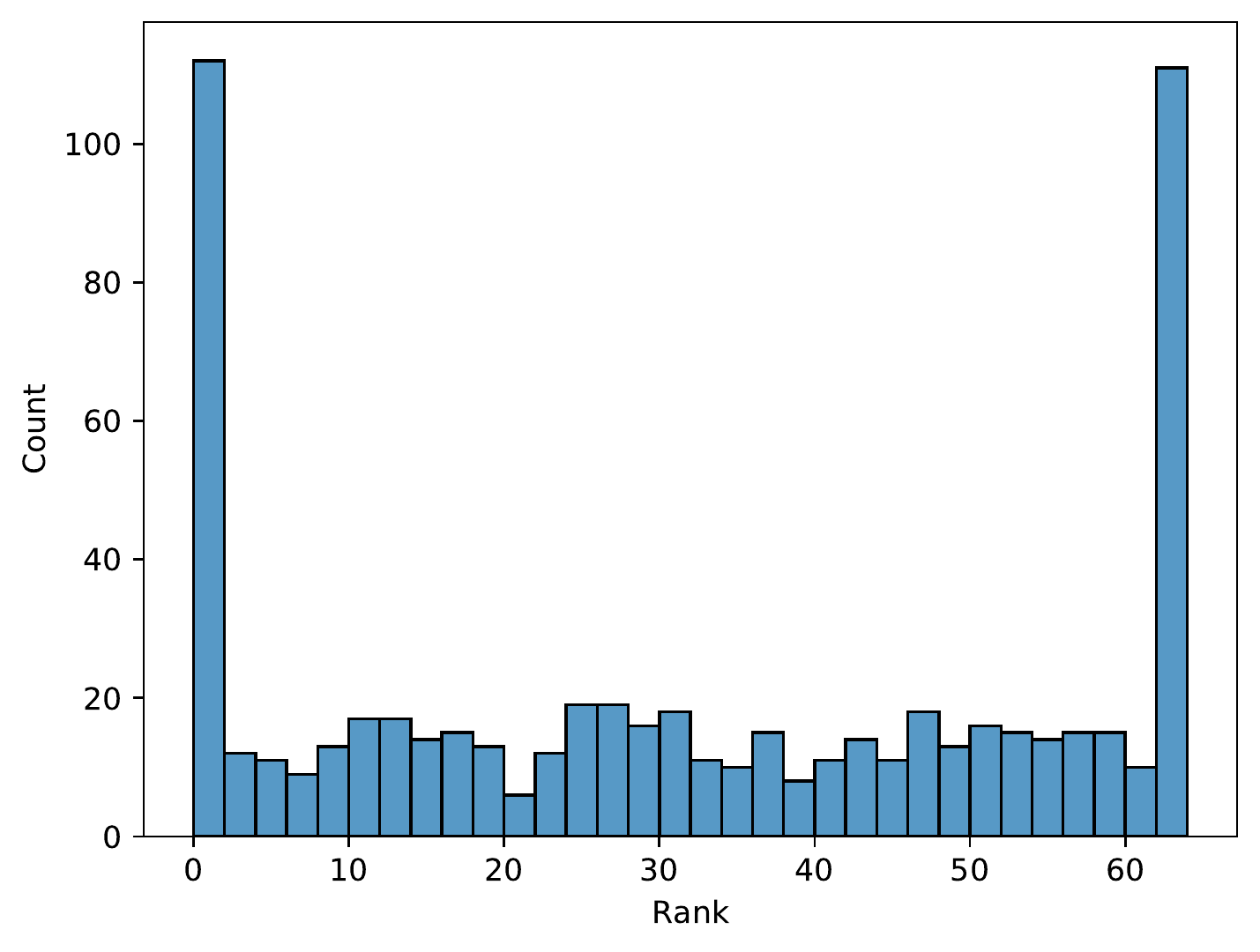}
		\subcaption{NUTS for the binary Gaussian mixture model, thinning factor 10.}
		\label{fig:sbc-bgmm-nuts}
	\end{subfigure}\hfill%
	\begin{subfigure}[t]{0.24\textwidth}
		\centering
		\includegraphics[width=\textwidth]{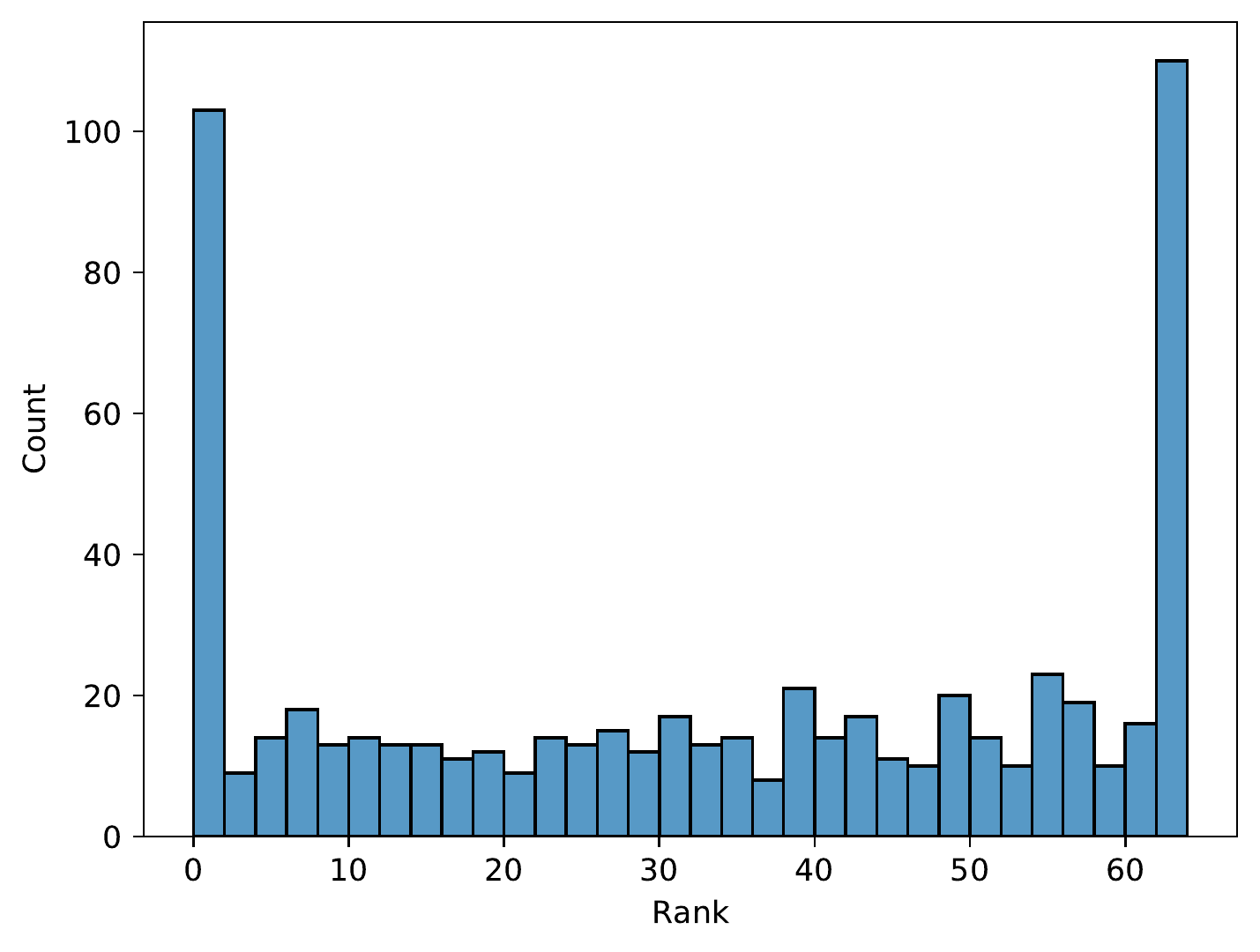}
		\subcaption{NUTS for the \emph{two-dimensional} binary Gaussian mixture model, thinning factor 10.}
		\label{fig:sbc-bgmm2d-nuts}
	\end{subfigure}
	\caption{Simulation-based calibration: rank histogram plots (630 simulations with 63 samples each).}
\end{figure*}

\paragraph{Pedestrian example}

For importance sampling, the rank histogram looks fairly uniform (\cref{fig:sbc-pedestrian-is}), which means that SBC does not detect an issue with the inference and thus increases the confidence in the validity of the importance samples.
For Pyro's HMC, simulation-based calibration is very slow: the rank histogram (\cref{fig:sbc-pedestrian-hmc}) took 32 hours (!) to produce.
(Note that only 332 simulations could be used, the rest were discarded because the acceptance rate was too low.)
The spikes at the boundary indicate that the samples have high autocorrelation and in fact, the effective sample size $L_\mathit{eff}$ was at most $\frac{L}{10}$, often much lower (depending on the chain).
The suggestion in \cite{SBC} is thus to apply thinning, with a factor of $\frac{L}{L_\mathit{eff}}$, which is at least 10 in our case.
This would increase the running time of SBC by the same factor, to at least 300 hours, but probably 600 or more.
We did not consider it a good use of resources to carry out this experiment.

\paragraph{Binary Gaussian Mixture Model}

We also considered the binary GMM (\cref{fig:gmm}) and a two-dimensional version of the same model.
The spikes at the boundary could again be a sign of high autocorrelation, but in this case, we already applied thinning with a factor of 10 (again based on the effective sample size).
Instead, as discussed in \cite{SBC}, this symmetric U-shape indicates that the computed data-averaged posterior is underdispersed relative to the prior distribution.
This interpretation is consistent with our knowledge about the model: HMC only finds one mode in the posterior distribution and misses the others.
Hence SBC successfully detects the issue, and in the case of the higher-dimensional model, it does so in less time than GuBPI (cf.~\cref{tab:SBC}).
For the other models, including the pedestrian example, GuBPI is faster.

\clearpage
\tableofcontents
\fi

\end{document}